\documentclass[11pt]{article}

\usepackage{graphicx} 
\usepackage{amsmath, amsfonts, amssymb, amsbsy, amsthm, mathtools} 
\usepackage{bbm, bm} 
\usepackage{xcolor} 
\usepackage{natbib} 
\usepackage{url} 
\usepackage{lmodern} 
\usepackage[T1]{fontenc}
\usepackage{mathrsfs} 
\usepackage{tikz} 
\usetikzlibrary{positioning} 

\usepackage[a4paper]{geometry} 
\usepackage{a4wide} 

\usepackage{booktabs} 
\usepackage{subcaption}
\usepackage{float} 
\usepackage{subdepth} 
\usepackage{microtype} 
\usepackage{array}
\usepackage{multirow}
\usepackage{enumitem}
\usepackage{algpseudocode}
\usepackage{algorithm} 
\usepackage{float}

\newcommand{\cbr}[1]{\left\{ {#1} \right\}}
\newcommand{\rbr}[1]{\left( {#1} \right)}
\newcommand{\sbr}[1]{\left[ {#1} \right]}

\newtheorem{thm}{Theorem}[section]

\newtheorem{lem}[thm]{Lemma}

\newtheorem{prop}[thm]{Proposition}
\theoremstyle{definition}
\newtheorem{defn}[thm]{Definition}
\newtheorem{rem}[thm]{Remark}

\newcommand{\RR}{\mathbb{R}}
\newcommand{\vc}[1]{\bm{#1}}
\newcommand{\siid}{\overset{\textnormal{iid}}{\sim}}
\newcommand{\diff}{\, \textnormal{d}}
\newcommand{\GP}{\textnormal{GP}}
\newcommand{\RV}{\textnormal{RV}}
\newcommand{\Exp}{\textnormal{Exp}}

\newcommand{\PP}{\mathbb{P}}
\newcommand{\EE}{\mathbb{E}}

\newcommand{\vZ}{\vc{Z}}
\newcommand{\vY}{\vc{Y}}

\newcommand{\vthe}{\vc{\theta}}

\newcommand{\vS}{\vc{S}}
\newcommand{\vV}{\vc{V}}
\newcommand{\vX}{\vc{X}}
\newcommand{\vz}{\vc{z}}
\newcommand{\vx}{\vc{x}}

\newcommand{\1}{\mathbbm{1}}

\newcommand{\indep}{\perp\!\!\!\!\perp}

\newcounter{commentcounter}

\newcounter{mircocounter}

\newcounter{questioncounter}

\allowdisplaybreaks 
\numberwithin{equation}{section} 
\graphicspath{Mirco/figures/} 

\usepackage{hyperref}

\begin{document}

\title{A sub-asymptotic model for bivariate threshold exceedances}
\author{Mirco Lescart\thanks{LIDAM/ISBA, UCLouvain, Voie du Roman Pays 20, 1348 Louvain-la-Neuve, Belgium. E-mails: mirco.lescart@uclouvain.be, anna.kiriliouk@uclouvain.be}
\and Anna Kiriliouk\footnotemark[1]
\and Philippe Naveau\thanks{Laboratoire des Sciences du Climat et de l’Environnement (IPSL, CNRS, CEA, UVSQ). E-mail: naveau@lsce.ipsl.fr}}

\maketitle

\begin{abstract}
Extreme value theory offers a statistical framework for quantifying the risk of rare events, with the generalized Pareto (GP) distribution providing the canonical limit model for univariate threshold exceedances. In many applications, however, extremes are intrinsically multivariate, requiring models that capture both marginal tail behaviours and joint extremal dependencies. 
Under asymptotic dependence, the multivariate GP distribution represents a suitable modelling family, but when asymptotic independence arises, sub-asymptotic models are needed.
In this work, we propose and study a flexible sub-asymptotic parametric class  to model  bivariate threshold exceedances.
Our new model accommodates a broad range of tail dependence behaviours  and contains the standardised multivariate GP distribution as a limiting case while retaining margins that converge to univariate GP tails. Our formulation allows extremal dependence to evolve naturally with the marginal parameters on the original data scale, facilitating direct computation and interpretation of failure probabilities. Model inference is done via a likelihood-free neural Bayes estimation approach, with tailored prior specifications. An extensive simulation study and an application to Belgian rainfall extremes illustrate the estimation framework and the flexibility of the model.
\smallskip

\end{abstract}
	

    \section{Introduction}

\emph{Extreme value theory (EVT)} provides a mathematical framework to statistically model the probability distributions of extreme (rare) events \citep{beirlant_statistics_2004}. A key model within EVT is the univariate \emph{generalized Pareto (GP)} distribution, which characterizes exceedances over a high threshold through a known and explicit parametric form. 
Under mild conditions, any non-degenerate limit distribution of (suitably standardised) threshold exceedances must be GP, just as the Gaussian distribution arises as the universal limit of (suitably standardised) sample means. The GP distribution underpins the peaks-over-threshold approach \citep{davison1990models}, widely applied in hydrology, finance, and engineering to quantify risks associated with rare events such as floods, financial crashes, or structural failures \citep{embrechts2013modelling,bousquet2021extreme}.

    In practice, many extreme phenomena are inherently multivariate: flooding may result from concurrent rainfall across several catchments, financial crises arise from assets crashing jointly, and infrastructure failures may stem from multiple stresses. Such examples highlight the need for models that capture both the marginal tail behaviour of individual components and their joint extremal dependence. 
    
    To extend the univariate GP model, the family of \emph{multivariate generalized Pareto (MGP)} distributions was proposed as the limiting law of (suitably standardised) multivariate threshold exceedances \citep{rootzen2006, rootzen_multivariate_2018-1, kiriliouk_peaks_2019}. In this setting, an exceedance is a data point for which \emph{at least one} of the components lies above a high threshold; it is common to say that the MGP distribution is supported on an \emph{L-shaped region} (see Figure~\ref{fig:subasymp_datasets}). Like its univariate counterpart, the MGP distribution is \emph{threshold-stable}, 
    in the sense that exceedances above higher thresholds remain MGP, preserving the marginal tail behaviour and extremal dependence structure. For \emph{asymptotically dependent (AD)} variables, for which extremes tend to occur jointly, threshold-stability is an additional restriction. Under \emph{asymptotic independence (AI)}, threshold-stability cannot arise.
    In practice, such cases are frequent \cite[see, e.g.][]{huser2025modeling} and a more refined knowledge of dependence   is needed.
 To illustrate this, Figure~\ref{fig:subasymp_datasets} displays three scatterplots of simulated bivariate AI samples.
 Visually, the dependencies among extremes seem to change from weak (left panel), to moderate (centre panel), and finally  strong (right panel). 
 However, by construction, all three samples exhibit asymptotic independence with the same level of \emph{residual} tail dependence. 
 To move beyond describing extremal dependence solely through asymptotic limit behaviour, a recent strategy is to also model intermediate to high quantiles. This approach allows for flexible tail decay and avoids the need to commit a priori to either AI or AD.
 
    
 Recent contributions to the development of sub-asymptotic models include unified frameworks for AI and AD modelling of bivariate threshold exceedances \citep{wadsworth2017modelling, engelke_extremal_2019}, as well as spatial constructions that accommodate distance-varying dependence regimes \citep{huser_bridging_2017, krupskii2018factor, huser_modeling_2019, morris2017space, hazra2024efficient, shi2024spatial}. Conditional extreme-value models \citep{heffernan2004conditional, wadsworth2022higher} offer a complementary approach, conditioning on a single variable exceeding a high threshold to describe joint tail behaviour. 
\citet{bacro2025multivariate} also proposed a  framework for constructing multivariate distributions with GP margins, exploiting the representation of a univariate GP variable as a ratio of exponential and gamma random variables. While their approach provides a tractable link between marginal and joint tail behaviour and can reach both AI and AD in the limit, its ability to model intermediate to strong AD is limited (see Section~\ref{sec:GMGP}). Also, sub-asymptotic models should ideally contain commonly used asymptotic models as a boundary or limiting case \citep[see also][]{zhang2025}, ensuring coherence with asymptotic theory.

In this work, we introduce a flexible sub-asymptotic bivariate model that extends the multivariate GP representation through sums of exponential and gamma components. 
Our approach specifically targets the L-shaped region, enabling a coherent description of both moderate and extreme co-occurrences, and converges to the (standardised) MGP distribution as a limiting case. 
Although our model shares some aspects with \citet{bacro2025multivariate}, it provides a much wider range of extremal dependence.  


Unlike many existing approaches, our framework does not require transforming margins to a common scale prior to modelling dependence. While probability integral transforms may be convenient in iid settings, they typically require explicit modelling of covariate-dependent marginal distributions in non-stationary applications, which can complicate interpretation and lead to uncertainty propagation \citep{kakampakou2024spatial}. Working on the original scale allows different variables to retain distinct Pareto tails, as illustrated in our analysis of extreme precipitation over Belgium (Section~\ref{sec:application}). In addition, it facilitates the interpretation and computation of failure probabilities (i.e., probabilities of the data falling in a specified extreme set) in terms of the observed variables \citep{kiriliouk2022estimating}.
 

 Concerning the estimation of our model  parameters, we avoid    
    using classical methods such as censored likelihood \citep{smith1997markov, wadsworth2014efficient,bacro2025multivariate}, or weighted least squares \citep{einmahl2012m, einmahl2018continuous}. Instead, we leverage  the parameter parsimony  and the simplicity of simulating random draws  from our model. This allows us  to 
    adapt    neural Bayes estimation techniques  \citep{sainsbury-dale_likelihood-free_2024, richards_neural_2024, andre_neural_2025}. 
    This approach has become  popular for models with intractable likelihood; see also \citet{rodder2025theoretical} for a characterization of its consistency and efficiency.
    
    The remainder of the paper is organized as follows. Section~\ref{sec:background} reviews univariate and bivariate GP models, asymptotic tail dependence coefficients and the main ingredients of the model of  \cite{bacro2025multivariate}. The sub-asymptotic bivariate GP distribution is introduced in Section~\ref{sec:ourmodel}, where we derive   marginal distributions, tail dependence properties and sub-asymptotic features of the model. Section~\ref{sec:estimation} describes the estimation procedure and Section~\ref{sec:simstudy} contains a simulation study. An  application to rainfall data in Belgium is provided in Section~\ref{sec:application}, and Section~\ref{sec:discussion} concludes. Finally, the main proofs are detailed in Section~\ref{app:proofs}. The Supplementary Material contains proofs of the lemmas stated in Section~\ref{app:proofs} and additional numerical results for the Belgian rainfall application.

      Throughout, bold symbols will refer to multivariate quantities,  $\vc{1}$ denotes a vector of ones (of a dimension clear from the context), and mathematical operations on vectors such as addition, multiplication and comparison are considered component-wise.

  

    \section{Background} \label{sec:background}


    \subsection{Univariate peaks-over-thresholds modelling}

    The univariate \emph{peaks-over-thresholds} approach \citep{davison1990models} provides the theoretical and practical foundation for modelling threshold exceedances via a \emph{generalized Pareto (GP)} distribution. 
    First, recall that the survival function of the GP distribution is given by
    \begin{equation*}
    	\overline{H}_{\xi, \sigma}(z)
    	= \begin{dcases}
    		\left(1+\frac{\xi z}{\sigma} \right)^{-1/\xi}, & \text{if } \xi \neq 0, \\
    		e^{-z/\sigma}, & \text{if } \xi=0,
    	\end{dcases}
    	\qquad z \geq 0, \,  1 + \frac{\xi z}{\sigma} > 0,
    \end{equation*}
    where $\xi \in \RR$ is the shape parameter or \emph{tail index} and $\sigma > 0$ is the scale parameter. If a random variable $Z$ has cdf $H_{\xi, \sigma}$, we also write $Z \sim \GP(\xi,\sigma)$. For $\xi > 0$, the survival function $\overline{H}_{\xi, \sigma}$ is regularly varying tail of order $1/ \xi$ (see Section~\ref{sec:marginal_proof} for a definition).
    
    Let $X$ be a random variable with cdf $F$ with upper endpoint $x^* = \infty$. 
    Suppose that there exists an auxiliary function $c$ such that, for any $x > 0$,
    \begin{equation}\label{eq:DA}
    	\lim_{u \rightarrow \infty} \PP \sbr{\frac{X - u}{c(u)} > x \mid X > u} = (1 + \xi x)^{-1/\xi}.
    \end{equation}
    Assumption \eqref{eq:DA} is equivalent to the fact that $F$ is in the max-domain of attraction of a generalized extreme-value distribution \citep{pickands1975statistical,balkema1974residual}. 
    Interpreting $c(u)$ as a scale parameter $\sigma > 0$, we get, for $x > u$ and $u$ sufficiently large, $\PP[X > x] \approx  \PP[X > u] \overline{H}_{\xi, \sigma}(x-u)$, which can be used to model the tail of $X$.
    
 The GP distribution is \emph{threshold-stable}: for any positive scalar $v$, the conditional random variable $Z - v \mid Z > v$ is again GP distributed with the same shape parameter $\xi$ and with scale $\sigma + \xi v$. Finally, it admits a convenient representation as the ratio of a standard exponential random variable over an independent gamma-distributed random variable.

    \begin{lem}[Exponential-gamma representation of the GP distribution] \label{lem:mixture_GPD}
    Any GP distributed random variable with a positive shape parameter $\xi$
    can be written as the ratio of two independent random variables, 
    a standard exponential random variable $E$ and a gamma random variable $G$, with shape parameter $\alpha$ and rate parameter $\beta$, i.e., with density
    $f_G(g) = \frac{\beta^\alpha}{\Gamma(\alpha)} g^{\alpha-1} e^{-\beta g}$ for $g>0$,
    $$
    \frac{E}{G} \sim  \GP(\xi= 1/\alpha, \sigma=\beta/\alpha).
    $$
    \end{lem}
   \noindent For a proof of Lemma~\ref{lem:mixture_GPD}, see, for example, \citet[page 157]{reiss1997statistical}. The shape parameter of the gamma distribution being positive, we cannot obtain GP distributions with negative tail indices from the representation in Lemma~\ref{lem:mixture_GPD}. In the remainder of this work, we focus on $\xi \geq 0$. Note that $\xi = 0$ is obtained by taking $\beta = \alpha c$ for some $c > 0$ and $\alpha \to \infty$, in which case $G$ converges to $1/c$ in probability (see the proof of Lemma~\ref{lem:mgpd_limit}), and $E/G$ is exponential with scale $c$.

    \subsection{Bivariate peaks-over-thresholds modelling}

  We recall some convenient (asymptotic) summary measures of extremal dependence before introducing the model in \citet{bacro2025multivariate} and our bivariate sub-asymptotic GP distribution.

    \subsubsection{Summary measures of extremes}\label{sec:summary_measures}
    
    Let $\vX=(X_1,X_2)^T$ be a bivariate random vector with joint cdf $F$ and continuous marginal distributions $F_1,F_2$. A measure of the upper tail dependence between $X_1$ and $X_2$ is
    \begin{equation}\label{eq:chi_curve}
    \chi(q) = \frac{\PP\left(F_1(X_1)>q,\;F_2(X_2)>q\right)}{1-q}, \qquad q \in [0,1).
    \end{equation}
    The \emph{(upper) tail dependence coefficient} $\chi = \lim_{q \uparrow 1} \chi(q)$ then quantifies to what extent the components of $\vX$ can be extreme simultaneously; we say that $X_1$ and $X_2$ exhibit \emph{asymptotic dependence (AD)} when $\chi>0$, and \emph{asymptotic independence (AI)} when $\chi=0$, see \citet{coles_dependence_1999}. The function $\chi(q)$ can be seen as a sub-asymptotic measure of tail dependence.
    
    When $\chi$ equals zero, the decay of joint extremes can be described by the \emph{residual tail dependence coefficient} $\eta = \lim_{q \uparrow 1} \eta(q)$, where
    \begin{equation*}
    \eta(q) = \frac{\log(1-q)}{\log \PP\left(F_1(X_1)>q,\,F_2(X_2)>q\right)}, \qquad q \in [0,1).
    \end{equation*}
    The residual tail dependence coefficient is derived from the \citet{ledford_statistics_1996} model assumption that
    \[
    \PP\left(F_1(X_1) >q,\,F_2(X_2)>q\right) = \ell(1-q) (1-q)^{1/\eta},
    \]
    where $\ell$ denotes a slowly varying function at zero, i.e.,
$\lim_{t \downarrow 0} \ell(ct)/\ell(t) = 1$ for all $c>0$. A value $\eta <1/2$ indicates negative residual dependence, meaning that the joint tail decays faster than under independence; $\eta =1/2$ means there is no residual tail dependence.  Values in the interval $(1/2,1)$ indicate  positive residual dependence, with joint tails decaying more slowly than in the independent case. Finally, the boundary case $\eta =1$ is the only one compatible with asymptotic dependence, where the probability of simultaneous extremes remains non-zero in the limit.
The coefficient $\eta(q)$ quantifies the amount of sub-asymptotic residual tail dependence.

    Evaluated for a full range of $q \in [0,1)$, the coefficients $\chi (q)$ and $\eta (q)$ are particularly informative when convergence to the asymptotic regime is slow.
They are related via
\begin{equation}\label{eq:chieta}
\eta (q)=\frac{\log(1-q)}{\log(1-q)+ \log \chi (q)}.
\end{equation}
Sometimes, we will add a subscript, writing $\chi_{\vX}$ or $\eta_{\vX}$, when confusion may arise. 
    
    \subsubsection{Bivariate generalized Pareto distributions}\label{sec:mgpd}

    The bivariate generalized Pareto distribution is the asymptotically justified model for the (conditional) threshold exceedances of a bivariate random vector as the threshold grows to infinity. We introduce it for a vector $\vX_E$ whose margins follow unit exponential distributions; a full treatment  can be found in \citet{rootzen_multivariate_2018} or \citet{rootzen2018nr2}. 
    
  Define $\mathcal L = \{\vc{x} \in \RR^2 : \max(\vx) > 0\}$. The vector $\vX_E$ is in the domain of attraction of a bivariate generalized Pareto vector $\vZ = (Z_1,Z_2)^T$ \citep{rootzen2006} if, for any $\vz \in \mathcal L$,
    \begin{equation*}
    	\PP[\vZ \leq \vz] = \lim_{u \rightarrow \infty} \PP[ \vX_E - u \vc{1} \leq \vz \mid \max(\vX_E) > u].
    \end{equation*}
The bivariate GP distribution is \emph{threshold-stable}: for any scalar $v\ge 0$, the conditional vector $\vZ-v\vc 1 \mid \max(\vZ)>v$ is again bivariate GP. In addition, any bivariate GP vector $\vZ$ can be written as the sum of a unit exponential random variable $E$ and a bivariate (\emph{``spectral''}) random vector $\vS$ whose component-wise minimum equals zero with probability one,
    \begin{equation*}\label{eq:mgpd}
    	\vZ = E - \vS, \qquad E \sim \textnormal{Exp}(1), \, \min(\vS) = 0 \textnormal{ a.s.}, \, E \indep \vS,
    \end{equation*}
    see Equation~(3.2) in \citep{kiriliouk_peaks_2019}.
     Since the margins of $\vX_E$ are unit exponential, we find that $Z_j \mid Z_j > 0$ is also unit exponential for $j=1,2$; we say that $\vZ$ is \emph{standardised}. We can obtain a general bivariate GP vector $\vZ'$ whose conditional margins $Z_j' \mid Z_j' > 0$ are GP$(\xi_j,\sigma_j)$ via the transformation
    \[
    Z'_j = \sigma_j \frac{\exp(\xi_j Z_{j}) - 1}{\xi_j}, \quad j = 1,2.
    \]   
    The bivariate GP distribution can be used to model data whose components exhibit asymptotic dependence. 
    The value of the tail dependence coefficient $\chi$ depends on the choice of $\vc{S}$, see \citet[Supplementary material, Section B]{kiriliouk_peaks_2019}.

    \subsubsection{Other multivariate extensions of the univariate GP model}\label{sec:GMGP}

    	The exponential-gamma representation of a univariate GP distribution (Lemma~\ref{lem:mixture_GPD}) provides a natural building block for multivariate extensions. For example, \cite{bopp2017exponential} introduce temporal dependence through a Markov chain structure for the gamma random variable, while \cite{yadav2021} replace the exponential numerator by a gamma or a generalized inverse Gaussian distribution, coupled with a latent process to account for spatial dependence.

    Recently, \cite{bacro2025multivariate} proposed a multivariate generalization of the representation in Lemma~\ref{lem:mixture_GPD} by combining independent exponential numerators with shared or individual Gamma denominators, leading to a model that allows for both asymptotic dependence and independence. It does not require marginal standardization and has directly interpretable GP margins. We present the simplest form of their so-called    
      gamma-mixture GP   model, defined by the bivariate vector
    \[
    \left( \frac{E_1}{G + G_1}, \frac{E_2}{G + G_2} \right),
    \]
    where $E_1$ and $E_2$ are standard exponential random variables, and $G$, $G_1$, and $G_2$ follow  gamma distributions with unit scale and shape parameters $\alpha$, $\alpha_1$ and $\alpha_2$, respectively. All variables are assumed to be mutually independent.
By construction,  
    each margin is  GP distributed with tail index $\xi_j = \left(\alpha + \alpha_{j} \right)^{-1}$.
    The residual tail dependence coefficient is
    \[
    \eta= \frac{\alpha +\max(\alpha_1,\alpha_2)}
    {\alpha +2\max(\alpha_1,\alpha_2)},
    \]
    and in the limiting setting $\alpha_1, \alpha_2 \rightarrow 0$ one obtains \(\chi=2^{-\alpha}\). Despite its flexibility, the model has some practical drawbacks.  In particular, in the asymptotic dependence regime 
    (when the tail indices of both components coincide at $\xi = 1/\alpha$),
    the tail dependence coefficient $\chi$ is fully determined by $\xi$, and quickly decreases to zero as $\xi$ increases. \cite{bacro2025multivariate} proposed to use \emph{beta scaling}, i.e., multiplying $G + G_j$ by beta random variables with appropriate parameters, to obtain stronger asymptotic dependence.  However, this requires an a priori commitment to asymptotic dependence. Moreover, sub-asymptotic dependence remains weak as the exponential numerators are independent.  
    
    The construction proposed in the next section retains the idea of extending the representation of Lemma~\ref{lem:mixture_GPD}   but adds an additive shift and possibly dependent numerators, allowing a richer range of extremal dependence and support on $\mathcal L$ (as opposed to $(0,\infty)^2$ for the  model  of \citet{bacro2025multivariate}), in line with the bivariate GP distribution.

	\section{A sub-asymptotic model for threshold exceedances}\label{sec:ourmodel}
	
	We propose a new model for bivariate extremes which generalizes both the bivariate GP and the \citet{bacro2025multivariate} models. Our construction captures asymptotic dependence and asymptotic independence in a unified framework, with explicit expressions for the tail coefficients $\chi$ and $\eta$, 
    and operates directly on the original scale of the data (no pre-processing or marginal standardization is required).
    We start by describing the model and its asymptotic properties in Section~\ref{sec:bivariate} before moving to its sub-asymptotic behaviour in Section~\ref{sec:subasymptotic}.
	
	\subsection{Model construction and asymptotic properties}
	\label{sec:bivariate}
    
Definition~\ref{def:biv_model} introduces our bivariate model construction. We will study its marginal distributions, including their speed of convergence towards the GP distribution, and its asymptotic tail dependence properties.
     \begin{defn}\label{def:biv_model} 
     The bivariate random vector $\vY=(Y_1,Y_2)^T$ is said to follow  a \emph{sub-asymptotic bivariate GP (sBGP) distribution} if it can be written as 
     \begin{equation*}
                  \vY = \left(\beta_1 \left( \frac{w E + (1-w) E_1}{G + G_1} - S_1 \right), \beta_2 \left( \frac{w E + (1-w) E_2}{G + G_2} - S_2 \right)\right)^T,
     \end{equation*}
     where $\beta_1,\beta_2$ are positive scale parameters, $w \in [0,1]$ is a weight, 
     $E_1,E_2$ and $E$ represent standard exponential random variables and $G,G_1$ and $G_2$ follow gamma distributions with unit scales and shape parameters $\alpha$, $\alpha_1$ and $\alpha_2$, respectively.  
     The bivariate vector $(S_1,S_2)^T$ is a non-negative random  vector  such that $\min(S_1,S_2)=0$. All random variables are mutually independent and 
     the non-negative constants $\alpha_1, \alpha_2, \alpha$ satisfy $\alpha+\min(\alpha_1,\alpha_2) > 0$.
    \end{defn}
By Breiman's lemma \citep{breiman1965some}, we know how  multiplication by an independent, lighter-tailed variable (here, a weighted sum of exponentials) preserves the tail order.
Hence, 
 the vector $\vY$ has heavy-tailed margins because $1/(G + G_j)$ is regularly varying of order $(\alpha + \alpha_j)$ for $j = 1,2$.
 The parameters $\alpha_1,\alpha_2,\alpha$ control the relative importance of individual and shared gamma components, drive residual tail dependence, and determine the marginal tail indices; see Proposition~\ref{prop:marg_Y}\ref{prop:marg_Y_tail}. The weight $w$ controls the balance between shared and individual exponential components: as $w$ increases from 0 to 1, the strength of (sub-)asymptotic dependence increases (see also Section~\ref{sec:subasymptotic}).  Table~\ref{tab:model_summary} lists the values of the coefficients $\chi$ and $\eta$ for key parameter configurations, which are obtained in Proposition~\ref{prop:tail_dep}. 

     	\begin{table}[h!]
		\centering
		\renewcommand{\arraystretch}{1.5}
		\setlength{\tabcolsep}{12pt}
		\begin{tabular}{cccc}
			\toprule
			Case & $w = 0$ & $w \in (0,1)$ & $w = 1$ \\
			\midrule
			$\max(\alpha_1,\alpha_2)>0$ (AI, $\chi = 0$) &
			\multicolumn{3}{c}{$\eta = \dfrac{\alpha + \max (\alpha_1, \alpha_2)}{\alpha+ 2\max (\alpha_1, \alpha_2)}$} \\
			\midrule
			$\alpha_1 = \alpha_2 = 0$ (AD, $\eta = 1$) &
			$\chi = 2^{-\alpha}$ &
			$\chi \in \left(2^{-\alpha}, 1\right)$ &
			$\chi = 1$ \\
			\bottomrule
		\end{tabular}
		\caption{Tail dependence coefficients $\chi$ and $\eta$ as functions of the weight $w$ and the parameters $\alpha_1$, $\alpha_2$, $\alpha$.}
		\label{tab:model_summary}
	\end{table}
    
        The \emph{shift vector} $(S_1,S_2)^T$ allows for distinct sub-asymptotic dependence shapes. For example, we can use
    \begin{equation}\label{eq:S_T_definition}
    S_j = \max(T_1,T_2) - T_j, \qquad T_j \siid N(0,\sigma_T^2), \qquad j=1,2.
    \end{equation}
    This construction implies that $S_j\ge0$ and automatically enforces $\min(S_1,S_2)=0$, ensuring that $\vY$ is supported on the canonical L-shaped region $\mathcal L$.
    The parameter $\sigma_T$ controls 
    the dispersion of $\vS$, 
   allowing one to model finite-sample dependence without altering the asymptotic tail dependence coefficients.

	\paragraph{Marginal distributions.} 

    \begin{prop}[Marginal density, moments, and upper GP tail behaviour of $Y_j$.]
    \label{prop:marg_Y}
    Let $\vY$ be as in Definition~\ref{def:biv_model}, $\bm S$ as in \eqref{eq:S_T_definition} and fix $j\in\{1,2\}$. 
    \begin{enumerate}[label=(\roman*)]
    \item \textbf{Density.} \label{prop:marg_Y_density}
    The marginal density of $Y_j$ is
    \begin{equation}\label{eq:marg_Y_density}
    f_{Y_j}(y)
    = \frac{1}{2 \beta_j} \Bigl\{ f_{V_j}(y/\beta_j)
    + \int_{0}^{\infty} f_{V_j}(y/\beta_j + s)\, f_{S_j \mid S_j > 0}(s)\,\diff s \Bigr\}.
    \end{equation}
    where
    \[
    V_j=\frac{wE+(1-w)E_j}{G+G_j},
    \]
    $f_{V_j}$ is given in Lemma~\ref{lem:marg_V}, and the distribution of $S_j \mid S_j > 0$ is given in \eqref{eq:Sconddist}.
    
    \item \textbf{Moments.} \label{prop:marg_Y_moments}
    Let $ \xi_j = (\alpha + \alpha_j)^{-1}$. The mean and variance are 
    \begin{align}
    \EE[Y_j] & = \beta_j \left( \frac{\xi_j}{1-\xi_j} - \EE[S_j] \right),
    \qquad \text{if } \xi_j \in (0,1), \label{eq:expec_Y}\\
    \textnormal{Var}[Y_j] & = \beta_j^2 \left( \frac{w^2+(1-w)^2+2\xi_j w(1-w)}{(\xi_j^{-1}-1)^2(1-2\xi_j)}
    + \textnormal{Var}[S_j] \right),
    \qquad \text{if } \xi_j \in \bigl(0,\tfrac12\bigr). \label{eq:var_Y}
    \end{align}
    
    \item \textbf{Upper GP tail convergence.} \label{prop:marg_Y_tail}
    If $\xi_j = (\alpha + \alpha_j)^{-1}$ and $\sigma_j =  \beta_jC^{\xi_j}\xi_j$ with $C = \frac{w^{1/\xi_j+1}-(1-w)^{1/\xi_j+1}}{2w-1}$, then   
    \[
    \dfrac{\PP(Y_j >x)}{\overline H_{\xi_j,\sigma_j}(x)}=1- \rbr{\frac{\Delta + o(1)}{x}}, \qquad \text{ as } x \rightarrow \infty,
    \]
    where	\(	\Delta= \dfrac{\EE[S_j] + D - C^{\xi_j}}{\xi_j/\beta_j} \)	and	\(D=\dfrac{w^{1/\xi_j+2}-(1-w)^{1/\xi_j+2}}{w^{1/\xi_j+1}-(1-w)^{1/\xi_j+1}}.\)
    \end{enumerate}
    \end{prop}
If  $w=0$ or $w=1$, the random variable $V_j$ follows a $\GP(\xi_j,\xi_j)$ distribution (Lemma~\ref{lem:mixture_GPD}).
Note that the parameter $w$ does not appear in the mean of $Y_j$ (Proposition~\ref{prop:marg_Y}\ref{prop:marg_Y_moments}). 
Figure~\ref{fig:density_varying_w_sigmaT} shows how $w$ and $\sigma_T$ modify the \emph{body} of the marginal density without changing its asymptotic tail. In the left panel, increasing $w$ makes the positive part closer to a Pareto shape: when $w=1$, the leading term reduces to $E/(G_j+G)$, yielding a Pareto-like density, whereas values of $w$ near $1/2$ replace the numerator by a hypo-exponential mixture, slightly smoothing the peak. In the right panel, larger $\sigma_T$ spreads mass toward negative values through the additive shift $S_j$, slowing the convergence toward the GP tail while preserving its slope.

    \begin{figure}[htbp]
        \centering
        \begin{subfigure}[b]{0.48\textwidth}
            \includegraphics[width=\textwidth]{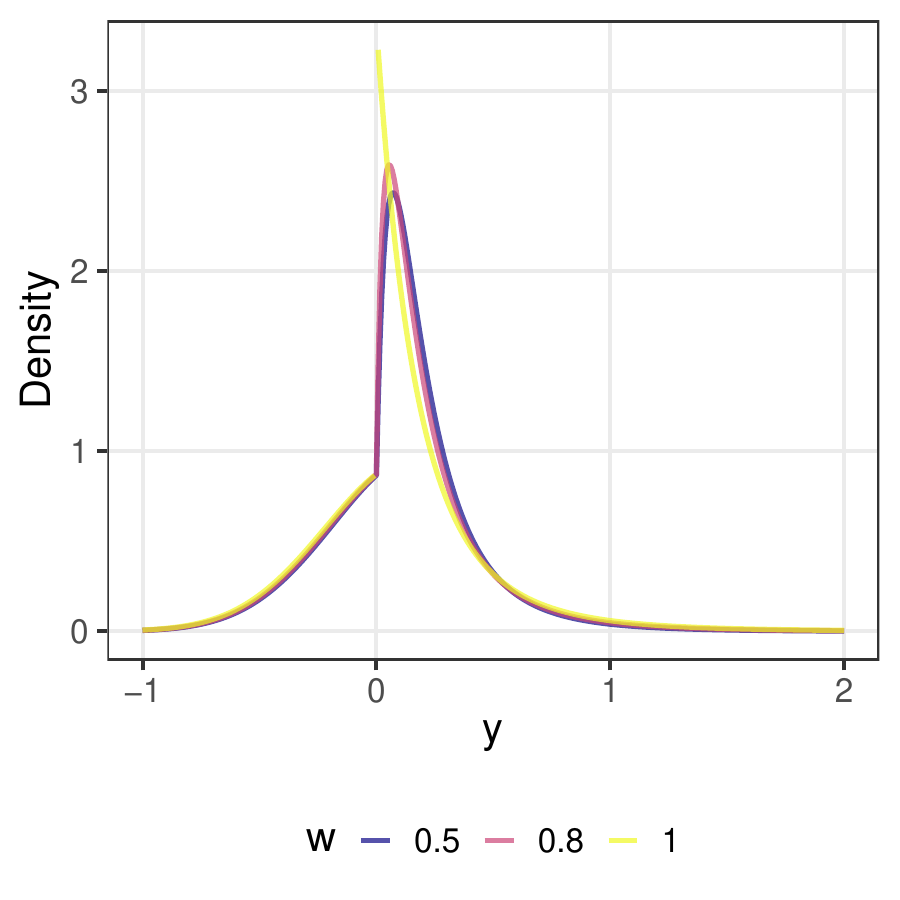}
        \end{subfigure}
        \hfill
        \begin{subfigure}[b]{0.48\textwidth}
            \includegraphics[width=\textwidth]{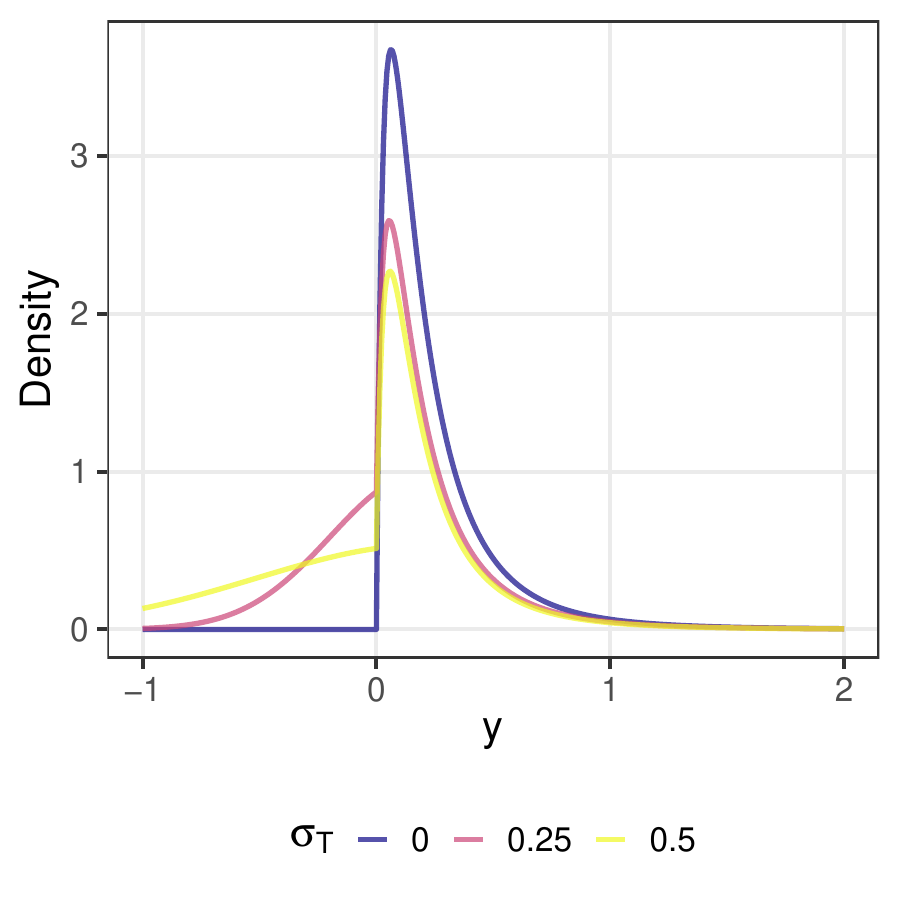}
        \end{subfigure}
        \caption{ Probability density functions of $Y_j$ in Definition 
      \ref{def:biv_model} with  $\xi_j = 0.2$ and $\beta_j = 1$. 
      In the left panel, the parameter $\sigma_T=0.25$ is fixed and 
       the weight $w \in \{0.5, 0.8, 1\}$ (blue, red, yellow lines), while $w=0.8$ and $\sigma_T \in \{0, 0.25, 0.5\}$ (blue, red, yellow lines) in the right panel.}
        \label{fig:density_varying_w_sigmaT}
    \end{figure}

	\paragraph{Connections with the multivariate GP distribution.}

    The multivariate GP distribution arises as a boundary case of our proposed construction. Recall that the gamma variables $G_j+G$ have shape parameter $\alpha+\alpha_j$. When $\alpha+\alpha_j\to\infty$, the relative fluctuations of $G_j+G$ vanish, and these variables behave as deterministic constants after rescaling. In addition, setting $w=1$ removes the mixture component in the numerator, leaving a purely additive exponential representation. In this regime, the vector $\vY$ converges to the classical additive form of a multivariate GP distribution with standardised margins.

    \begin{lem}[Multivariate upper GP limit]\label{lem:mgpd_limit}
        Let ${\bf Y}=(Y_1,Y_2)^T$ be a sBGP random vector as in  Definition~\ref{def:biv_model} with $w=1$. 
        Let $\vS^{(\beta)}=(S^{(\beta)}_1,S^{(\beta)}_2)^T$ be the associated shift vector, possibly depending on 
        $\bm \beta=(\beta_1,\beta_2)^T$, and set $\beta_j=\sigma_j(\alpha+\alpha_j)$ for some $\sigma_j>0$.
        Assume that $\alpha+\alpha_j\to\infty$ for $j=1,2$ and that
        \[
        \beta_j S^{(\beta)}_j  \;\xrightarrow{p}\;  S^0_j , \qquad j=1,2,
        \]
        for some non-negative random vector $\vS^0=(S^0_1,S^0_2)^T$ with $\min(S^0_1,S^0_2)=0$ a.s.
        Then
        \[
        (Y_1,Y_2) \;\xrightarrow{p}\; 
        \big(\sigma_1 E-S^0_1,\;\sigma_2 E-S^0_2\big),
        \]
        where $E\sim\mathrm{Exp}(1)$. The limit vector is bivariate GP with
        shape vector $\bm \xi=(0,0)^T$ and scale vector $\bm \sigma=(\sigma_1,\sigma_2)^T$.
    \end{lem}
	  
    In the multivariate GP limit, the shift vector $(S_1,S_2)^T$ allows for distinct asymptotic dependence shapes \citep[Supplementary material, Section D]{kiriliouk_peaks_2019}. 
    For the general model as in Definition~\ref{def:biv_model}, $(S_1,S_2)^T$ has no impact on asymptotic model properties, driving sub-asymptotic dependence only.   
		When $w= \beta_1 = \beta_2 = 1$, the sBGP model of Definition~\ref{def:biv_model} can be rewritten as 
		\[
		(Y_1,Y_2)=\left( \frac{E-\tilde S_1}{G + G_1},\; \frac{E-\tilde S_2}{G + G_2}\right),
		\]
		where $\tilde S_j=(G+G_j) S_j$ for $j = 1,2$.
		Since $\tilde{S}_1,\tilde{S}_2\ge 0$ and $\min(\tilde{S}_1,\tilde{S}_2)=0$, we find a bivariate standardised GP vector $(E - \tilde{S}_1, E - \tilde{S}_2)$ in which each component is divided by a gamma random variable. 

	\paragraph{Tail dependence properties.}
	The tail dependence of $\vY$ is governed by the shared exponential- and gamma-distributed variables.  We now derive the values of the coefficients $\chi$ and $\eta$ that are listed in  Table~\ref{tab:model_summary}. We see that the distinction between AD and AI depends solely on the configuration of the gamma-distributed variables. This is consistent with the fact that the gamma-distributed denominator (i.e. the multiplicative inverse gamma-distributed factor) has heavier tails than the weighted sum of exponentials in the numerator.

	In particular, we find that the residual tail dependence coefficient $\eta$ is entirely determined by the gamma shape parameters $\alpha, \alpha_1, \alpha_2$ and is independent of the weight $w$. By contrast, in case of asymptotic dependence, the tail dependence coefficient $\chi$ can reach a full spectrum of tail dependence (i.e. $\chi \in [0,1)$) thanks to the weight $w$.

\begin{prop}[Tail dependence structure]\label{prop:tail_dep}
    Let $\vY$ be as in Definition~\ref{def:biv_model}, with marginal tail indices $\xi_j = (\alpha+\alpha_j)^{-1}$ for $j=1,2$.
    
    \begin{enumerate}[label=(\roman*)]
     \item \textbf{Residual tail dependence}\label{prop:tail_dep_eta}
      If $\alpha_1>0$ or $\alpha_2>0$, then $Y_1$ and $Y_2$ are asymptotically independent, so that
    $\chi=0$ and 
    \[
    \eta
    = \frac{\alpha + \max(\alpha_1,\alpha_2)}{\alpha + 2\max(\alpha_1,\alpha_2)},
    \]
    In particular, 
    $\eta = \tfrac12$ if $\alpha = 0$.
    \item \textbf{Tail dependence.}\label{prop:tail_dep_chi}
   If $\alpha_1=\alpha_2=0$, then by assumption, $\alpha > 0$, and $Y_1, Y_2$ are asymptotically dependent, so that $\eta = 1$ and
    \[
    \chi
    = \frac{2w-1}{3w-1}\,
    \frac{2 w^{\alpha+1}-2^{-\alpha}(1-w)^{\alpha+1}}
         {w^{\alpha+1}-(1-w)^{\alpha+1}},
    \]
    with $\chi = 2^{-\alpha}$ when $w=0$ and $\chi = 1$ when $w=1$.
    \item \textbf{Constraints induced by marginal asymmetry.}\label{prop:tail_dep_eta_bounds}
    The residual tail dependence coefficient satisfies
    \[
    \eta \in
    \left[
    \frac12,\;
    \frac{1}{\,2-\frac{\min(\xi_1,\xi_2)}{\max(\xi_1,\xi_2)}\,}
    \right],
    \]
    while $\chi >0 $ only if $\xi_1 = \xi_2$.
    \end{enumerate}
\end{prop}
 \noindent Proposition~\ref{prop:tail_dep}\ref{prop:tail_dep_eta_bounds} illustrates that increasing asymmetry between the marginal tail indices $\xi_1$ and $\xi_2$
    reduces the maximal attainable value of $\eta$.
     When $\xi_1=\xi_2$, 
the model attains full flexibility, with $\eta\in[1/2,1]$ spanning the range from independence to asymptotic dependence; within the latter regime ($\eta=1$), the tail dependence coefficient $\chi$ varies continuously in $[2^{-\alpha},1]$. Note that asymptotic dependence cannot arise when $\xi_1 \neq \xi_2$.
    
    The weight $w$ does not impact the value of $\eta$ when the model exhibits AI but has a dominant role in shaping the tail dependence coefficient $\chi$ in case of AD, as illustrated in Figure~\ref{fig:chi_eta_curves_side_by_side}. The  left panel shows  that varying $w$ leads to large changes in $\chi$, whereas changes in the marginal tail index $\xi=1/\alpha$ have a comparatively minor effect. The right panel illustrates how quickly $\eta$ increases to 1 as a function of the shared coefficient $\alpha$.

    \begin{figure}[ht]
  \centering
  \includegraphics[width=0.9\linewidth]{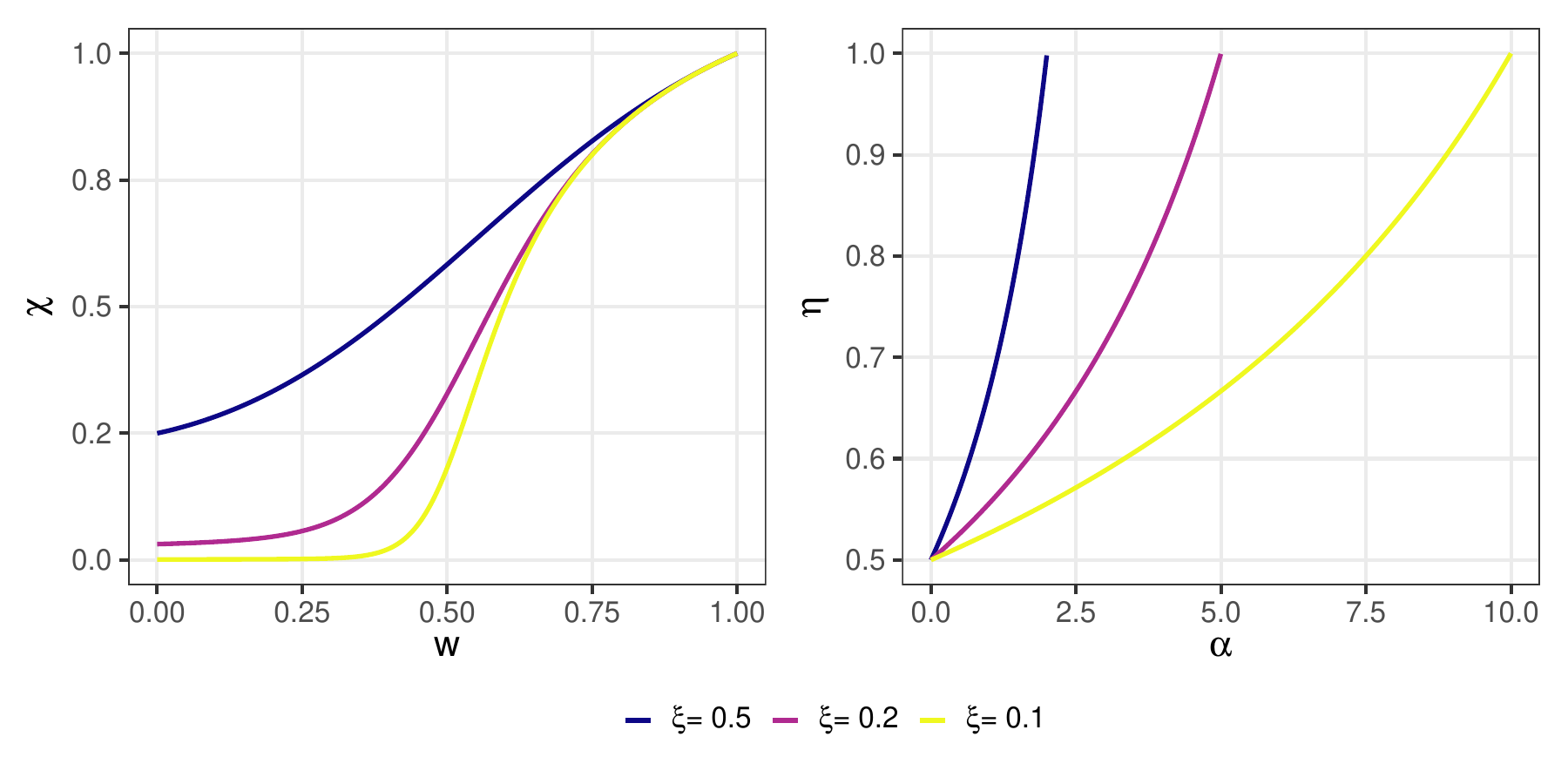}
  \caption{Limiting  tail dependence coefficients $\chi$ (left panel) and $\eta$ (right panel) as functions of the weight $w$ in Definition  \ref{def:biv_model}, for the special case of identical tail indices  $\xi_1=\xi_2=\xi$. }
  \label{fig:chi_eta_curves_side_by_side}
\end{figure}

\subsection{Sub-asymptotic model features}
\label{sec:subasymptotic}


While Section~\ref{sec:bivariate} focused on asymptotic properties, convergence to the limiting coefficients $(\chi,\eta)$ can be slow, so that dependence at practically relevant thresholds may deviate substantially from the asymptotic regime. To characterize such behaviour, we examine the model’s \emph{sub-asymptotic} features, which describe extremal dependence at intermediate quantile levels that are typically accessible in finite samples.

Asymptotic coefficients alone are therefore insufficient to summarize extremal dependence in practice: distinct processes can share the same limiting coefficients $(\chi,\eta)$ while exhibiting markedly different dependence strengths at moderately high quantile levels. This phenomenon is illustrated in Figure~\ref{fig:subasymp_datasets}, which shows three simulated datasets with identical $(\chi,\eta)=(0,0.9)$ but different mixture weights $w$, resulting in clearly different sub-asymptotic dependence patterns despite identical asymptotic behaviour.
Hence, model diagnostics will focus on the full curves $\chi(q)$ and $\eta(q)$ rather than just their limits $\chi$ and $\eta$ \citep[see also][]{huser2025modeling,delloroFlexibleSpaceTime2025}. The curves $\eta(q)$ and $\chi(q)$ contain equivalent information, see \eqref{eq:chieta}. We only focus on $\chi(q)$ in what follows.

\begin{figure}[ht]
  \centering
  \includegraphics[trim={.4cm 0 .4cm 0},clip,width=0.99\linewidth]{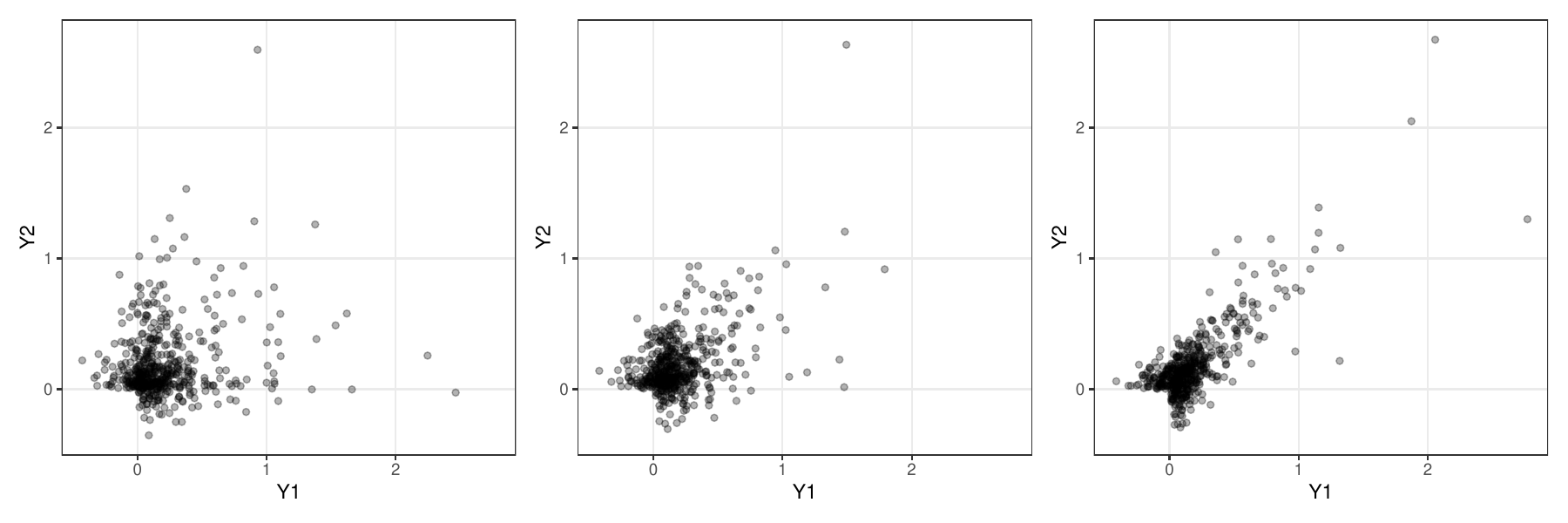}
  \caption{Scatterplots of samples of size $n=1000$ simulated from the sBGP model (Definition~\ref{def:biv_model}) with identical coefficients $(\chi,\eta) = (0, 0.9)$ and parameters $(\alpha,\alpha_1,\alpha_2,\beta_1,\beta_2,\sigma_T) = (4.44, 0.56, 0.56, 1, 1, 0.1)$. From left to right: $w = 0.1$, $w = 0.5$, $w = 0.9$.}
  \label{fig:subasymp_datasets}
\end{figure}

No closed-form expression is available for $\chi(q)$ when $q<1$. We approximate it empirically using a large simulated sample $\vY_1,\ldots,\vY_N$ from the model,
\begin{equation}\label{eq:chi_esti}
    \hat{\chi} (q) =  \frac{1}{N(1-q)} \sum_{i=1}^N \mathbf{1}\{ R_{i1} > (N+1)q,\; R_{i2} > (N+1)q \}, 
\end{equation}
where $R_{ij}$ denotes the rank of $Y_{ij}$ among $(Y_{1j},\ldots,Y_{Nj})$ for $j=1,2$. 
This estimator corresponds to the proportion of joint exceedances above the empirical $q$-quantiles, normalized by the expected number of marginal exceedances $N(1-q)$. 
Figure~\ref{fig:chi_curve_sub} displays the resulting $\chi(q)$ curves. For small $w$ (left panel), curves corresponding to different $\eta$ values remain nearly indistinguishable except extremely close to $q=1$, reflecting weak sub-asymptotic separation. As $w$ increases (middle and right panels), the weighted exponential sum induces stronger joint extremes and the curves separate earlier, making $\eta$ practically identifiable from finite samples. The orange curves correspond to the three datasets of Figure~\ref{fig:subasymp_datasets}$:$ although they share identical asymptotic coefficients, their $\chi(q)$ trajectories differ substantially at intermediate quantiles, highlighting the importance of full $\chi(q)$ diagnostics."

\begin{figure}[ht]
  \centering
  \includegraphics[width=\linewidth]{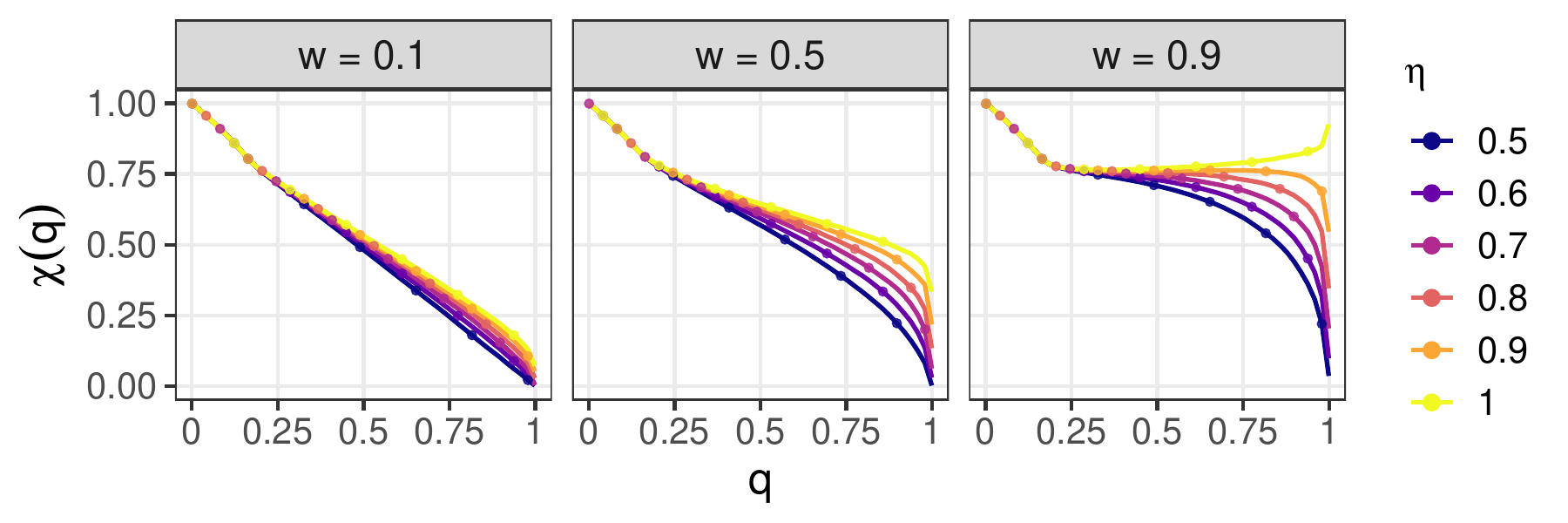}
  \caption{Empirical $\chi(q)$ curves for the sBGP model (\ref{def:biv_model}). The parameters $(\alpha,\alpha_1,\alpha_2)$ are chosen such that $\xi_1=\xi_2=0.2$ and 
  $\eta \in \{0.5,0.6,\ldots,1\}$, with $\sigma_T=0.1$.  Left: $w=0.1$; Middle: $w=0.5$; Right: $w=0.9$. For each pair $(w,\eta)$, ${\chi}(q)$ is estimated using \eqref{eq:chi_esti}, based on a single synthetic sample of size $N=10^6$, with $q$ ranging up to $0.999$.}
  \label{fig:chi_curve_sub}
\end{figure}

In case of asymptotic dependence, Figure~\ref{fig:eta_w0} further illustrates the slow convergence of $\eta(q)$ to its asymptotic limit $\eta=1$, for $w=0$, $\sigma_T=0$, and $\xi_1=\xi_2=0.2$ (equivalently, $\alpha_1=\alpha_2=0$ and $\alpha=5$), for which a closed-form expression of $\eta(q)$ is available \citep{bacro2025multivariate}.
Hence, for practically observable quantiles (e.g., $q\le 0.9999$), $\eta(q)$ may misleadingly suggest asymptotic independence. A positive shift $\sigma_T$ lowers the entire curve but does not accelerate convergence, i.e.,
\[
\eta(q;\sigma_T>0) < \eta(q;\sigma_T=0), \qquad q<1,
\]
while the limit is unchanged.

\begin{figure}[ht]
  \centering
  \includegraphics[,width=0.7\linewidth]{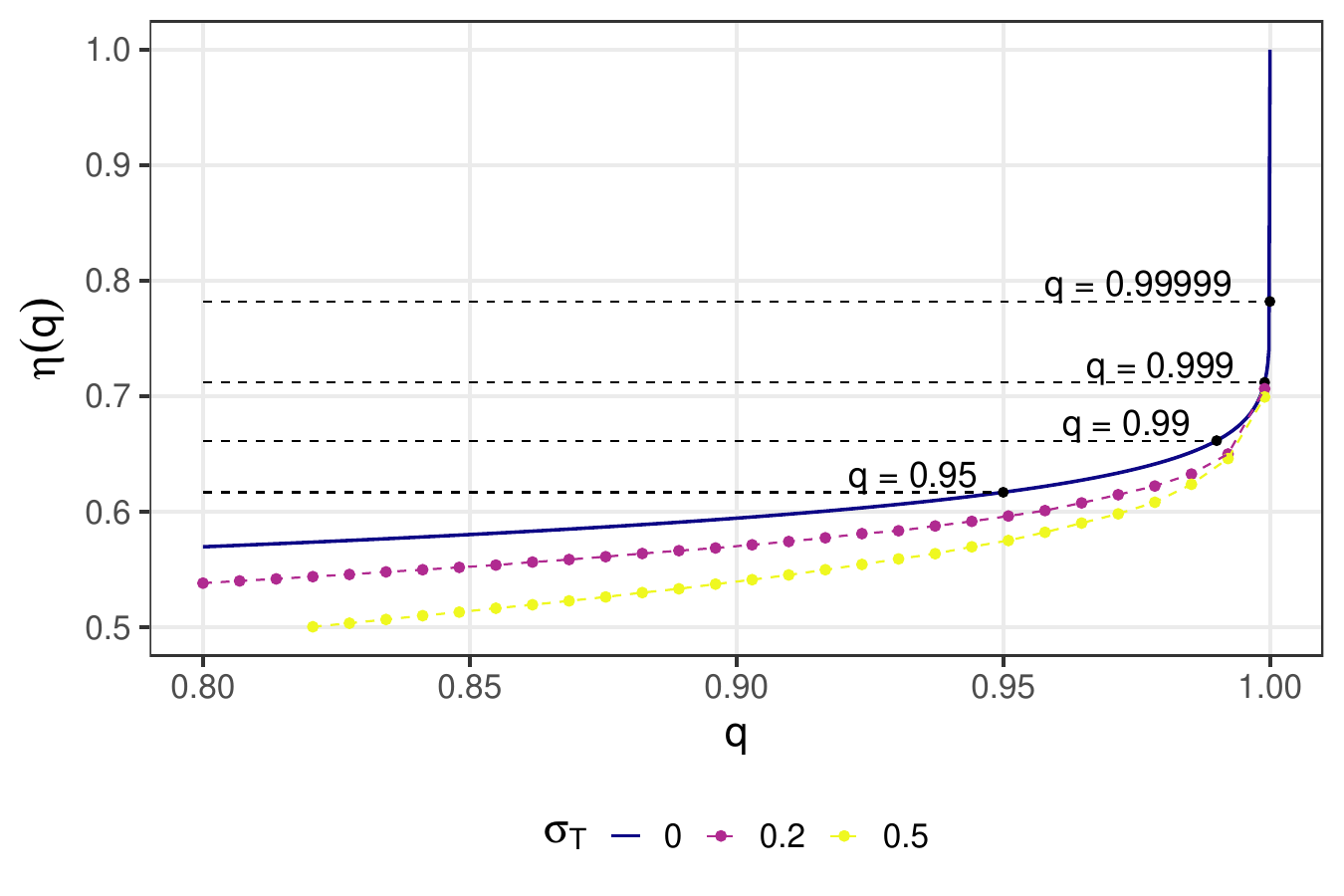}
  \caption{Residual tail-dependence curves $\eta(q)$ for $w=0$ in the asymptotically dependent case with $\alpha_1=\alpha_2=0$ and $\alpha=5$ (so that $\xi_1=\xi_2=0.2$ and $\chi=1/32$). Blue: theoretical curve with $\sigma_T=0$; Purple and yellow: empirical curves with respectively $\sigma_T=0.2$ and $\sigma_T=0.5$.}

  \label{fig:eta_w0}
\end{figure}

\paragraph{Inflection point in the $\chi(q)$ curve.}
The $\chi(q)$ curves in Figure~\ref{fig:chi_curve_sub} display an inflection around low to moderate quantiles. This feature is a consequence of the L-shaped support $\mathcal L$. As $q$ increases, the conditioning set in  the $\chi(q)$ curve definition~\eqref{eq:chi_curve} progressively aligns with the geometry of the support. The transition around the marginal threshold creates the observed change in curvature, which is therefore an intrinsic property of the model rather than an artifact. Beyond this point, the curve continues smoothly toward its limiting value. 

The shift vector $\vS$ (as defined in~\eqref{eq:S_T_definition}) drives this behaviour by ensuring the model is supported on $\mathcal L$, capturing observations where one component is extreme while the other remains moderate. Consequently, $\vS$ governs the finite-sample dependence and the shape of the $\chi(q)$ curve, while leaving the asymptotic coefficients $(\chi, \eta)$ unchanged.


	\section{Neural Bayes Estimation}\label{sec:estimation}
	Inference for complex extreme value models is generally challenging, as their likelihood functions often lack closed-form expressions or require computationally expensive numerical approximations. In such scenarios, likelihood-free methods relying on simulation from the model become particularly valuable; examples include approximate Bayesian computation \citep{sisson_overview_2018}, synthetic likelihoods \citep{wood_statistical_2010}, and pseudo-marginal methods \citep{andrieu_pseudo-marginal_2009}. More recently, neural network–based likelihood-free approaches, such as the neural Bayes estimator (NBE) introduced in \citet{sainsbury-dale_likelihood-free_2024}, have been increasingly adopted because of their computational efficiency (after initial training) and their ability to learn informative summary statistics automatically, rather than relying on manually chosen ones. Recent work by \citet{andre_neural_2025} further demonstrated the effectiveness of the NBE framework for inference in complex bivariate extreme-value models.


    The NBE approximates the Bayes estimator using a neural network trained on simulated data. Unlike classical inference methods that rely on explicit likelihoods or carefully selected summary statistics, the NBE directly learns a mapping from the data to the parameter estimates, automatically identifying informative features.

    Let $\Theta \subseteq \mathbb{R}^{p}$ denote the parameter space and $f(\, \cdot \mid \vthe)$ the density of the sBGP distribution, depending on the unknown parameter vector $\vthe \in \Theta$.
    Following \citet{sainsbury-dale_likelihood-free_2024}, the Bayes estimator $\hat{\bm{\theta}}$ minimizes the Bayes risk 
    $\int_{\Theta} R(\vthe, \hat{\vthe})\,\pi(\text{d}\vthe )$, where $R(\vthe, \hat{\vthe}) = \mathbb{E}_{\vY \sim f(\, \cdot \mid \vthe)}[L(\vthe, \hat{\vthe}(\vY))]$ is the expected loss under $f(\cdot \mid \vthe)$, and $\pi$ is a prior distribution on $\vthe$.
    Since this quantity is generally intractable, the NBE approximates it by empirical risk minimization over simulated observations
    \[
    \vY^{(k)}_1,\ldots, \vY^{(k)}_{M_k}  \stackrel{\text{iid}}{\sim} f(\cdot \mid \vthe^{(k)}), \quad
    \vthe^{(k)} \sim \pi(\cdot), \quad k = 1,\ldots, K
    \]
    where $\vY^{(k)}_i = (Y^{(k)}_{i1},Y^{(k)}_{i2})$ for $i = 1,\ldots, M_k$ and the sample size $M_k$ is randomly drawn from a discrete distribution to improve generalization across datasets of varying size.
    The neural network represents the estimator $\hat{\vthe}(\cdot)$ and is trained by minimizing the empirical squared loss
     \begin{equation}\label{eq:classical_loss}
    \hat{\vthe}_{\text{NBE}}= \arg\min_{\hat{\vthe}\in\mathcal{F}} \left\| D^{-1}(\vthe-\hat{\vthe}(\vY))\right\|_2^2 
    = \arg\min_{\hat{\vthe}\in\mathcal{F}} \sum_{k=1}^{K} \left\| D^{-1}\left(\vthe^{(k)}- \hat{\vthe}(Y^{(k)})\right)\right\|^2.
    \end{equation}
    where $Y^{(k)}$ is the $M_k\times2$ matrix with rows $\vY^{(k)}_1,\ldots,\vY^{(k)}_{M_k}$ and $\mathcal{F}$ denotes the class of functions represented by the neural network. To account for the different scales of the parameters, we normalize the loss using a diagonal scaling matrix $D$, whose entries correspond to typical magnitudes (or empirical standard deviations) of each component of $\vthe$.
    
    In addition to the classical loss defined in \eqref{eq:classical_loss}, we also consider a penalized estimator, denoted $  \hat{\vthe}_{\text{NBE}}^{(p)}$, which minimizes the loss function
    \begin{equation}\label{eq:pen_loss}
    L(\vthe,\hat{\vthe},\vY) = \left\| D^{-1}(\vthe-\hat{\vthe}(\vY))\right\|_2^2+\lambda D_{11}^{-2}\left(\hat{\eta}(\vY) - \hat{\eta}_{emp}(\vY)\right)^2.
    \end{equation}
    The additional term penalizes discrepancies between the estimated residual tail dependence coefficient $\hat{\eta}(\vY)$ and its empirical counterpart $\hat{\eta}_{emp}(\vY)$, where $\hat{\eta}_{emp}(\vY)$ denotes a Hill-type estimator computed from pseudo-observations $U_{ij}=R_{ij}/(n+1)$, with $R_{ij}$ the rank of $X_{ij}$ among $(X_{1j},\ldots,X_{nj})$. Setting $Z_i = \min\{(1-U_{i1})^{-1},(1-U_{i2})^{-1}\}$,
    \begin{equation}\label{eq:eta_Hill}
    \hat{\eta}_{emp}(\vY) = \frac{1}{k} \sum_{j=1}^k \log \frac{Z_{(n-j+1)}}{Z_{(n-k)}},
    \end{equation}
    where $k = \lfloor n/10 \rfloor$. This acts as a regularization toward a standard non-parametric estimator of $\eta$, improving stability when $\eta$ is weakly identifiable. In practice, we set $\lambda = 0.5$.

    Permutation invariance across replicates is enforced through the DeepSets architecture \citep{zaheer_deep_2018}. Each observation is first transformed by a feature map $\psi(\cdot)$, the resulting representations are aggregated via the mean, and the summary vector is mapped to parameter estimates by $\phi(\cdot)$:
    \[
    \hat{\vthe}(Y^{(k)})=\phi\left(\frac{1}{M_k}\sum_{i=1}^{M_k}\psi\left(\vY^{(k)}_i\right)\right).
    \]
    Once trained, the NBE delivers fast, amortized parameter inference without further simulation.

    \paragraph{Prior choice.}
    The performance of amortised, simulation-based inference depends critically on the choice of the parameter prior, which shapes the synthetic training distribution and determines which regions of the parameter space are learned accurately. More generally, accuracy depends not only on the number of simulations but also on their quality and relevance to the inferential task; well-designed synthetic data improves generalisation in deep models \citep{he_data_2019, ding_case_2019}.
	
    For the parameters $(\beta_1,\beta_2,\sigma_T,w)$ the choice is straightforward, as they are only weakly interdependent, and the distribution of $\vS$ is chosen as in \eqref{eq:S_T_definition}. In contrast, $(\alpha,\alpha_1,\alpha_2)$ determine the tail indices $\xi_j = 1/(\alpha+\alpha_j)$ and the residual tail dependence coefficient $\eta$. Unrealistically large tail indices generate synthetic samples with excessively large values, causing the training procedure to fail. Hence, we place the prior directly on $(\xi_1,\xi_2,\eta)$, the real quantities of interest, and recover $(\alpha,\alpha_1,\alpha_2)$ via the transformations in Proposition~\ref{prop:tail_dep}.

    To enforce the constraint linking $\xi_1,\xi_2$ and $\eta$ (Proposition~\ref{prop:tail_dep}\,\ref{prop:tail_dep_eta_bounds}) while maintaining a symmetric treatment of the margins (i.e. the same univariate prior for $\xi_1$ and $\xi_2$), we sample $\eta$ and one of the tail indices freely, and then draw the remaining tail index conditionally. 
    This strategy avoids drawing $\xi_1, \xi_2$ independently and then $\eta$ conditionally, which would constrain the admissible range of $\eta$ too strongly when $\xi_1$ and $\xi_2$ are substantially different. However, asymptotic dependence corresponds to the limiting case $\alpha_1=\alpha_2=0$, such a prior never produces samples with $\eta = 1$, causing the network to underperform when tail dependence is strong. To address this, we adopt a mixture prior that draws from the scheme described above (the \emph{AI scheme}) with probability $0.9$ and fixes $\eta=1$ in the remaining $10\%$ of cases (the \emph{AD scheme}). This ensures that the training set includes a non-negligible proportion of asymptotically dependent configurations, allowing the estimator to learn effectively across the full dependence spectrum.

    Formally, the sampling scheme (using a conditional prior for $\xi_2$) is
    \[
    \pi_{\eta} \sim 0.9 \,\text{Unif}(\tfrac{1}{2},1) + 0.1\,\delta_{\{\eta = 1\}}, 
    \quad 
    \pi_{\xi_1} \sim \text{Unif}(0,\tfrac{1}{2}),
    \quad 
    \pi_{\xi_2 \mid \eta,\xi_1} \sim \text{Unif} \left(\frac{\xi_1(2\eta-1)}{\eta},\, \frac{\xi_1\eta}{2\eta-1}\right),
    \]
    where $\delta_{\{\eta = 1\}}$ denotes the Dirac measure at $\eta=1$. In addition, we take
    \[
    \pi_{\sigma_T} \sim \text{Unif}(0,1),
    \quad 
    \pi_{w} \sim \text{Unif}(0,1),
    \quad
    \pi_{\beta_j} \sim \text{Unif}(0,1000) \quad (j=1,2), 
    \]
    After drawing $(\xi_1,\xi_2,\eta)$, we invert the mapping to $(\alpha,\alpha_1,\alpha_2)$ to generate training samples. The uniform prior $\mathrm{Unif}(0,\tfrac{1}{2})$ on $\xi_1$ and the conditional construction for $\xi_2$ restrict tail indices to realistic values, preventing synthetic samples with excessively heavy tails. Similarly, the prior $\mathrm{Unif}(0,1)$ for $\sigma_T$ was chosen because larger values generate overly dispersed datasets. 
    
    Finally, to improve robustness to varying sample sizes, each synthetic dataset used for training is generated with a sample size drawn from $\mathrm{Unif}(100, 1000)$.

    \paragraph{Network architecture.}\label{par:nbe_arch}
	
    As mentioned earlier, we use a DeepSets-based network. It consists of two components: a feature extractor \(\psi\) applied independently to each replicate, and a final mapping \(\phi\) from the aggregated representation to parameter estimates. The architecture uses ReLU activations throughout, except for softplus activations on positive parameters and sigmoid activations for \(\eta\) and \(w\), constraining them to \([0.5,1]\) and \([0,1]\), respectively. 
    
    The outer (feature-extraction) network $\psi$ consists of four fully connected layers with mapping dimensions $2 \mapsto 64 \mapsto 64 \mapsto 128 \mapsto 128$. The inner mapping $\phi$ takes as input the concatenation of the aggregated representation with a vector of summary statistics $S = \{\hat{\chi}(q)\}$ for $q \in \{0.50, 0.60, 0.70, 0.80, 0.85, 0.90, 0.95, 0.98\}$, where $\hat{\chi}(q)$ denotes the empirical estimator defined in~\eqref{eq:chi_esti}. This results in an input dimension of 136. It then passes through four layers with mapping dimensions $136 \mapsto 128 \mapsto 64 \mapsto 64 \mapsto 7$, with ReLU activations at all hidden layers and parameter-specific activations at the output layer.
This concatenation of learned summary features and expert-chosen features follows a similar strategy to that employed in \citet{delloroFlexibleSpaceTime2025}, where empirical $\chi(q)$ values were likewise incorporated as additional expert summary statistics to complement learned representations. The network is trained with the Adam optimizer (learning rate $10^{-4}$, batch size 32).

\begin{rem}\label{rem:bound_prior_error}
Due to the prior specification and the sigmoid activations, the estimator cannot attain boundary values of the supports of $\eta$ and $w$. Estimates are therefore biased near the boundaries; for instance, $\eta$ tends to be underestimated at $1$ and overestimated at $1/2$, with similar effects for $w$ at $0$ and $1$.
\end{rem}


\section{Simulation Studies}\label{sec:simstudy}

We evaluate the performance of the Neural Bayes Estimator for the sBGP distribution introduced in Section~\ref{sec:bivariate}, with a focus on estimation accuracy and the recovery of extremal dependence features. The estimation procedure was implemented using the \texttt{NeuralEstimators} package \citep{sainsbury-dale_likelihood-free_2024}. Training required approximately one hour on a single NVIDIA Volta V100 GPU (64~GB RAM), while inference takes only a few milliseconds per dataset on a standard laptop.

We examine the estimator’s performance in a fixed-parameter regime, where the true parameter vector remains constant across simulations. This setting allows us to assess both estimation accuracy and uncertainty quantification under different (known) dependence structures. We consider three configurations
\begin{equation}\label{eq:true_theta_param}
\vthe^{(1)} = (3,0,0,20,30,0.1,0.8), \, 
\vthe^{(2)} = (2,1,1,20,30,0.1,0.6), \, 
\vthe^{(3)} = (1,2,2,20,30,0.1,0.2), 
\end{equation}
where $\vthe = (\alpha,\alpha_1,\alpha_2,\beta_1,\beta_2,\sigma_T,w)$. All three cases share the same values for $(\beta_1,\beta_2,\sigma_T)$, so that differences across configurations are entirely driven by the first four parameters, which govern the marginal tail indices, the residual dependence coefficient, and the strength of sub-asymptotic association.

These configurations span distinct regimes of extremal dependence. For $\vthe^{(1)}$, we obtain $\eta=1$, corresponding to asymptotic dependence; combined with $w=0.8$, this yields strong extremal dependence across all quantile levels. For $\vthe^{(2)}$, we have $\eta=0.75$, corresponding to asymptotic independence with positive residual dependence, while $w=0.6$ induces moderate sub-asymptotic association. 
Finally, $\vthe^{(3)}$ yields $\eta=0.6$, corresponding to asymptotic independence with weaker residual dependence, further attenuated at intermediate levels by the small value $w=0.2$. These cases therefore represent, respectively, an asymptotic dependence regime, an intermediate regime with residual tail dependence, and a regime of weak dependence both asymptotically and at sub-asymptotic levels. Figure~\ref{fig:fixed_params_chi_full} (top) shows samples from the sBGP distribution under the three configurations.

For each configuration $\vthe^{(i)}$, $i=1,2,3$, we generate $K = 1000$ independent datasets of size $N = 1000$. We then apply the NBE procedure to obtain a point estimate $\hat{\vthe}^{(i)}_k$ and the associated dependence curve $\hat{\chi}^{(i)}_k(q)$, $k = 1, \ldots, K$. Uncertainty is assessed conditionally on each fitted model via a non-parametric bootstrap with $B = 200$ resamples, yielding empirical confidence intervals for both model parameters and associated dependence measures. The results are summarized below in terms of estimation accuracy of point estimates and calibration of the resulting confidence intervals.

\paragraph{Estimation accuracy.}
Figure~\ref{fig:fixed_params_boxplots_full} displays boxplots of the point estimates $\hat{\vthe}^{(i)}_k$ across the $K$ replicated datasets for each configuration. Overall, the NBE demonstrates good estimation accuracy, with parameter estimates closely centred around their true values. Estimation variability for $\eta$ increases in the second configuration and becomes more pronounced in the third, which is consistent with the progressively weaker sub-asymptotic dependence induced by smaller values of $w$ (see Section~\ref{sec:subasymptotic}). Finally, $\eta$ is never estimated exactly equal to~1 due to the upper-bound constraint discussed in Remark~\ref{rem:bound_prior_error}, which prevents the estimator from attaining boundary configurations. This behaviour is expected and does not affect the overall recovery of the dependence structure.

\begin{figure}[htbp]
	\centering
	\includegraphics[width=\textwidth]{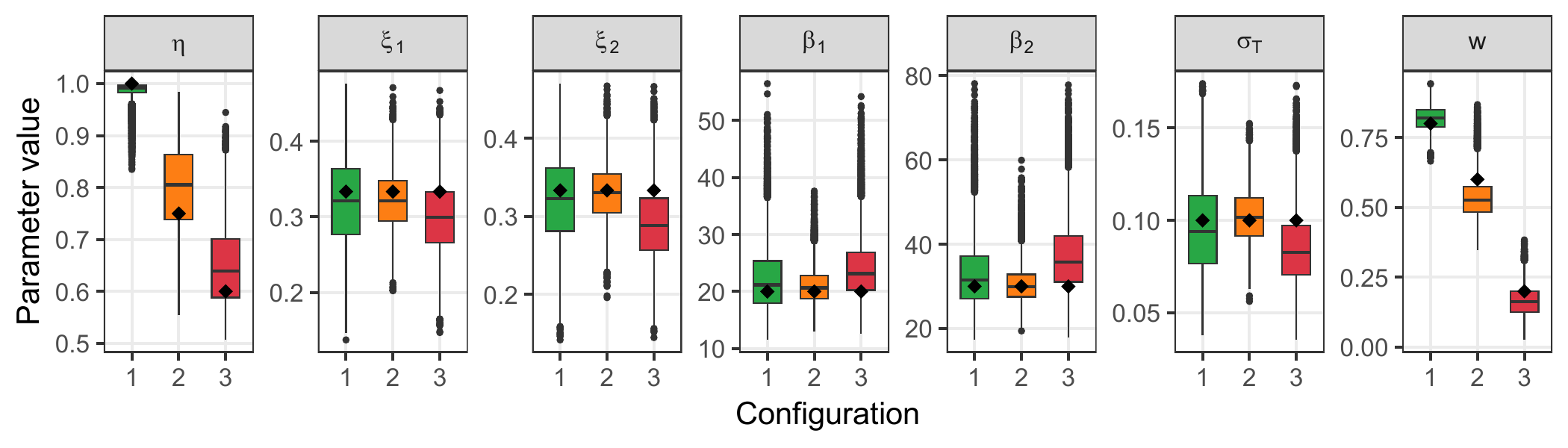}
    \caption{Boxplots for each parameter ($\eta, \xi_{1}, \xi_{2}, \beta_{1}, \beta_{2}, \sigma_{T}, w$) under three dependence settings: $\vthe^{(1)}$ (green), $\vthe^{(2)}$ (orange), and $\vthe^{(3)}$ (red), as defined in~\eqref{eq:true_theta_param}, based on $K = 1000$ samples of size $N = 1000$. Diamond markers indicate true values.}
	\label{fig:fixed_params_boxplots_full}
\end{figure}

To assess whether estimation errors in the asymptotic dependence parameters propagate to sub-asymptotic dependence measures, we also examine recovery of the function $\chi(q)$ over the full range $q \in (0,1]$. For each fitted parameter vector $\hat{\vthe}^{(i)}_k$, we generate a synthetic sample of size $10^4$ and compute the corresponding empirical estimator $\hat{\chi}(q)$ defined in Equation~\eqref{eq:chi_esti}. Since no closed-form expression of $\chi(q)$ is available for the model, this simulation-based evaluation is used throughout the paper to approximate the dependence curves numerically. Reference curves are obtained from large samples of size $10^5$ generated under the true parameters $\vthe^{(i)}$. Figure~\ref{fig:fixed_params_chi_full} displays a random subset of 100 estimated curves for visual clarity. Across all three configurations, the estimated curves closely match the reference dependence structure, indicating that $\chi(q)$ is accurately recovered even in regimes where individual dependence parameters exhibit increased variability.

\paragraph{Confidence intervals.}
Uncertainty is quantified using a non-parametric bootstrap applied conditionally on each fitted model. From each of the $K$ datasets and for each parameter configuration, we generate $B = 200$ bootstrap resamples, refit the NBE, and recompute both parameter estimates and the associated $\chi(q)$ curves. This yields empirical $95\%$ confidence intervals for model parameters and associated dependence measures. Table~\ref{tab:ci_summary} reports empirical coverage probabilities and mean confidence interval widths for a representative subset of parameters across the three configurations.

Overall, the confidence intervals are well calibrated, with empirical coverage close to the nominal level for most parameters. In the first configuration, where $\eta = 1$, coverage for $\eta$ is not meaningful since the true value lies on the boundary of the parameter space and cannot be attained by the estimator (Remark~\ref{rem:bound_prior_error}); in this case, uncertainty is more appropriately assessed through $\chi$. In the remaining configurations, coverage for $\eta$ improves substantially. Interval widths increase as $w$ decreases, reflecting the higher uncertainty associated with weaker sub-asymptotic dependence (see Section~\ref{sec:subasymptotic}). The marginal and scale parameters $(\xi_1, \xi_2, \beta_1, \beta_2, \sigma_T)$ exhibit stable coverage across configurations, with interval widths adapting sensibly to the scale of each parameter. Finally, Table~\ref{tab:ci_summary} shows that empirical coverage of the $\chi(q)$ curves remains high across quantiles, indicating reliable uncertainty quantification for tail dependence measures even in settings where parameter-level intervals widen.

\begin{figure}[htbp]
 \centering 
 \includegraphics[width=0.98\textwidth]{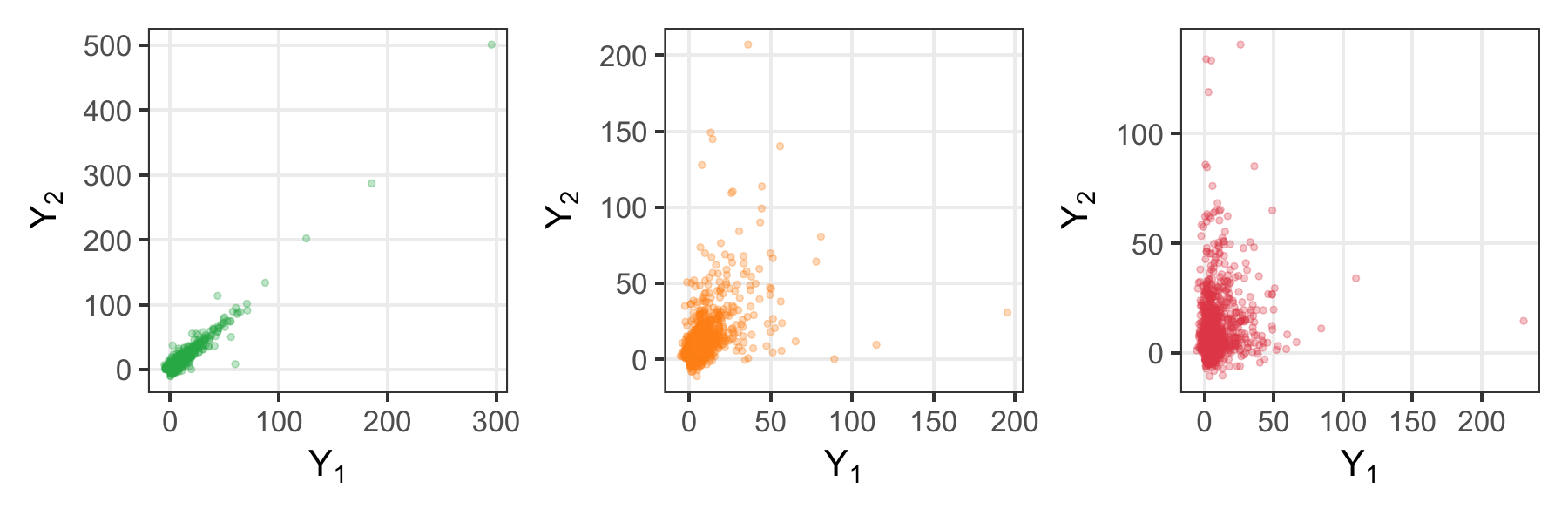}
 \includegraphics[width=0.98\textwidth]{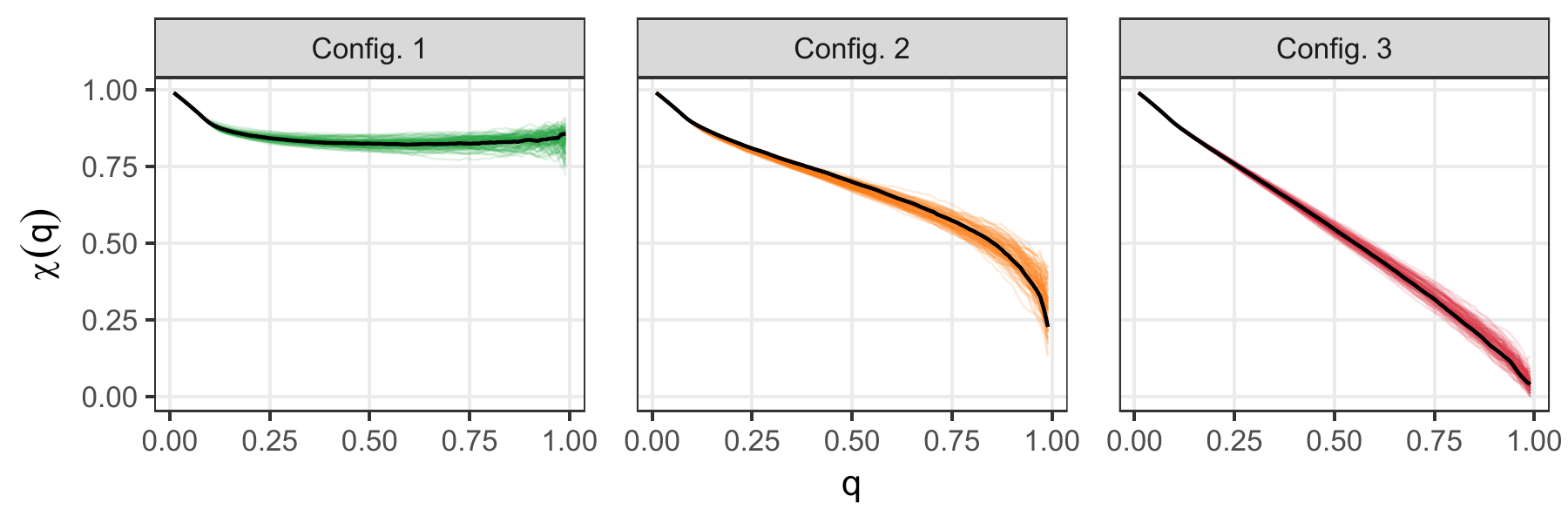}
 \caption{Top: samples from the sBGP distribution (Definition~\ref{def:biv_model}) with parameters $\vthe^{(1)}$ (green), $\vthe^{(2)}$ (orange), and $\vthe^{(3)}$ (red) (see~\eqref{eq:true_theta_param}). Bottom: estimated $\chi(q)$ curves. Solid lines correspond to reference curves empirically estimated from samples of size $10^5$ generated under the true parameters. Coloured lines show 100 estimated curves based on samples of size $10^4$ generated under fitted parameter values.}
 \label{fig:fixed_params_chi_full}
\end{figure}

\begin{table}[htbp]
\centering
\begin{tabular}{lcccccccccccc}
\toprule
 & \multicolumn{3}{c}{Config. 1} & \multicolumn{3}{c}{Config. 2} & \multicolumn{3}{c}{Config. 3} \\
\cmidrule(lr){2-4} \cmidrule(lr){5-7} \cmidrule(lr){8-10}
$\hat{\theta}_j$ & True & Cov. & Width & True & Cov. & Width & True & Cov. & Width \\
\midrule
$\hat{\eta}$        & 1.00 & NA   & 0.06 & 0.75 & 0.94 & 0.24 & 0.60 & 0.92 & 0.22 \\
$\hat{\xi}_{1}$     & 0.33 & 0.90 & 0.15 & 1.50 & 0.90 & 0.10 & 0.33 & 0.79 & 0.13 \\
$\hat{\xi}_{2}$     & 0.33 & 0.90 & 0.15 & 1.50 & 0.92 & 0.09 & 0.33 & 0.72 & 0.13 \\
$\hat{\beta}_{1}$   & 20.00 & 0.90 & 15.60 & 20.00 & 0.93 & 8.88 & 20.00 & 0.81 & 14.80 \\
$\hat{\beta}_{2}$   & 30.00 & 0.91 & 22.40 & 30.00 & 0.97 & 12.00 & 30.00 & 0.73 & 22.60 \\
$\hat{\sigma}_{T}$  & 0.10 & 0.90 & 0.07 & 0.10 & 0.99 & 0.04 & 0.10 & 0.80 & 0.06 \\
$\hat{w}$           & 0.80 & 0.92 & 0.12 & 0.60 & 0.76 & 0.21 & 0.20 & 0.81 & 0.15 \\

\midrule
$\chi(0.50)$ & 0.82 & 0.95 & 0.05 & 0.70 & 0.87 & 0.05 & 0.55 & 0.98 & 0.05 \\
$\chi(0.70)$ & 0.82 & 0.96 & 0.06 & 0.60 & 0.96 & 0.07 & 0.37 & 0.98 & 0.08 \\
$\chi(0.80)$ & 0.83 & 0.99 & 0.07 & 0.54 & 0.99 & 0.09 & 0.27 & 0.96 & 0.09 \\
$\chi(0.90)$ & 0.84 & 0.98 & 0.09 & 0.45 & 0.98 & 0.12 & 0.16 & 0.93 & 0.10 \\
$\chi(0.95)$ & 0.84 & 0.99 & 0.11 & 0.37 & 0.96 & 0.16 & 0.10 & 1.00 & 0.10 \\
$\chi(0.99)$ & 0.85 & 1.00 & 0.18 & 0.23 & 1.00 & 0.25 & 0.03 & 1.00 & 0.10 \\
\bottomrule
\end{tabular}
\caption{Empirical coverage probability (Cov.) and mean confidence interval width (Width) for model parameters and tail dependence curves $\chi(q)$ under three fixed configurations (see~\eqref{eq:true_theta_param}), based on $N=1000$, $K=1000$, $B=200$ bootstrap resamples per estimate, together with the true values.}
\label{tab:ci_summary}
\end{table}

All results presented in this section are based on the non-penalized estimator defined in~\eqref{eq:classical_loss}. The impact of the penalized version, as well as additional experiments based on randomly generated parameter configurations, are reported in the Supplementary Material.

	\section{Application to extreme precipitation in Belgium}\label{sec:application}
	
	We illustrate the capabilities of the sBGP model for extreme precipitation over Belgium, using gridded data from the ERA5 reanalysis dataset \citep{hersbach2020era5}. This dataset provides hourly precipitation estimates on a $0.25^\circ$ grid from 1950 to the present, derived from a global reanalysis of meteorological observations. Our aim is to assess the model's ability to capture a range of extremal dependence behaviours, from strong asymptotic dependence to near-independence, while also accurately representing marginal distributions. We use the penalized version of the neural Bayes estimator defined in \eqref{eq:pen_loss}. This penalization is optional and does not alter the underlying model, but stabilizes inference in practice; results based on the non-penalized estimator are given in the Supplementary Material.

	From the full temporal coverage (1950–-2024), we compute weekly maxima of daily precipitation (in millimetres) for all grid points in the rectangle $[0^\circ,7^\circ]$E $\times$ $[47^\circ,52^\circ]$N; this region is shown in Figure~\ref{fig:stations_squares_map} of the Supplementary material.  This block size is chosen to mitigate short-range temporal dependence. We restrict attention to the autumn season (21 September–21 December), which is associated with more homogeneous large-scale precipitation processes.    At each grid point, we extract threshold exceedances relative to the empirical $70^{\text{th}}$ percentile.	Other thresholds were also examined and led to qualitatively similar fits and interpretations; see Section~\ref{app:be_rainfall_pen} in the Supplementary Material.
	
	We start with a reference grid point centred at ($4.25^\circ$ E, $50.75^\circ$ N) that contains Brussels. We first consider all pairs formed between this reference location and grid points within the surrounding region, and fit the proposed bivariate model defined in Definition~\ref{def:biv_model} to each pair. This allows us to examine how the estimated model parameters vary across space. Next, we select three representative pairs within the region, and apply our model to evaluate the extremal dependence curve $\chi(q)$ and the marginal tail distributions. Finally, we compare the fit of our model with that of a multivariate GP distribution.

    \subsection{Spatial patterns of bivariate precipitation extremes}
    For each of the 506 pairs consisting of Brussels and another grid point in the region, we fit the sBGP model to the threshold exceedances. Figure~\ref{fig:spatial_pattern_param} illustrates several key parameters as spatial fields, i.e., each map reports, at a given location, the parameter estimate obtained for the pair between Brussels and that grid point. Maps for the remaining parameters are provided in Section~\ref{app:be_rainfall_pen} of the Supplementary Material. 

    In the following, the extremal dependence coefficient $\hat{\eta}$ is displayed only for locations where $\hat{w}>0.5$; as illustrated in Figure~\ref{fig:chi_curve_sub}, estimates of $\eta$ become unreliable when $w$ is small.
    
    The parameters $\hat{w}$ and $\hat{\sigma}_T$ vary coherently with distance from the reference grid point: fitted pairs tend to display weaker dependence as distance increases, with $\hat{w}$ decreasing and $\hat{\sigma}_T$ increasing.
    For locations with sufficiently large $\hat{w}$, the estimated extremal dependence coefficient $\hat{\eta}$ tends to decrease with distance from the reference grid point, reflecting a gradual weakening of extremal dependence.

    \begin{figure}[htbp]
		\centering
		\begin{minipage}{0.35\textwidth}
			\centering
			\includegraphics[width=\linewidth, trim={2cm 0 2cm 0},clip]{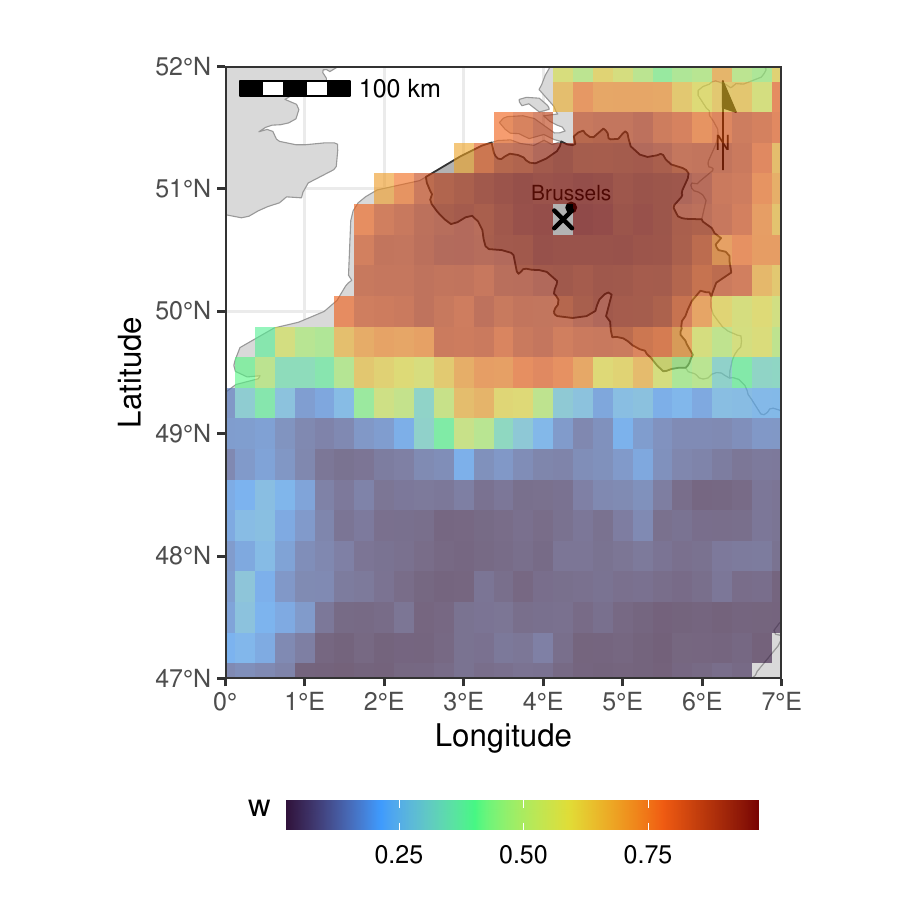}
            
		\end{minipage}\hfill
		\begin{minipage}{0.32\textwidth}
			\centering
			\includegraphics[width=\linewidth, trim={2.5cm 0 2.5cm 0},clip]{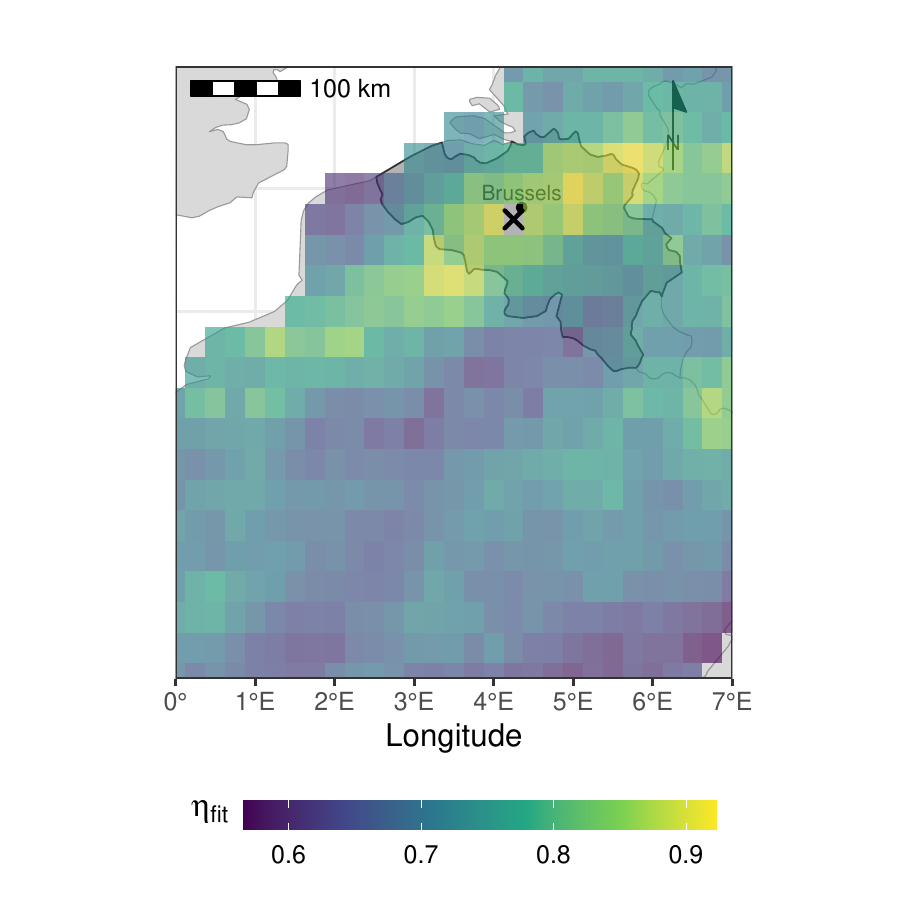}
		\end{minipage}\hfill
		\begin{minipage}{0.32\textwidth}
			\centering
			\includegraphics[width=\linewidth, trim={2.5cm 0 2.5cm 0},clip]{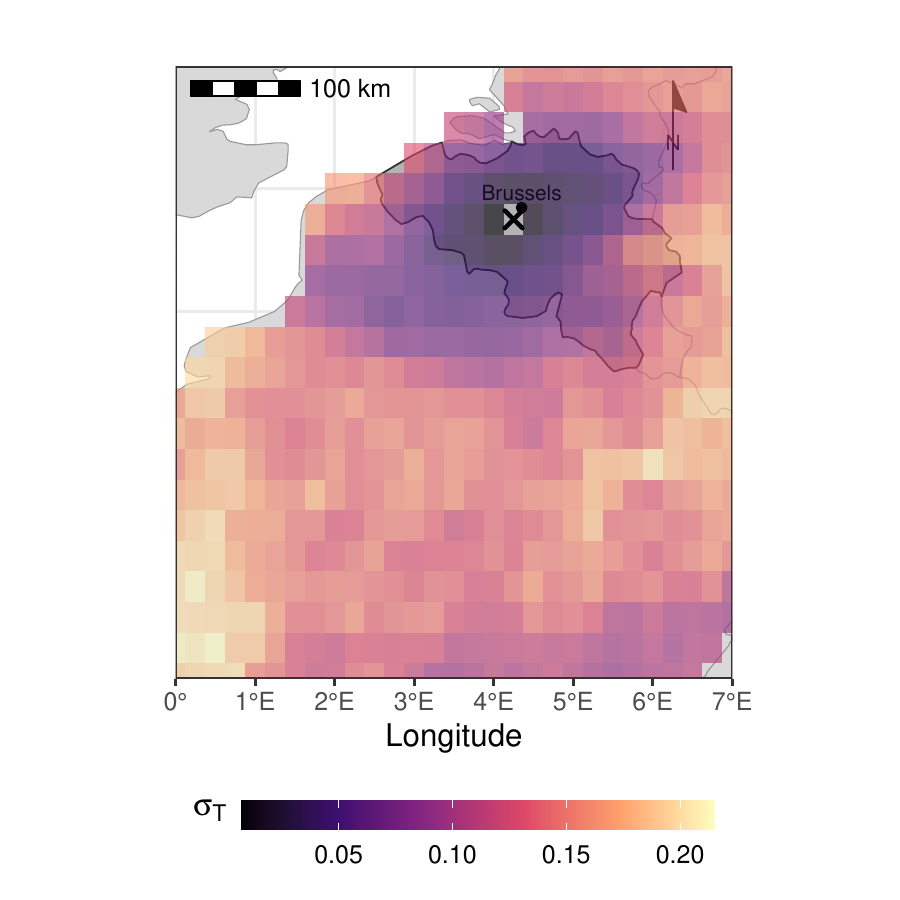}
		\end{minipage}		
		\caption{Spatial fields of estimated dependence parameters for bivariate precipitation extremes between Brussels (reference location, black cross) and each grid point. Each panel displays the parameter estimate obtained from fitting the sBGP model to threshold exceedances of the corresponding bivariate pair. From left to right: $\hat w$, $\hat\eta$, and $\hat\sigma_T$.}
		\label{fig:spatial_pattern_param}
	\end{figure}

\subsection{Diagnostics for selected pairs}
We select two pairs to illustrate different dependence regimes: (Brussels, Nivelles) and (Ostend, Spa). These pairs differ in distance between sites, inland versus coastal position, and regional topographic exposure. The selected grid points are shown in Section~\ref{app:be_rainfall_pen} of the Supplementary Material, with corresponding exceedances displayed in Figure~\ref{fig:joint_distrib}.

For all $t = 1, \ldots, T$, let
\[
\vc{Z}_t^{(i)} = 
\begin{cases}
(Z_t^{\text{Bru}}, Z_t^{\text{Niv}}), & i=1,\\
(Z_t^{\text{Ost}}, Z_t^{\text{Spa}}), & i=2,
\end{cases}
\]
denote the weekly rainfall observations. 
\begin{figure}[htbp]
  \centering
  \begin{minipage}{0.49\textwidth}
    \centering
    \includegraphics[width=\linewidth]{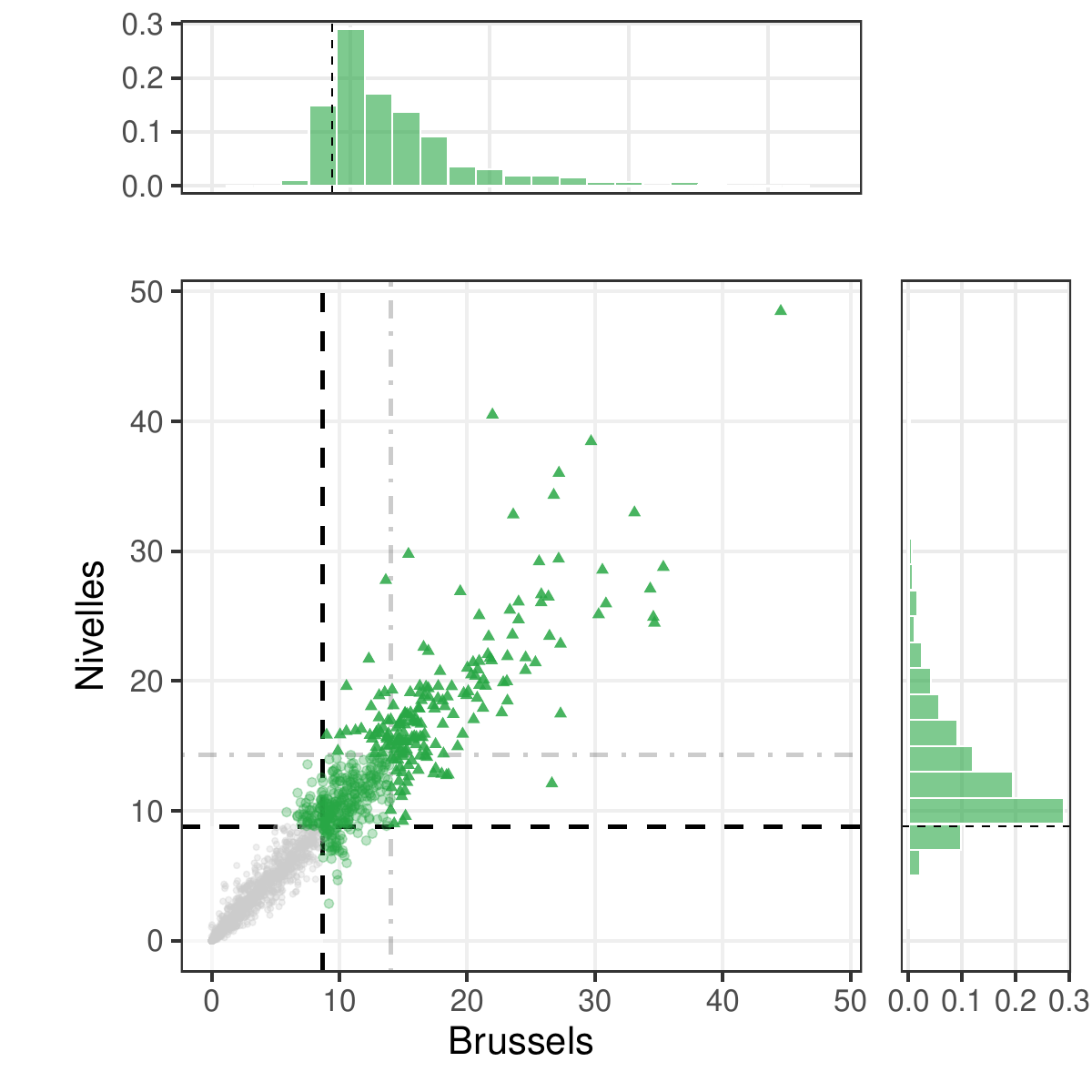}
  \end{minipage}\hfill
  \begin{minipage}{0.49\textwidth}
    \centering
    \includegraphics[width=\linewidth]{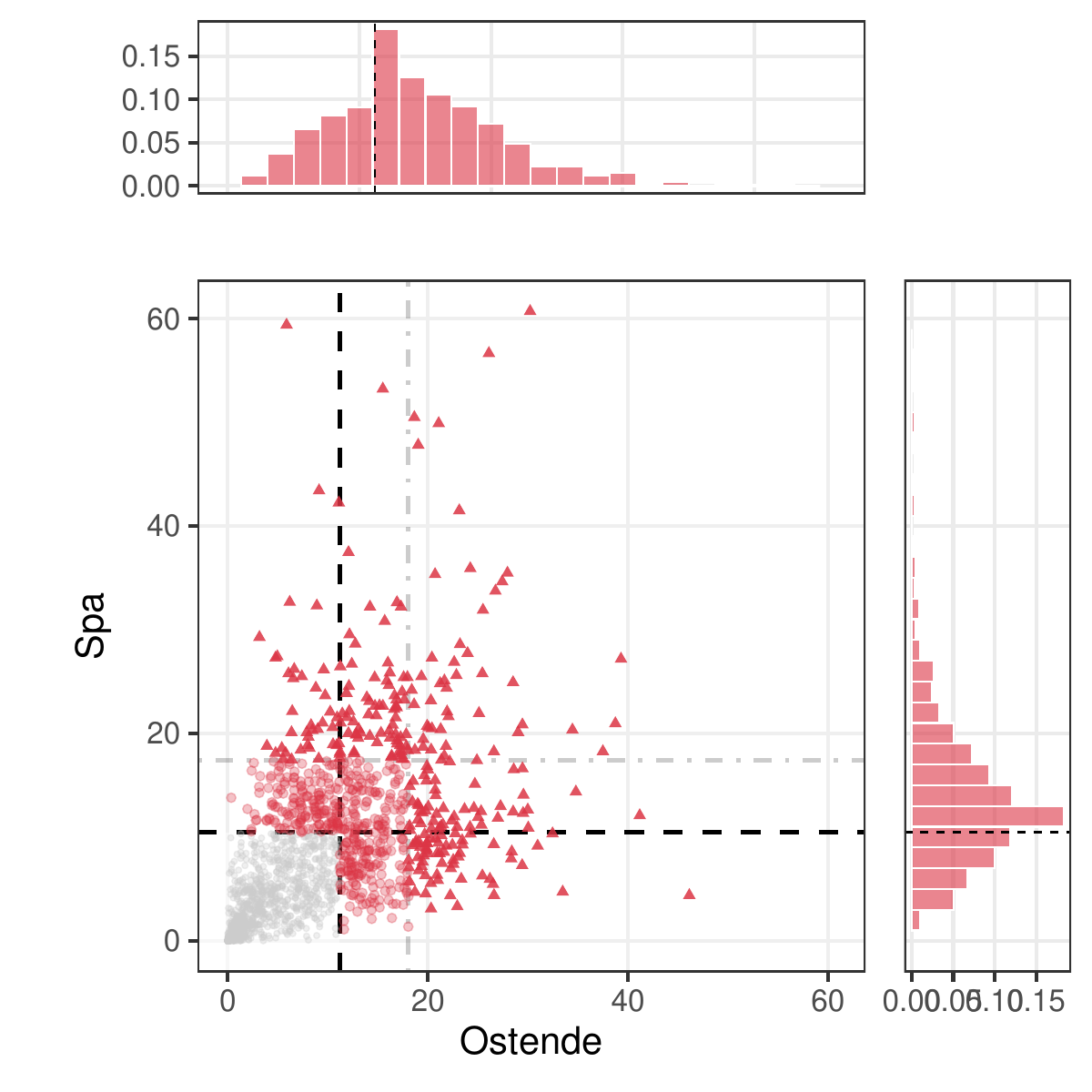}
  \end{minipage} 
  \caption{Bivariate scatterplots of threshold exceedances of weekly precipitation maxima for Brussels–Nivelles and Ostend–Spa. 
  Green points denote exceedances above the 70th percentile (see~\eqref{eq:Ztilde_def}), used for fitting the sBGP model, and green triangles denote exceedances above the 90th percentile, used for the comparison with the bivariate GP model in Section~\ref{sec:mgpd}. Black and gray dashed lines indicate the marginal 70th- and 90th-percentile thresholds, respectively. The marginal histograms display the distributions of marginal exceedances above the 70th percentile.}
  \label{fig:joint_distrib}
\end{figure}

Let
\begin{equation}\label{eq:Ztilde_def}
\begin{aligned}
\tilde{\mathcal{Z}}^{(1)}_{0.7} 
&= \{(Z_t^{\text{Bru}}, Z_t^{\text{Niv}}): 
\hat{F}_{\text{Bru}}(Z_t^{\text{Bru}}) > 0.7 
\;\text{or}\; 
\hat{F}_{\text{Niv}}(Z_t^{\text{Niv}}) > 0.7\}, \\
\tilde{\mathcal{Z}}^{(2)}_{0.7} 
&= \{(Z_t^{\text{Ost}}, Z_t^{\text{Spa}}): 
\hat{F}_{\text{Ost}}(Z_t^{\text{Ost}}) > 0.7 
\;\text{or}\; 
\hat{F}_{\text{Spa}}(Z_t^{\text{Spa}}) > 0.7\}.
\end{aligned}
\end{equation}
denote the subsets of exceedances for the pairs Brussels–Nivelles and Ostend–Spa, respectively. 
Here, $\hat{F}_j$ denotes the empirical cumulative distribution function of station $j$. At this threshold, the numbers of exceedances are 522 for Brussels–Nivelles and 671 for Ostend–Spa.

We fit the sBGP model to the exceedances of the two selected pairs, $\tilde{\mathcal{Z}}^{(1)}_{0.7}$ and $\tilde{\mathcal{Z}}^{(2)}_{0.7}$. Parameter estimates with 95\% confidence intervals are reported in Table~\ref{tab:pen_rep_pairs}. Marginal tail parameters $(\xi_j, \beta_j)$ vary moderately across locations. 
The dependence parameters, however, display clear contrasts between the two regimes: for Brussels–Nivelles, $\hat{\eta} = 0.85$ with $\hat{w} = 0.95$; for Ostend–Spa, $\hat{\eta} = 0.72$ with $\hat{w} = 0.16$. The shift parameter $\hat{\sigma}_T$ also differs substantially between the two pairs, from $0.03$ for the nearby inland pair to $0.14$ for the distant coastal--inland pair. 
Results for alternative thresholds are reported in the Supplementary Material.

\begin{table}[htbp]
\centering
\begin{tabular}{lcc}
\toprule
 & Brussels--Nivelles & Ostend--Spa \\
\midrule
$\eta$      & 0.85 (0.72, 0.93) & 0.72 (0.64, 0.83) \\
$\xi_1$     & 0.11 (0.09, 0.15) & 0.14 (0.11, 0.17) \\
$\xi_2$     & 0.10 (0.08, 0.14) & 0.14 (0.11, 0.18) \\
$\beta_1$   & 49.46 (35.32, 62.48) & 41.68 (28.32, 56.27) \\
$\beta_2$   & 61.54 (39.39, 73.06) & 46.79 (35.46, 60.85) \\
$\sigma_T$  & 0.03 (0.02, 0.04) & 0.14 (0.10, 0.20) \\
$w$         & 0.95 (0.91, 0.96) & 0.16 (0.08, 0.31) \\
\bottomrule
\end{tabular}
\caption{
Parameter estimates (with 95\% confidence intervals) for the sBGP model (Definition~\ref{def:biv_model}) fitted to the exceedance sets $\tilde{\mathcal{Z}}^{(1)}_{0.7}$ and $\tilde{\mathcal{Z}}^{(2)}_{0.7}$ defined above.
}
\label{tab:pen_rep_pairs}
\end{table}

Overall, the Brussels–Nivelles pair displays pronounced sub-asymptotic dependence ($w \approx 0.95$), while the Ostend–Spa pair shows much weaker sub-asymptotic dependence ($w \approx 0.16$). The latter pair also has a smaller estimated residual tail dependence coefficient. However, as explained in Section~\ref{sec:subasymptotic}, when $w$ is small, the focus should be on sub-asymptotic measures such as $\chi(q)$ or $\eta(q)$.

All dependence diagnostics in this subsection, including $\chi(q)$ curves, are computed from exceedances in $\tilde{\mathcal{Z}}^{(1)}_{0.7}$ and $\tilde{\mathcal{Z}}^{(2)}_{0.7}$,  so that the quantiles $q$ refer to empirical quantiles within each exceedance subset rather than the full sample.  
Figure~\ref{fig:pen_p_chi} presents the estimated $\hat{\chi}_{\tilde{\mathcal{Z}}_{0.7}}(q)$ curves, defined in Equation~\eqref{eq:chi_esti}, obtained from $B=100$ nonparametric bootstrap replications of the exceedance data. For each bootstrap resample, model parameters were re-estimated, yielding a bootstrap sample of parameter values, and the associated $\chi(q)$ curves were evaluated numerically using the same simulation-based procedure as in the simulation study.

Brussels–Nivelles exhibits strong dependence up to around the 0.9 quantile of the exceedance subset $\tilde{\mathcal{Z}}^{(1)}_{0.7}$, after which the $\chi(q)$ curve drops abruptly, revealing a weakening of dependence in the upper tail. Ostend–Spa shows a marked decline at the most extreme quantiles, indicating near-independence in the tail. These distinct patterns are well reproduced by the fitted model. 

For both pairs, the $\chi_{\tilde{\mathcal{Z}}_{0.7}}(q)$ curves display visible discontinuities at low quantiles. This feature reflects the censoring inherent in the L-shaped support of the model, where at least one variable remains below its threshold. The transition point corresponds to the proportion of observations that are not simultaneously extreme in both margins. A detailed discussion of this mechanism is provided in Section~\ref{sec:subasymptotic}.

\begin{figure}[htbp]
  \centering
  \includegraphics[width=0.95\textwidth]{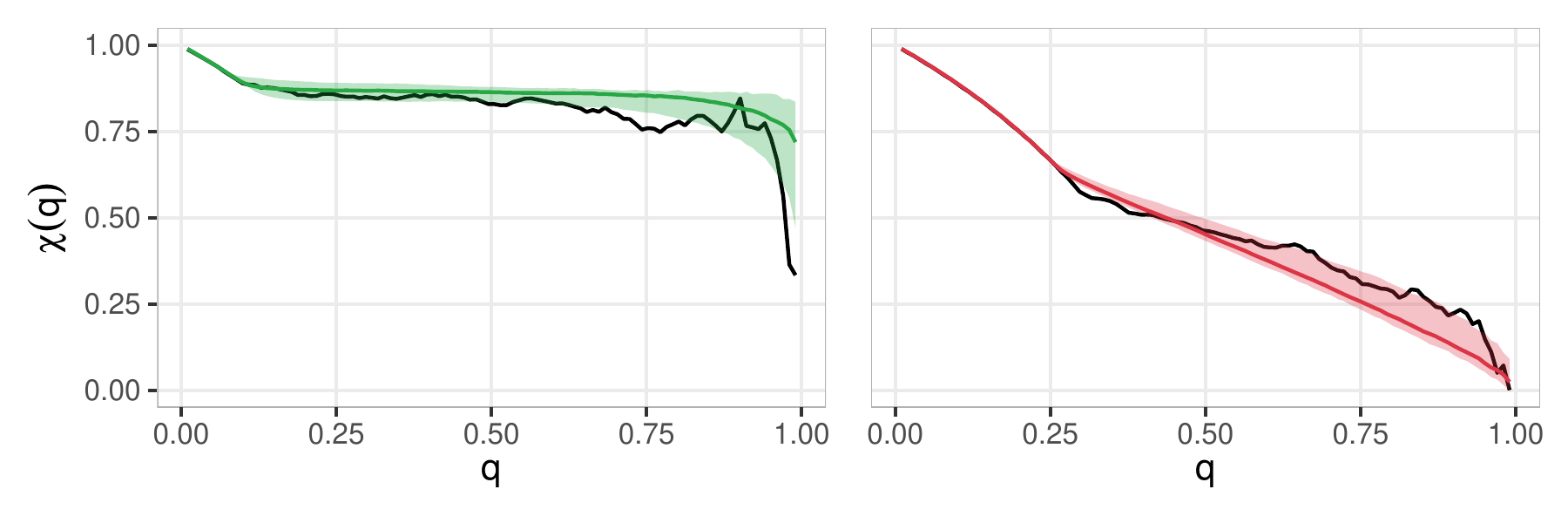}
  \caption{Estimated $\chi_{\tilde{\mathcal{Z}}_{0.7}}(q)$ curves for Brussels–Nivelles and Ostend–Spa, where $\chi(q)$ is estimated using~\eqref{eq:chi_esti}.  Quantiles $q$ are defined relative to the exceedance subset $\tilde{\mathcal{Z}}_{0.7}$ (i.e., within the exceedance region, $u_j = F_j^{-1}(0.7)$).  The black line shows the empirical $\chi_{\tilde{\mathcal{Z}}_{0.7}}(q)$ from observed exceedances; the green and red lines are the model estimates; the shaded area denotes the 95\% bootstrap confidence region.}
  \label{fig:pen_p_chi}
\end{figure}

Marginal adequacy is assessed via QQ-plots (Figure~\ref{fig:pen_qq_marginals_main}) comparing the empirical exceedance distributions in $\tilde{\mathcal{Z}}^{(1)}_{0.7}$ and $\tilde{\mathcal{Z}}^{(2)}_{0.7}$ to their corresponding fitted marginal models. Overall, deviations from the identity line are minor, indicating satisfactory marginal fits. In particular, the model reproduces well the upper-tail quantiles, which are of primary importance in applications.

\begin{figure}[htbp]
  \centering
  \includegraphics[width=\textwidth]{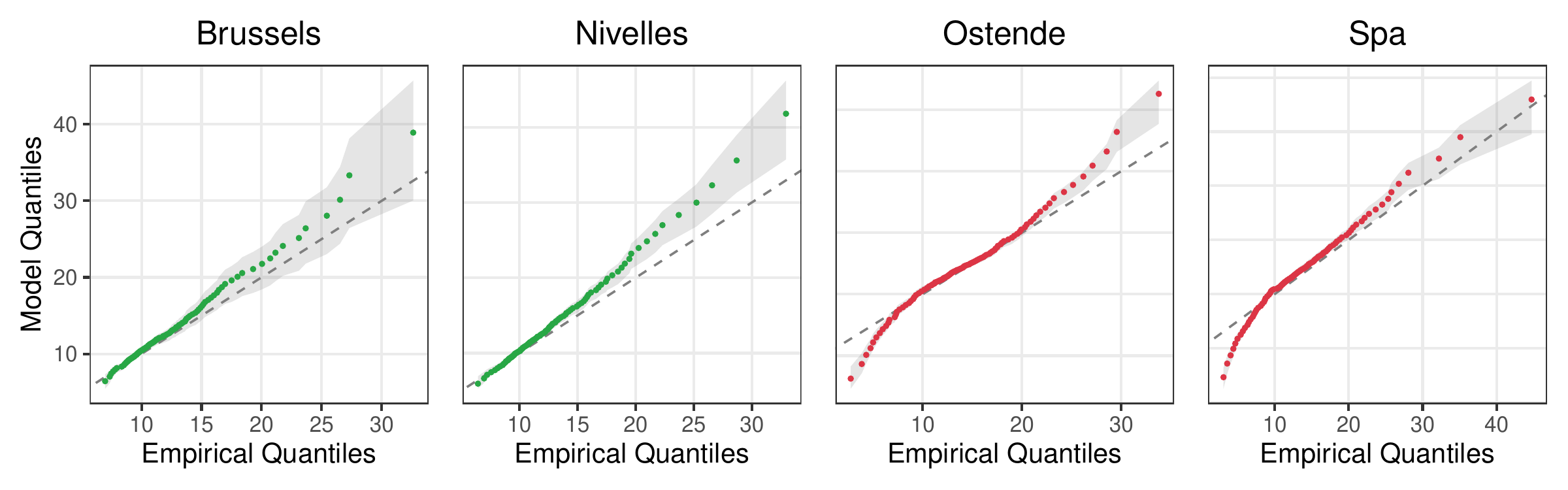}
  \caption{QQ-plots comparing empirical marginal exceedances to fitted marginal distributions for Brussels–Nivelles and Ostend–Spa under the sBGP model (Definition~\ref{def:biv_model}). 
  Each panel corresponds to one margin of the respective pair. 
  All plots are based on exceedances above the marginal threshold ($u_j = F_j^{-1}(0.7)$), with dashed vertical lines marking the threshold values. 
  Theoretical quantiles are shown on the original data scale, obtained by adding the threshold to the modeled exceedances.}
  \label{fig:pen_qq_marginals_main}
\end{figure}

In addition to the QQ-plots, Figure~\ref{fig:density_marginals_main} compares the empirical marginal densities with their fitted counterparts. The estimated model densities, computed from Equation~\eqref{eq:marg_Y_density}, closely follow the empirical histograms, confirming that the marginal fits adequately reproduce both the bulk and the upper tail of the data.

\begin{figure}[htbp]
  \centering
  \begin{minipage}{0.25\textwidth}
    \centering
    \includegraphics[width=\linewidth]{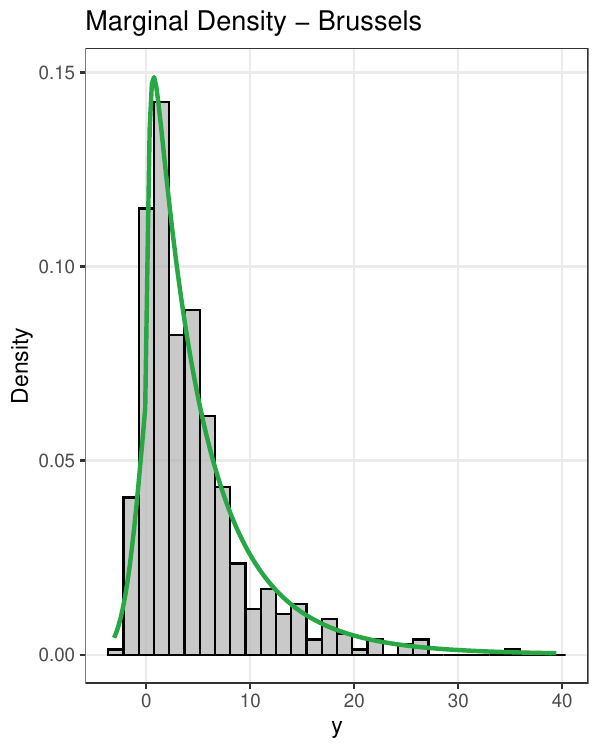}
    \caption*{Bruxelles}
  \end{minipage}\hfill
  \begin{minipage}{0.25\textwidth}
    \centering
    \includegraphics[width=\linewidth]{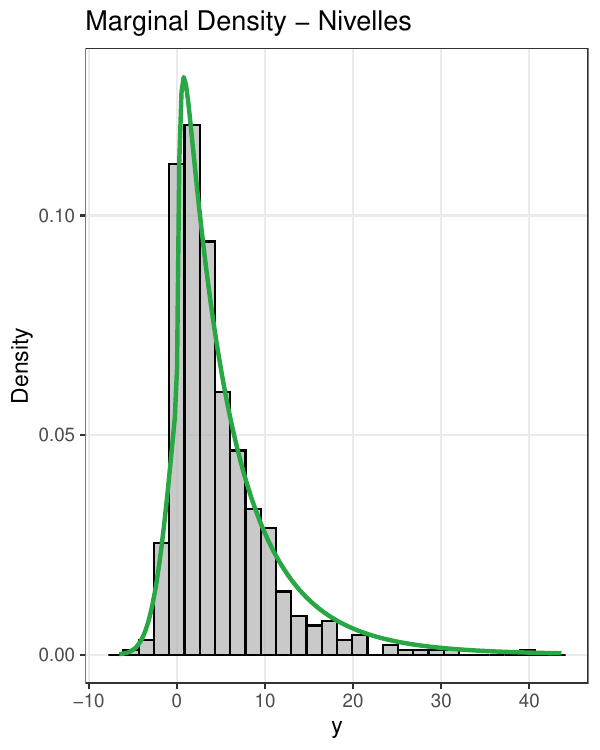}
    \caption*{Nivelles}
  \end{minipage}\hfill
  \begin{minipage}{0.25\textwidth}
    \centering
    \includegraphics[width=\linewidth]{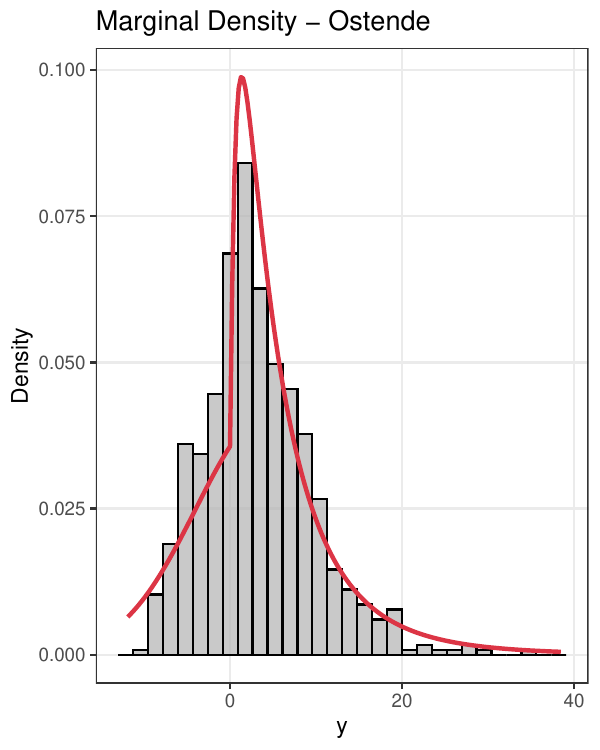}
    \caption*{Ostende}
  \end{minipage}\hfill
  \begin{minipage}{0.25\textwidth}
    \centering
    \includegraphics[width=\linewidth]{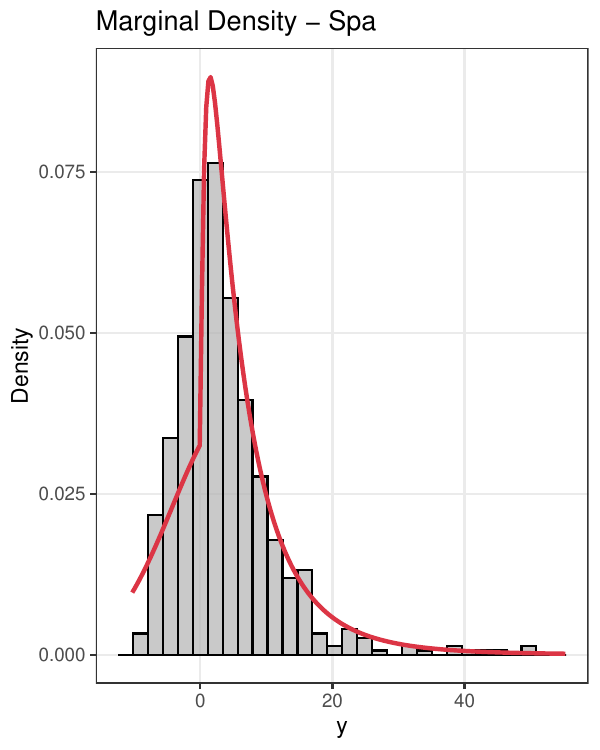}
    \caption*{Spa}
  \end{minipage}
  \caption{Empirical histograms (gray bars) and fitted marginal densities (solid lines) for Brussels–Nivelles and Ostend–Spa. Each panel corresponds to one margin of the respective pair. Model-based densities are computed from the fitted sBGP model (Definition~\ref{def:biv_model}) according to Equation~\eqref{eq:marg_Y_density}, using exceedances above the marginal threshold ($u_j = F_j^{-1}(0.7)$).}
  \label{fig:density_marginals_main}
\end{figure}

\subsection{Comparison with the bivariate GP distribution}

To benchmark the proposed approach, we compare it with the bivariate GP distribution (see Section~\ref{sec:mgpd}), the standard asymptotic model for threshold exceedances under asymptotic dependence. In both models, the shift (spectral) vector $\mathbf{S}$ is defined as $\bm{S}=\max(\bm{T})-\bm{T}$. In our model, $\bm{T}$ is specified as in~\eqref{eq:S_T_definition}, whereas in the bivariate GP model it is specified through an independent Gumbel generator with common shape and scale parameters, i.e.
\[
T_1, T_2 \stackrel{\text{iid}}{\sim} \mathrm{Gumbel}(a_T, b_T);
\]
other choices of $\bm T$ led to similar results.

The bivariate GP distribution is threshold-stable and theoretically valid only at sufficiently high quantiles for asymptotically dependent vectors; we therefore restrict the comparison to the pair Brussels–Nivelles and fit both models using exceedances above the empirical 90th percentile \citep{kiriliouk_peaks_2019}. The corresponding exceedance subset is defined as
\[
\tilde{\mathcal{Z}}^{(1)}_{0.9} 
= 
\{(Z_t^{\text{Bru}}, Z_t^{\text{Niv}}): \hat{F}_{\text{Bru}}(Z_t^{\text{Bru}}) > 0.9 \text{ or } \hat{F}_{\text{Niv}}(Z_t^{\text{Niv}}) > 0.9\},
\]
that is, the set of observations where at least one margin exceeds its 90th-percentile threshold. 
At this higher threshold, the number of exceedances for the Brussels–Nivelles pair is 187.

\paragraph{Results.}

Table~\ref{tab:params_mgpd_full} reports the parameter estimates for both models fitted to the Brussels–Nivelles pair. For the proposed model, the estimated dependence and tail parameters are $\hat{\eta}=0.66$ and $\hat{w}=0.94$, together with moderate marginal tail indices ($\hat{\xi}_1=0.13$, $\hat{\xi}_2=0.13$). The corresponding scale parameters are $\hat{\beta}_1=37.8$ and $\hat{\beta}_2=36.6$, while the shift variability parameter $\hat{\sigma}_T=0.06$ indicates a moderate level of sub-asymptotic variability. For the GP distribution, parameter estimates are reported for comparison and exhibit slightly larger tail indices together with a strong dependence parameter $\hat{\alpha } = 2.4133$, consistent with its asymptotic dependence structure.

Confidence intervals were obtained via bootstrap procedures: a nonparametric bootstrap was used for the proposed model, whereas a parametric bootstrap was employed for the GP distribution. In the latter case, bootstrap samples were generated from the fitted distribution and refitted to obtain the empirical distribution of the estimators. The $\chi(q)$ curves for the GP distribution were then evaluated numerically using the same simulation-based procedure as described in the simulation study, from which 95\% confidence bands were derived.

Figure~\ref{fig:p_marg_chi_models_compare} compares the estimated $\hat{\chi}_{\tilde{\mathcal{Z}}^{(1)}_{0.9}}(q)$ curves and the marginal QQ plots for both models. While the proposed model reproduces the empirical marginal distributions satisfactorily, the GP distribution exhibits discrepancies in the marginal QQ plots, indicating a poorer fit. This behaviour is more likely attributable to structural model misspecification rather than estimation issues: the GP model relies on threshold stability, an assumption that is not supported by the data in this case. When threshold stability is violated, fitting the GP distribution at a fixed level may induce biases in the parameter estimates and, in this case, an overestimation of the marginal tail indices $\xi_j$, resulting in overly heavy estimated tails and a degraded marginal fit. Regarding extremal dependence, the proposed model closely follows the empirical $\chi(q)$ curve across most quantiles, capturing the gradual decrease observed beyond $q\approx0.9$. In contrast, the GP model yields an almost constant $\chi(q)$ function, as expected under asymptotic dependence, but is less able to reflect the observed weakening of dependence.

\begin{table}[htbp]
\centering
\begin{tabular}{lcc}
\toprule
 & \textbf{Our model} & \textbf{MGPD} \\
\midrule
$\eta$        & 0.66 (0.57, 0.82) & --- \\
$\xi_1$       & 0.15 (0.09, 0.22) & 0.26 (0.09, 0.40) \\
$\xi_2$       & 0.14 (0.08, 0.22) & 0.27 (0.11, 0.41) \\
$\beta_1^*$   & 6.70 (4.98, 7.79) & 4.37 (3.69, 5.43) \\
$\beta_2^*$   & 6.28 (5.21, 7.51) & 4.29 (3.36, 5.16) \\
$\sigma_T$    & 0.07 (0.04, 0.13) & --- \\
$w$           & 0.95 (0.87, 0.98) & --- \\
$a_T$         & --- & 2.41 (1.97, 2.91) \\
$b_T$         & --- & -0.02 (-0.18, 0.08) \\
\bottomrule
\end{tabular}
\caption{Parameter estimates with 95\% confidence intervals for the sBGP and the bivariate GP model, fitted to the pair Brussels–Nivelles using exceedances above the 90th percentile. Parameters that are not defined for a given model are indicated by "---". For the proposed model, $\beta_j^*$ denotes the rescaled version of $\beta_j$, obtained by dividing by the expectation of the first latent component (see Eq.~\eqref{eq:expec_Y}), in order to ensure comparability with the MGPD parametrisation. For the MGPD, $\beta_j^*$ corresponds to the standard scale parameter.}
\label{tab:params_mgpd_full}
\end{table}

\begin{figure}[htbp]
  \centering
  \begin{subfigure}[b]{0.34\textwidth}
      \includegraphics[width=\linewidth]{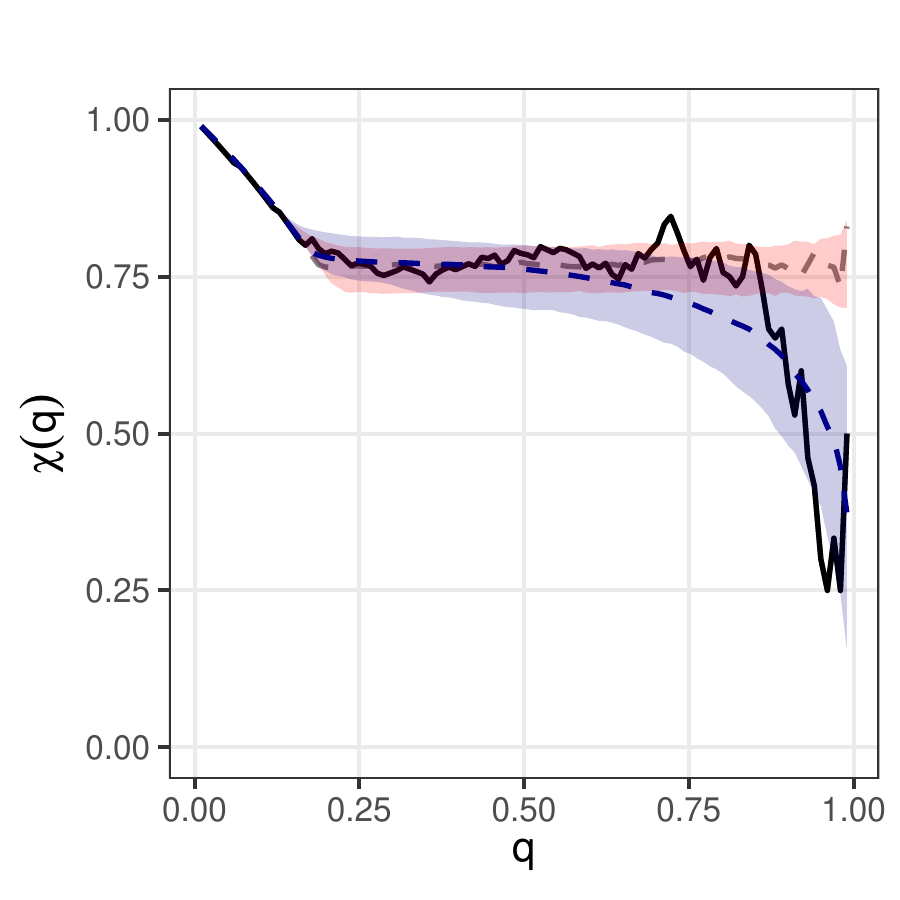}
    \caption{$\chi(q)$ curves.}
    \label{fig:chi_pair1_mgpd_with_ci}
  \end{subfigure}
  \hfill
  \begin{subfigure}[b]{0.64\textwidth}
    \includegraphics[width=\linewidth]{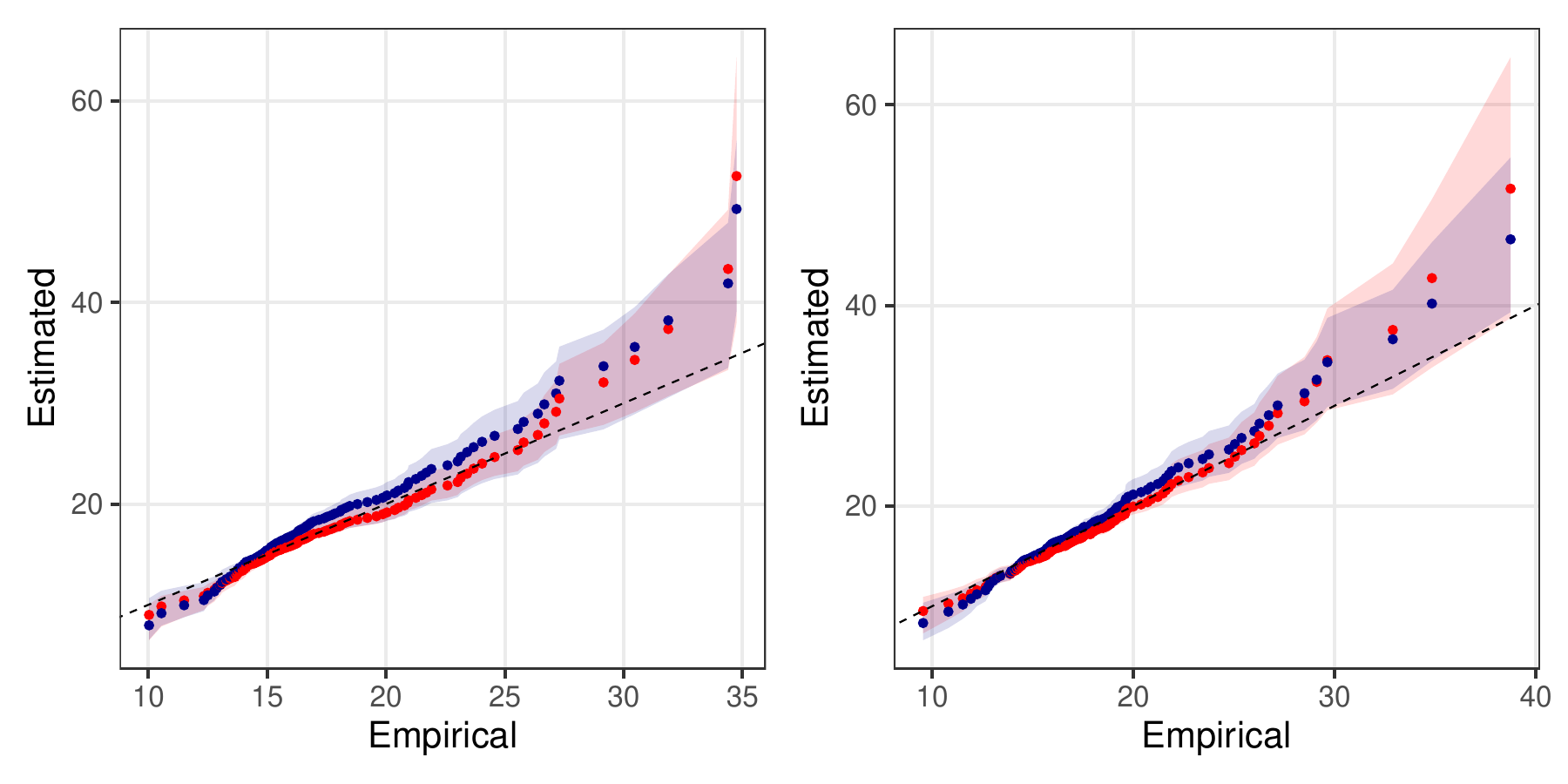}
    \caption{Marginal QQ plots.}
    \label{fig:qq_pair1_combined}
  \end{subfigure}
  \caption{Comparison between the sBGP and the bivariate GP model for Brussels–Nivelles.
  (a) Estimated $\chi_{\tilde{\mathcal{Z}}^{(1)}_{0.9}}(q)$ curves and (b) marginal QQ plots. 
  Blue lines and shaded regions correspond to the proposed model, and red lines to the MGPD. 
  Dashed lines represent model-based estimates, solid lines the empirical values, and shaded areas indicate 95\% bootstrap confidence regions. 
  Both models were fitted to exceedances $\tilde{\mathcal{Z}}^{(1)}_{0.9}$ (threshold $u_j = F_j^{-1}(0.9)$).}
  \label{fig:p_marg_chi_models_compare}
\end{figure}


\section{Discussion}
\label{sec:discussion}

We introduced a flexible sub-asymptotic model for bivariate threshold exceedances that extends classical generalized Pareto constructions. The proposed model accommodates both asymptotic dependence and asymptotic independence within a unified framework while preserving margins that converge to generalized Pareto tails. By explicitly modelling the sub-asymptotic regime, the model captures a wide range of dependence behaviours at practically relevant thresholds. Inference is performed through a neural Bayes estimator trained on simulated data, providing an efficient likelihood-free estimation strategy. Simulation experiments and the rainfall application illustrate that the model captures a broad range of extremal dependence behaviours while remaining interpretable and computationally tractable.

In theory, the construction admits a natural extension to dimensions $d > 2$. Let $E_J \sim \Exp(1)$ and $G_J \sim \Gamma(\alpha_J,1)$ for each subset $J \subseteq D = \{1,\ldots,d\}$ with $|J| \ge 1$, all mutually independent. For $j \in D$, let
\[
S_j = \max(T_1,\ldots,T_d) - T_j, \quad (T_1,\ldots,T_d) \sim \mathcal{N}_d(0,\sigma I_d),
\]
we may define the components of a $d$-variate sub-asymptotic GP vector $\vY=(Y_1,\ldots,Y_d)$  as
\[
Y_j = \beta_j \left( \frac{\sum_{J \subseteq D: j\in J} w_J E_J}{\sum_{J \subseteq D: j\in J} G_J} - S_j \right),
\]
where $w_J \ge 0$ and $\sum_{J \subseteq D: j\in J} w_J = 1$. In practice, however, the full model involves many parameters and may suffer from identifiability issues, suggesting that parsimonious specifications are preferable. Hence, future work will focus on extending the sBGP distributions to $d > 2$ while retaining a limited number of parameters.
	

	\section{Proofs}\label{app:proofs}

    This section contains proofs of all main results stated in the main paper. Some of these rely on additional lemmas whose proofs have been deferred to the Supplementary Material.


	  Recall that a function $f: (0, \infty) \to (0, \infty)$ is \emph{regularly varying} (at infinity) of order $\rho \in \RR$, written $f \in \RV (\rho)$, if for all $x > 0$,
    \[
    \lim_{t \rightarrow \infty} \frac{f(tx)}{f(t)} = x^{\rho}.
    \]
     A cumulative distribution function (cdf) $F$ with upper endpoint $x^* = \infty$ has regularly varying (upper) tail of order $\rho \geq 0$ if the survival function $\overline{F}(x) = 1 - F(x)$ is in $\RV (-\rho)$. A regularly varying $F$ is necessarily \emph{heavy-tailed}; for a thorough inventory of upper tail behaviour classes, see \cite{engelke_extremal_2019}.
	
	\subsection{Marginal behaviour}\label{sec:marginal_proof}

    Recall the model in Definition~\ref{def:biv_model}. Write
    \[
    Y_j = \beta_j (V_j - S_j), \qquad \text{where } V_j = \frac{w E + (1-w) E_j}{G + G_j} \qquad j = 1,2,
    \]
    and $\xi_j = (\alpha + \alpha_j)^{-1}$.
    Below, Lemma~\ref{lem:marg_V} derives the survival function and density of \(V_j\) and its first two moments, and Lemma~\ref{lem:asymV} gives an asymptotic expansion of the survival function of $V_j$. These are used in the proof of Proposition~\ref{prop:marg_Y}. 
	
	\begin{lem}[Distribution of \(V_j\)]\label{lem:marg_V}
		\text{ }
		\begin{enumerate}
			\item For $x \geq 0$ and $w \in [0,1]$, 
			\[
			\overline F_{V_j}(x)
			= 
             \begin{dcases}
             \frac{w(1+x/w)^{-1/\xi_j} - (1-w)(1+x/(1-w))^{-1/\xi_j}}
			{2w-1}, & w \in (0,1)\setminus\{\tfrac12\}, \\
                   (1 + 2x)^{-1/\xi_j - 1} \rbr{ 1 + 2x (1 + \xi^{-1})}, & w = \tfrac12. \\
            \end{dcases}
			\]
            
            The corresponding density is
            \begin{equation*}
    f_{V_j}(x) = \begin{dcases}
        \frac{\1\{x\ge0\}}{\xi_j(2w-1)}
            \left[\rbr{1+\frac{x}{w} }^{-1/\xi_j-1}-\rbr{1+\frac{x}{1-w}}^{-1/\xi_j-1}\right], & w \in (0,1)\setminus\{\tfrac12\}, \\
                  \1\{x\ge0\} \frac{4x}{\xi_j} \rbr{ 1 + \frac{1}{\xi_j}} \rbr{1 + 2x}^{-1/\xi_j - 2},  & w = \tfrac12. \\
    \end{dcases}
    \end{equation*}

			\item  The expected value and the variance of \(V_j\) are
            \[
    \begin{dcases}
    \EE[V_j] = \frac{\xi_j}{1-\xi_j}, & 0 < \xi_j < 1, \\
    \textnormal{Var}[V_j]
    = \frac{w^2+(1-w)^2 +2\xi_j w(1-w)}{(\xi_j^{-1}-1)^2(1-2\xi_j)}, & 0 < \xi_j < \tfrac 12.
    \end{dcases}
    \]
		\end{enumerate}
	\end{lem}

\begin{lem}[Asymptotic expansion of $\overline F_{V_j}$]\label{lem:asymV}
     For $x \rightarrow \infty$ and $w \in [0,1]$,
			\begin{equation}\label{eq:Vasym}
				\overline F_{V_j} (x) = C x^{-1/\xi_j} \left(1 - \frac{D}{\xi_j x} + o(x^{-1}) \right),
			\end{equation}
			where, for $w \in [0,1] \setminus \tfrac 12$,
			\begin{equation}\label{eq:CandD}
				C = \frac{w^{1/\xi_j+1}-(1-w)^{1/\xi_j+1}}{2w-1} \quad \text{ and } \quad  D= \frac{w^{1/\xi_j+2}-(1-w)^{1/\xi_j+2}}
				{w^{1/\xi_j+1}-(1-w)^{1/\xi_j+1}}.
			\end{equation}
			If $w \in \{0,1\}$, $C = D = 1$. Expressions for $w = 1/2$ are given in the proof.
\end{lem}


	   \begin{proof}[\normalfont\bfseries Proof of Proposition~\ref{prop:marg_Y}]
       \textbf{ }
       \begin{description}
           \item[\ref{prop:marg_Y_density}] We state the results for $j = 1$ without loss of generality. Recall that $S_1 = \max(T_1,T_2) - T_1 = \max(0,T_2 - T_1)$ so that $\PP[S_1 = 0] = 1/2$ for $T_1,T_2 \siid N(0, \sigma^2_T)$ as in~\eqref{eq:S_T_definition}. We have
\begin{align*}
    F_{Y_1} ( y) & = \frac 12 \cbr{ \PP[Y_1 \leq y \mid S_1 = 0] +  \PP[Y_1 \leq y \mid S_1 > 0]} \\
    & = \frac 12 \cbr{ F_{V_1}( y/\beta_1) +  \EE[ F_{V_1} ( y/\beta_1 + S_1) \mid S_1 > 0]} \\
    & = \frac 12 \cbr{ F_{V_1}( y/\beta_1) + \int_{0}^{\infty} F_{V_1}(y/\beta_1 + s) f_{S_1 \mid S_1 > 0} (s) \diff s}
\end{align*}
Differentiating with respect to $y$ given the density expression in \eqref{eq:marg_Y_density}. In addition, 
\begin{align}\label{eq:Sconddist}
    \PP[S_1 \leq s \mid S_1 > 0] = \frac{\PP[0 \leq T_2 - T_1 \leq s]}{\PP[T_2 > T_1]} = 2 \Phi \rbr{\frac{s}{\sqrt{2} \sigma_T}} - 1,
\end{align}
which is a half-normal distribution (i.e. folded normal with mean zero) with scale $\sqrt{2} \sigma_T$.  
           \item[\ref{prop:marg_Y_moments}] The mean and variance of $V_j$ are given in Lemma~\ref{lem:marg_V}. If $T_1,T_2 \siid N(0, \sigma^2_T)$, then 
           \[
           \EE[S_j] = \frac{\sigma_T}{\sqrt{\pi}}, \qquad  \textnormal{Var}[S_j] = (1 - 1/\pi) \sigma_T^2
           \]
           \item[\ref{prop:marg_Y_tail}]   First, note that 
		\[
 \lim_{x \rightarrow \infty} \frac{\overline F_{V_j-S_j}(x)}{\overline F_{V_j}(x)} =   \lim_{x \rightarrow \infty} \EE \left[\frac{\overline F_{V_j}(x + S_j)}{\overline F_{V_j}(x)} \right] = \EE \left[ \lim_{x \rightarrow \infty} \frac{\overline F_{V_j}(x + S_j)}{\overline F_{V_j}(x)} \right],
		\]
		where the last step follows from dominated convergence. Using expansion \eqref{eq:Vasym}, 
		we get, for $s > 0$ and $x \rightarrow \infty$,
		\begin{align*}
			\frac{\overline F_{V_j}(x + s)}{\overline F_{V_j}(x)} & =  \left( 1 + \frac{s}{x} \right)^{-1/\xi_j} \left( \dfrac{1- \frac{D}{\xi_j(x+s)} + o(x^{-1})}{1- \frac{D}{\xi_j x} + o(x^{-1})} \right) \\
			& = \left( 1 + \frac{s}{x} \right)^{-1/\xi_j}( 1 + o(x^{-1})) = 1 - \frac{s}{\xi_j x} + o(x^{-1}),
		\end{align*}
		where the last equality is based on expansion \eqref{eq:approx} in the Supplementary Material, and 
		\begin{equation}\label{eq:part1}
		\frac{\overline F_{Y_j}(x)}{\overline F_{\beta_j V_j}(x)} = \frac{\overline F_{V_j - S_j}(x/\beta_j)}{\overline F_{V_j}(x/\beta_j)} =   1 - \frac{\EE[S_j]}{ (\xi_j/\beta_j) x} + o(x^{-1}), \qquad x \rightarrow \infty.
		\end{equation}  
		
		Second, if \(Z\sim\GP(\xi_j,\sigma_j)\) with \(\sigma_j= \beta_j C^{\xi_j} \xi_j \), then $\overline F_Z(x) = \overline H_{\xi_j,\sigma_j}(x) =  \left( 1 + x /  (\beta_j C^{\xi_j} ) \right)^{-1/\xi_j}$. Hence, using \eqref{eq:Vasym} for the numerator and \eqref{eq:approx} again for the denominator, we find, as $x \rightarrow \infty$,
		\begin{equation}\label{eq:part2}
		\frac{\overline F_{\beta_j V_j}(x)}{\overline F_{Z}(x)} = 
		\frac{C (x/\beta_j)^{-1/\xi_j} \left(1 - \frac{\beta_j D}{\xi_jx} + o (x^{-1}) \right)}{\left(C^{-\xi_j} (x/\beta_j) \right)^{-1/\xi_j} (1 - \frac{\beta_j C^{\xi_j}}{\xi_j x}+ o(x^{-1}))}
		= 1 - \left( \frac{D - C^{\xi_j}}{(\xi_j/\beta_j) x} \right) + o\left(x^{-1}\right).
		\end{equation}
		Combining \eqref{eq:part1} and \eqref{eq:part2}, we find, as $x \rightarrow \infty$,
		\begin{equation*}
			\frac{\overline F_{Y_j} (x)}{\overline F_Z (x)} = \frac{\overline F_{Y_j} (x)}{\overline F_{\beta_j V_j} (x)} \, \frac{\overline F_{\beta_j V_j} (x)}{\overline F_Z (x)} = \left(1 - \frac{\EE[S_j]}{ (\xi_j/\beta_j)  x} + o(x^{-1}) \right) \left(1 - \left( \frac{D- C^{\xi_j}}{(\xi_j/\beta_j) x} \right) + o(x^{-1}) \right).
		\end{equation*}
       \end{description}
    \end{proof}

	\subsection{Connection with the multivariate GP distribution}\label{sec:coeff_dep_proof}

    \begin{proof}[\normalfont\bfseries Proof of Lemma~\ref{lem:mgpd_limit}]
        For $\beta_j = \sigma_j (\alpha + \alpha_j)$, we have
        \[
        \frac{G + G_j}{\sigma_j (\alpha + \alpha_j)}
        = \frac{G + G_j}{\beta_j}
        \sim \Gamma(\alpha + \alpha_j, \beta_j),
        \]
        with moment-generating function
        \[
        M(t)
        =
        \left(1 - \frac{t}{\sigma_j  (\alpha + \alpha_j)}\right)^{-(\alpha + \alpha_j)},
        \quad
        t < \sigma_j (\alpha + \alpha_j).
        \]
        As $(\alpha + \alpha_j) \to \infty$, $M(t) \to \exp(t/\sigma_j)$, which corresponds to a
        degenerate random variable equal to $1/\sigma_j$. Hence, $\frac{G + G_j}{\beta_j}  \;\xrightarrow{p}\;  1/\sigma_j$ for $j=1,2$.  By the continuous mapping theorem,
        \begin{equation}\label{eq:conv1}
        \beta_j \frac{E}{G + G_j}
        \xrightarrow{p}
        \sigma_j E,
        \qquad j=1,2.
        \end{equation}
        Combining \eqref{eq:conv1} with the assumption that $\beta_j S^{(\beta)}_j\xrightarrow{p} S^0_j$ and using again the continuous mapping theorem, we obtain
        \[
        \beta_j\Big(\frac{E}{G+G_j}-S^{(\beta)}_j\Big)
        \xrightarrow{p}
        \sigma_j E-S^0_j ,
        \qquad j=1,2.
        \]
        Joint convergence in probability follows immediately, yielding the stated limit for $(Y_1,Y_2)$.
    \end{proof}
    
	\subsection{Tail dependence coefficients}\label{sec:coeff_dep_proof2}
	We now derive the residual tail dependence coefficient $\eta$ and the tail dependence coefficient $\chi$. Without loss of generality, we take $\beta_1 = \beta_2 = 1$ for notational simplicity. As before, we write $\vY = \vV - \vS$ for $\vS = (S_1,S_2)^T$ and $\vV = (V_1,V_2)^T$ with
    \begin{equation}\label{eq:def_V_j}
        V_j = \frac{w E + (1-w) E_j}{G + G_j}, \qquad j = 1,2.
    \end{equation}
    In the lemmas below, we start by deriving the joint distribution function of $\vV$ and its asymptotic expansion. We then show that $\eta_{\vY} = \eta_{\vV}$ and $\chi_{\vY} = \chi_{\vV}$ and finally state the proof of Proposition~\ref{prop:tail_dep}.

\begin{lem}[Joint distribution of $\vV$]\label{lem:SF_V}
Let $x_1,x_2\ge 0$ and $w\in(0,1)\setminus\{\tfrac12\}$. Consider the vector $\vV=(V_1,V_2)$ defined in \eqref{eq:def_V_j}. Then the joint survival function admits the decomposition
\[
\PP(V_1>x_1,V_2>x_2)=A_1(x_1,x_2)+A_2(x_1,x_2)+A_{12}(x_1,x_2),
\]
where
\[
A_{12}(x_1,x_2)=\left(1+\frac{x_1+x_2}{1-w}\right)^{-\alpha}\prod_{j=1}^2\left(1+\frac{x_j}{1-w}\right)^{-\alpha_j},
\]
and, for $j=1,2$,
\[
A_j(x_1,x_2)=\frac{w(1-w)^{\alpha_{3-j}}}{1-2w}\left(1+\frac{x_j}{w}\right)^{-(\alpha+\alpha_j)}x_{3-j}^{-\alpha_{3-j}}\,
\EE\big[\mathcal I_j(x_1,x_2)\big],
\]
with
\[
\mathcal I_1(x_1,x_2)=\int_0^{H_1}\left(\exp\left(\frac{1-2w}{w}\min\{\tilde D_1,\,H_1-y\}\right)-1\right)
\frac{y^{\alpha_2-1}}{\Gamma(\alpha_2+1)}\exp\left(-\frac{1-w}{x_2}y\right)\,dy,
\]
where
\[
\tilde D_1=\frac{x_1w}{(1-w)(x_1+w)}(G+G_1),\qquad
H_1=\tilde D_1+E_2-\frac{x_2w}{(1-w)(x_1+w)}G,
\]
and $\mathcal I_2$ defined similarly as $\mathcal I_1$ by swapping indices $1$ and $2$.
\end{lem}

	\begin{lem}[Joint tail distribution of $\vV$]\label{lem:extra}
		We have, for a constant $C > 0$, 
\[
\PP \rbr{ F_{V_1} (V_1) > q, F_{V_2} (V_2) > q}   \sim C\,(1-q)^{\,1+\tfrac{\max(\alpha_1,\alpha_2)}{\alpha+ \max(\alpha_1,\alpha_2)}}, \qquad q \uparrow 1.  
\]
	\end{lem}

\medskip

	\begin{lem}[Shift–invariance]\label{lem:dep_shift}
		We have
		\[
		\chi_{\vY} = \chi_{\vV},
		\qquad
		\eta_{\vY} = \eta_{\vV}.
		\]
	\end{lem}
	\medskip

\begin{proof}[\normalfont\bfseries Proof of Proposition~\ref{prop:tail_dep}]
We compute $\eta_{\vV}$ and $\chi_{\vV}$ which by Lemma~\ref{lem:dep_shift} give $\eta_{\vY}$ and $\chi_{\vY}$. Also, suppose that $w \in [0,1] \setminus \{\tfrac{1}{2}\}$; the case $w = \tfrac12$ can be obtained by applying l'Hôpital's rule. 
\begin{description}
\item[\ref{prop:tail_dep_eta}] Lemma~\ref{lem:extra} shows that the joint tail of $\vV$ is of Ledford--Tawn form $L(1-q) (1-q)^{1/\eta}$ with $L$ slowly varying at $0$ and
\[
\frac{1}{\eta}=1+\frac{\max(\alpha_1,\alpha_2)}{\alpha+\max(\alpha_1,\alpha_2)}=\frac{\alpha+2 \max(\alpha_1,\alpha_2)}{\alpha+\max(\alpha_1,\alpha_2)},
\qquad 
\eta=\frac{\alpha+\max(\alpha_1,\alpha_2)}{\alpha+2 \max(\alpha_1,\alpha_2)}.
\]
\item[\ref{prop:tail_dep_chi}] Suppose that $\alpha_1 = \alpha_2 = 0$, in which case $V_1$ and $V_2$ have the same distribution.
            We can isolate $E$ in the joint survival function of $\vV$ (Lemma~\ref{lem:SF_V}) and apply elementary algebra to obtain
			\begin{equation*}
				\PP(V_j > x) = \EE \left[ \exp\left( -\frac{x}{w} G + \frac{1-w}{w} E_j \right) \mathbbm{1} \{x G > (1-w) E_j \} \right] 
				+ \PP \left( x G < (1-w) E_j \right),
			\end{equation*}
            and
			\begin{align}
				\PP(V_1 > x, V_2 > x) = \EE &\left[ \exp\left( -\frac{x}{w} G + \frac{1-w}{w} \min(E_1, E_2) \right) \mathbbm{1} \{x G > (1-w) \min(E_1, E_2) \} \right] \nonumber \\
				&+ \PP \left( x G < (1-w) \min(E_1, E_2) \right).
			\end{align}
			Since $E_j \sim \text{Exp}(1)$ and $\min(E_1, E_2) \sim \text{Exp}(2)$, both expressions share the same structure. One can then compute:
			\[
			\PP(V_1 > x) = p_1(x), \qquad \PP(V_1 > x, V_2 > x) = p_2(x),
			\]
			with
            \[
            p_s(x) =\frac{s w}{(s+ 1)w - 1 } \left(1 + \frac{x}{w} \right)^{-\alpha}
                        -\frac{1-w}{(s+1)w-1}\left(1+\frac{s\,x}{1-w}\right)^{-\alpha}
            \]
            for $s = 1,2$. This expression holds for $w \neq (s+1)^{-1}$; if not, we take the continuous extension.
            Finally, by the definition of $\chi$, taking the ratio $p_2(x)/p_1(x)$ and letting $x \to \infty$ yields the stated closed-form expression.
\item[\ref{prop:tail_dep_eta_bounds}]
		Assume without loss of generality that $\alpha_2 \ge \alpha_1$, i.e., $\xi_1 \ge \xi_2$. Using $\xi_j = 1/(\alpha + \alpha_j)$ and $\eta = (\alpha + \alpha_2)/(\alpha + 2\alpha_2)$, we express
		\[
		\alpha = \frac{2\eta - 1}{\eta\,\xi_2},
		\qquad
		\alpha_1 = \frac{1}{\xi_1} - \frac{2\eta - 1}{\eta\,\xi_2}.
		\]
		The constraint on the shape parameter $\alpha_1 \ge 0$ leads to the inequality
		\[
		\eta \le \frac{\xi_1}{2\xi_1 - \xi_2}.
		\]
		On the other hand, since $\eta = (\alpha + \alpha_2)/(\alpha + 2\alpha_2) \ge 1/2$ for all $\alpha,\alpha_2 \ge 0$, we obtain
		\[
		\eta \in \left[ \frac{1}{2},\; \frac{\xi_1}{2\xi_1 - \xi_2} \right].
		\]
		Repeating the same reasoning with $\alpha_1 \ge \alpha_2$ (i.e., $\xi_2 \ge \xi_1$) yields the symmetric bound with indices swapped, completing the proof.
	\end{description}
    \end{proof}

\section*{Acknowledgements}

This work was supported by the Fonds de la Recherche 
Scientifique--FNRS under Grant No F.4036.24.

Computational resources have been provided by the supercomputing facilities of the Université catholique de Louvain (CISM/UCL) and the Consortium des Équipements de Calcul Intensif en Fédération Wallonie Bruxelles (CÉCI) funded by the Fonds de la Recherche Scientifique de Belgique (F.R.S.-FNRS) under convention 2.5020.11 and by the Walloon Region.

Part of Naveau’s research work was supported by the Agence Nationale de la Recherche via the SICIM and SHARE PEPR Maths-Vives project (France 2030 ANR-24-EXMA-0008), EXSTA grant (ANR-23-CE40-0009-01), PORC-EPIC, the PEPR TRACCS program (PC4 EXTENDING, ANR-22-EXTR-0005), and the PEPR IRIMONT (France 2030 ANR-22-EXIR-0003). He has also benefited from the Geolearning research chair.

\clearpage
\appendix

\numberwithin{equation}{section}
\renewcommand{\thesection}{A.\arabic{section}}
\renewcommand{\thefigure}{A.\arabic{figure}}
\renewcommand{\thetable}{A.\arabic{table}}
\renewcommand{\theequation}{A.\arabic{equation}}

\section*{Supplementary material}
\addcontentsline{toc}{section}{Supplementary material}

    \section{Additional proofs}

  \begin{proof}[\normalfont\bfseries Proof of Lemma~\ref{lem:marg_V}]
		\text{ }
		\begin{enumerate}
			\item For \(w \in (0,1)\), the numerator \(N = w E + (1-w) E_j\) of $V_j$ is a sum of two independent exponential variables with different scale parameters. Its distribution is hypo-exponential with cdf
			\[
			F_N(u) = 1 - \frac{(1 - w) e^{-u / w} - w e^{-u / (1 - w)}}{2w - 1}.
			\]
			Let $f_{G + G_j}$ denote the density of \(G + G_j \sim \Gamma(1/\xi_j, 1)\). Straightforward calculations give the cdf of \(V_j = N / (G + G_j)\), for $x \geq 0$, 
			\begin{align}\label{eq:FV}
				F_{V_j}(x) = \PP(V_j \le x) & = \int_0^\infty F_N(x t) \cdot f_{G+G_j}(t) \, \diff t \notag \\
				& = 1 - \frac{w(1 + \frac{x}{w})^{-1/\xi_j} - (1-w) (1 + \frac{x}{1 - w})^{-1/\xi_j}}{2w - 1}. 
			\end{align}
            Differentiating \eqref{eq:FV} with respect to \(x\) yields the stated density.  
            The cases \(w \in \{0,1\}\) can be obtained 
            directly from Lemma~\ref{lem:mixture_GPD}. The expression for \(w = \tfrac{1}{2}\) can be obtained by applying l'H\^opital's rule to the above formula.

			\item The mean and variance follow from the facts that, writing $V_j = N M$ for $M = (G + G_j)^{-1}$,
      \[
\begin{aligned}
\EE[N]        &= 1, 
&\qquad \EE[N^2]     &= 1 + w^2 + (1-w)^2, \\
\EE[M]   &= \frac{1}{\xi_j^{-1}-1}, 
&\qquad \EE[M^{2}]  &= \frac{1}{(\xi_j^{-1}-1)(\xi_j^{-1}-2)}.
\end{aligned}
\]

		\end{enumerate}
	\end{proof}
    
  \begin{proof}[\normalfont\bfseries Proof of Lemma~\ref{lem:asymV}]
  
            For $\lambda > 0$ and $\alpha > 0$, we have 
			\begin{equation}\label{eq:approx}
				\frac{(1 + \lambda x)^{-\alpha}}{(\lambda x)^{-\alpha}} = 1 - \alpha (\lambda x)^{-1} + o(x^{-1}), \qquad \text{ as } x \rightarrow \infty.
			\end{equation}
			Expanding $(1 + \lambda x)^{-1/\xi_j}$ in \eqref{eq:FV} according to \eqref{eq:approx} for $\lambda = 1/w$ and $\lambda = 1/(1-w)$ and collecting constants, we get
			\[
			\overline F_{V_j} (x) =  C x^{-1/\xi_j} \left( 1 - \frac{D}{\xi_j x} + o(x^{-1}) \right), \qquad \text{as } x \rightarrow \infty,
			\]
            for $w \in [0,1] \setminus  \{\tfrac 12\}$, where $C$ and $D$ are given in \eqref{eq:CandD}.
            When $w = \tfrac 12$, l'H\^opital's rule gives
            \[
            C = 2^{-1/\xi_j} (\xi_j^{-1} + 1) \quad \text{ and } \quad D = \frac{1 + 2 \xi_j}{2 + 2 \xi_j}.
            \]
            \end{proof}

            \begin{proof}[\normalfont\bfseries Proof of Lemma~\ref{lem:SF_V}]
We write
\[
\PP \left[V_1>x_1,V_2>x_2\right]
=\PP\left[E >\max_{j=1,2}\left\{\frac{x_j(G+G_j)-(1-w)E_j}{w}\right\}\right].
\]
Conditioning on all variables except $E$ and defining
$M_j:=x_j(G+G_j)-(1-w)E_j$, we obtain
\[
\PP(V_1>x_1,V_2>x_2)
=\EE\left[e^{-(M_1\vee M_2)/w}\1\{M_1\vee M_2>0\}\right]
+\PP\{M_1\vee M_2<0\}.
\]
Decomposing on $\{M_1>M_2\}$ and $\{M_2\ge M_1\}$ yields
\[
\PP(V_1>x_1,V_2>x_2)= A_1+A_2+A_{12},
\]
with $A_{12}=\PP\{M_1\vee M_2<0\}$,
\[
A_1=\EE\left[e^{-M_1/w}\1\{M_1>M_2\vee0\}\right],\, \text{ and }
A_2=\EE\left[e^{-M_2/w}\1\{M_2>M_1\vee0\}\right].
\]

\begin{description}
\item[Term $A_{12}$:]
Since $M_1\vee M_2<0$ iff $(1-w)E_j>x_j(G+G_j)$ for $j=1,2$, conditioning on
$(G,G_1,G_2)$ and using independence of exponentials gives
\[
A_{12}
=\EE\left[e^{-\frac{x_1+x_2}{1-w}G}
e^{-\frac{x_1}{1-w}G_1}
e^{-\frac{x_2}{1-w}G_2}\right].
\]
Applying the Laplace transform of Gamma variables yields
\[
A_{12}
=\left(1+\frac{x_1+x_2}{1-w}\right)^{-\alpha}
\left(1+\frac{x_1}{1-w}\right)^{-\alpha_1}
\left(1+\frac{x_2}{1-w}\right)^{-\alpha_2}.
\]
\item[Term $A_1$:]
Using $M_j=x_j(G+G_j)-(1-w)E_j$, we write
\[
A_1=\EE\left[e^{-\frac{x_1}{w}(G+G_1)}
e^{\frac{1-w}{w}E_1}\1\{M_1>M_2,\;M_1>0\}\right].
\]
The constraints are equivalent to $E_1\le B$, where
\[
B=\min\left\{\frac{x_1}{1-w}(G+G_1),\;
\frac{x_1}{1-w}(G+G_1)+E_2-\frac{x_2}{1-w}(G+G_2)\right\}.
\]
Conditioning on all variables except $E_1$ and evaluating the exponential
integral (valid for $w\neq\frac12$) yields
\[
A_1=\frac{w}{1-2w}\,
\EE\left[e^{-\frac{x_1}{w}(G+G_1)}
\Big(e^{\frac{1-2w}{w}B_+}-1\Big)\right],
\qquad B_+=\max(B,0).
\]
Introduce the rescaled variables
$\tilde G=\frac{w}{x_1+w}G$ and $\tilde G_1=\frac{w}{x_1+w}G_1$.
Using scaling invariance of Gamma distributions,
\[
A_1=\frac{w}{1-2w}\left(1+\frac{x_1}{w}\right)^{-(\alpha+\alpha_1)}
\EE\left[e^{\frac{1-2w}{w}\tilde B_+}-1\right],
\]
where $\tilde B$ denotes the corresponding rescaled bound.
Since the integrand vanishes on $\{\tilde B\le0\}$,
\[
\EE\left[e^{\frac{1-2w}{w}\tilde B_+}-1\right]
=\EE\left[\big(e^{\frac{1-2w}{w}\tilde B_+}-1\big)\1\{\tilde B>0\}\right].
\]
Conditioning on $(G,G_1,E_2)$ and integrating out $G_2$ gives
\begin{equation}\label{eq:A1_final}
A_1=\frac{w(1-w)^{\alpha_2}}{1-2w}
\left(1+\frac{x_1}{w}\right)^{-(\alpha+\alpha_1)}
x_2^{-\alpha_2}\,
\EE[\mathcal I_1],
\end{equation}
where $\mathcal I_1$ is the integral term
\[
\mathcal I_1=\int_0^{H_1}\left(\exp\left(\frac{1-2w}{w}\min\{\tilde D_1,\,H_1-y\}\right)-1\right)\frac{y^{\alpha_2-1}}{\Gamma(\alpha_2+1)}\exp\left(-\frac{1-w}{x_2}y\right)\,dy,
\]
with
\[
\tilde D_1=\frac{x_1w}{(1-w)(x_1+w)}(G+G_1),\qquad H_1=\tilde D_1+E_2-\frac{x_2w}{(1-w)(x_1+w)}G.
\]
\item[Term $A_2$:]
The expression for $A_2$ follows by symmetry after exchanging indices
$1$ and $2$.
\end{description}

Combining $A_1$, $A_2$ and $A_{12}$ yields the stated survival function.
\end{proof}

\medskip

            \begin{proof}[\normalfont\bfseries Proof of Lemma~\ref{lem:extra}]
Let $(x_1,x_2)=\big(u_1(q),u_2(q)\big)$ with $q \uparrow 1$. Since $V_j$ has a GP tail (see the proof of Proposition~\ref{prop:marg_Y}~\ref{prop:marg_Y_tail}), it is regularly varying with index $\alpha+\alpha_j$, hence for some constants $C_j>0$,
\[
u_j(q) := F_{V_j}^{-1}(q) \sim C_j\,(1-q) ^{-1/(\alpha+\alpha_j)},\qquad q\uparrow1,
\]
and therefore $u_j(q)^{-\beta}\sim C_j^{-\beta}\, (1-q)^{\beta/(\alpha+\alpha_j)}$ for any $\beta>0$. In particular, assume without loss of generality that $\alpha_1<\alpha_2$, in which case $u_2(q)/u_1(q)\to0$. We analyse successively the three contributions $A_1,A_2,A_{12}$ from
\[
\PP(V_1>x_1,V_2>x_2)=A_1(x_1,x_2)+A_2(x_1,x_2)+A_{12}(x_1,x_2),
\] 
obtained in Lemma~\ref{lem:SF_V}.
\begin{description}
\item[Contribution of $A_{12}$:]
Using $(1+x/a)^{-k}\sim (x/a)^{-k}$ and $x_1 + x_2 \sim x_1$, we obtain
\[
A_{12}(x_1,x_2)\sim (1-w)^{\alpha+\alpha_1+\alpha_2}\,x_1^{-(\alpha+\alpha_1)} x_2^{-\alpha_2}
\sim c_{12}\,(1-q)^{\,1+\frac{\alpha_2}{\alpha+\alpha_2}},
\]
for some constant $c_{12}\in(0,\infty)$.
\item[Contribution of $A_{1}$:]
From \eqref{eq:A1_final} and $(1+x_1/w)^{-(\alpha+\alpha_1)}\sim w^{\alpha+\alpha_1}x_1^{-(\alpha+\alpha_1)}$, we get
\[
A_1 (x_1,x_2)
\sim K_1\,x_1^{-(\alpha+\alpha_1)} x_2^{-\alpha_2}\,\EE \left[\mathcal I_1(x_1,x_2)\right],
\qquad 
K_1:=\frac{w(1-w)^{\alpha_2}}{1-2w}\,w^{\alpha+\alpha_1}.
\]
In the regime $x_1,x_2\to\infty$ with $x_2/x_1\to0$, $\mathcal I_1(x_1,x_2)\to \mathcal I_{1,\infty}$ a.s. and is dominated by an integrable envelope (Gamma moments), so dominated convergence yields $\EE[\mathcal I_1(x_1,x_2)]\to\kappa_1\in(0,\infty)$. Consequently,
\[
A_1 (x_1,x_2)\sim c_1\,x_1^{-(\alpha+\alpha_1)} x_2^{-\alpha_2}
\sim c_1\,(1-q)^{\,1+\frac{\alpha_2}{\alpha+\alpha_2}},
\]
for some $c_1\in(0,\infty)$.
\item[Contribution of $A_{2}$:]
By symmetry (swap indices $1$ and $2$ in \eqref{eq:A1_final}),
\[
A_2(x_1,x_2)=\frac{w(1-w)^{\alpha_1}}{1-2w}\left(1+\frac{x_2}{w}\right)^{-(\alpha+\alpha_2)}x_1^{-\alpha_1}\,\EE[\mathcal I_2(x_1,x_2)].
\]
In the regime $x_1,x_2\to\infty$ with $x_2/x_1\to0$, the constraint defining $\mathcal I_2$ forces $G=O(x_2/x_1)$; with the scaling $G=(x_2/x_1)U$ and the small-argument expansion of the Gamma density, one obtains
\[
\EE[\mathcal I_2(x_1,x_2)]\sim \kappa_2\left(\frac{x_2}{x_1}\right)^{\alpha},
\qquad \kappa_2\in(0,\infty),
\]
and therefore
\[
A_2 (x_1,x_2)\sim c_2\,x_1^{-(\alpha+\alpha_2)} x_2^{-\alpha_1}
\sim c_2\,(1-q)^{\,1+\frac{\alpha_2}{\alpha+\alpha_2}},
\]
for some $c_2\in(0,\infty)$.
\end{description}
If $\alpha_1 > \alpha_2$, the same expansions hold with the roles of $\alpha_1$ and $\alpha_2$ reversed. Hence, we find
\begin{align}\label{eq:final}
\bar{F}_{V_1,V_2} \left(u_1(q), u_2(q) \right) \sim (c_1+c_2+c_{12})\,(1-q)^{\,1+\frac{\max(\alpha_1,\alpha_2)}{\alpha+\max(\alpha_1,\alpha_2)}}.  
\end{align}
            \end{proof}

  \begin{proof}[\normalfont\bfseries Proof of Lemma~\ref{lem:dep_shift}]

Since $V_j$ is regularly varying and $S_j \geq 0$, by dominated convergence, we have 
$ \bar F_{Y_j}(x) = \bar F_{V_j}(x) (1 + o(1))$ as $x\to\infty$.
By \citet[Proposition 2.6(vi)]{resnick_heavy-tail_2007}, this yields
$w_j(q):= F^{-1}_{Y_j}(q)= F^{-1}_{V_j}(q)(1 + o(1))$ as  $q\uparrow 1$.
From the proof of Lemma~\ref{lem:extra}, recall that $u_j(q) := F_{V_j}^{-1}(q) \sim C_j (1-q)^{-1/(\alpha + \alpha_j)}$. In particular, the proof of Lemma~\ref{lem:extra} gives
\begin{equation}\label{eq:extra}
\bar{F}_{V_1,V_2} \rbr{ u_1(q) \{1 + o(1)\}, u_2 (q) \{1 + o(1)\}} \sim \bar{F}_{V_1,V_2} \rbr{ u_1(q), u_2 (q)} 
\end{equation}

Next, note that
\[
\bar F_{Y_1,Y_2}(w_1(q),w_2(q))
= \EE \left[\bar F_{V_1,V_2}\big(w_1(q)+S_1,\;w_2(q)+S_2\big)\right].
\]
Define
\[
R(q):=\frac{\bar F_{V_1,V_2}(w_1(q)+S_1,\;w_2(q)+S_2)}
{\bar F_{V_1,V_2}(w_1(q),w_2(q))}.
\]
Conditionally on $(S_1,S_2)$, the shifts become deterministic constants. Since $w_j(q)\to\infty$, we have $S_j=o(w_j(q))$ a.s., so that equation~\ref{eq:extra} yields $R(q)\to1$ almost surely.
Moreover, since $0\le R(q)\le 1$, we have $\EE[R(q)]\to 1$ by bounded convergence and
\[
\bar F_{Y_1,Y_2}(w_1(q),w_2(q))
\sim
\bar F_{V_1,V_2}(w_1(q),w_2(q)),
\qquad q\uparrow 1.
\]
Finally, using again equation~\ref{eq:extra},
\[
\bar F_{Y_1,Y_2} \big(F^{-1}_{Y_1}(q),F^{-1}_{Y_2}(q)\big)
\sim
\bar F_{V_1,V_2} \left(F^{-1}_{V_1}(q),F^{-1}_{V_2}(q)\right),
\qquad q\uparrow 1.
\]
so that dividing by $(1-q)$ and letting $q \uparrow 1$ gives $\chi_{\vY} =\chi_{\vV}$.
If $\chi_{\vV}=0$, then taking logarithms gives also $\eta_{\vY}=\eta_{\vV}$.
\end{proof}

\section{Simulation Study: additional results}\label{app:sim_study}

This appendix provides additional simulation results that complement the analysis in Section~\ref{sec:simstudy}. We report results under randomized parameter configurations, as well as results for the penalized estimator in the fixed-parameter setting.

\subsection{Randomized parameters} \label{subsec:sim_random}    
We assess the global performance of the NBE under randomly generated parameter configurations. For each of the $K = 10^3$ replications, we draw the parameters
\[
(\eta, \xi_1, \xi_2, \beta_1, \beta_2, \sigma_T, w)
\]
independently from their prior distributions specified in Section~\ref{sec:estimation}. The values of $\eta$, $\xi_1$ and $\xi_2$ are then transformed into the model parameters $(\alpha, \alpha_1, \alpha_2)$ according to the construction described in Section~\ref{sec:bivariate}, and a dataset of size $N = 1000$ is simulated from the resulting model. Each dataset is then provided to the NBE for estimation.

Figure~\ref{fig:simstudy_params_and_chi_combined} compares true and estimated values for all parameters. The points cluster tightly around the diagonal, indicating accurate recovery of both the model parameters and the extremal dependence structure. Estimates of $\eta$ exhibit higher variability for lower values, and stabilize as $\eta$ increases. This pattern is expected since weaker dependence (smaller $\eta$) corresponds to lighter joint tails, providing less information for reliable estimation. In contrast, when $\eta$ approaches~$1$, the stronger dependence structure yields more stable and precise estimates across replications.

    \begin{figure}[htbp]
        \centering
        \includegraphics[width=\textwidth]{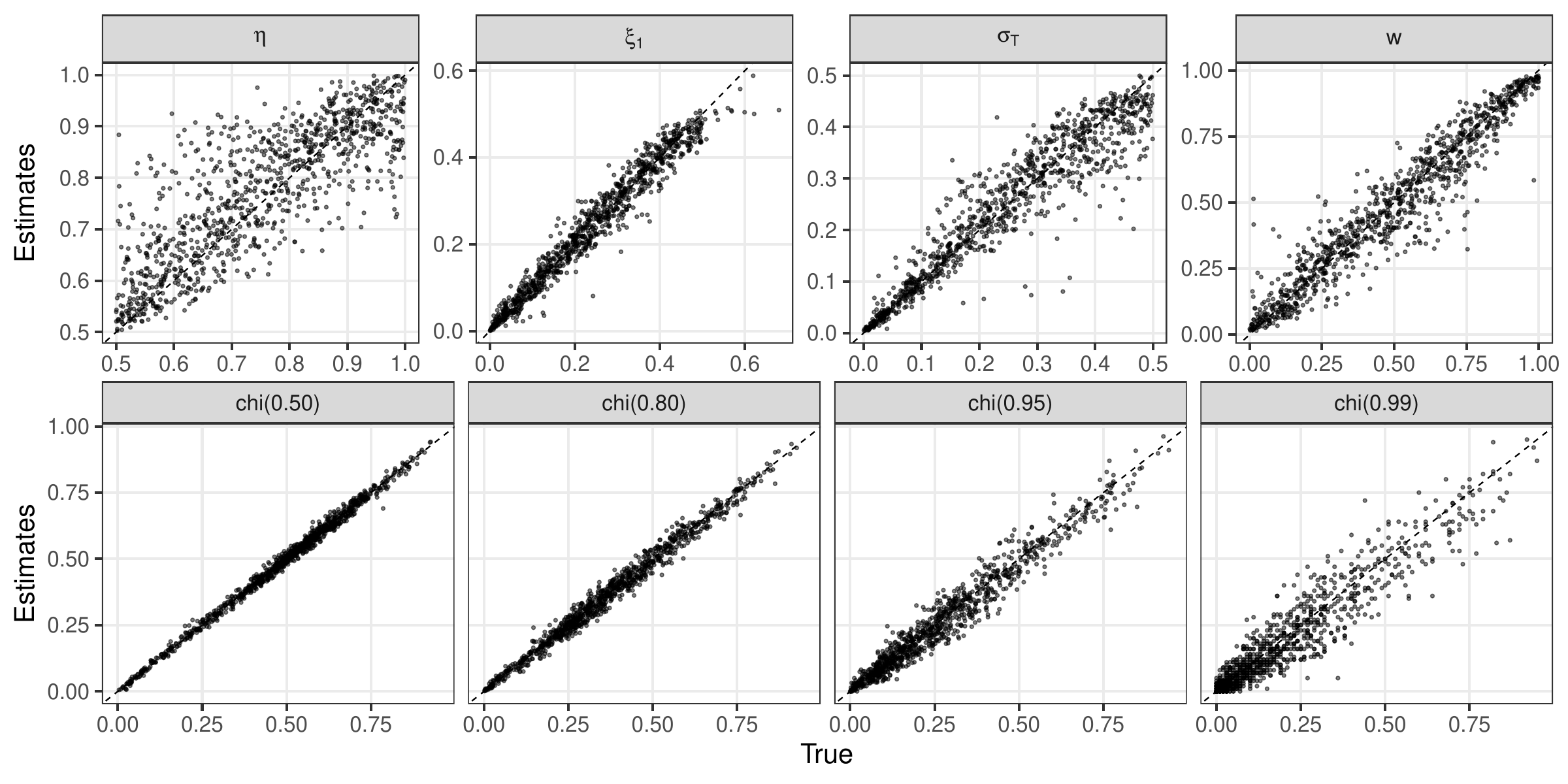}
         \includegraphics[width=0.8\textwidth]{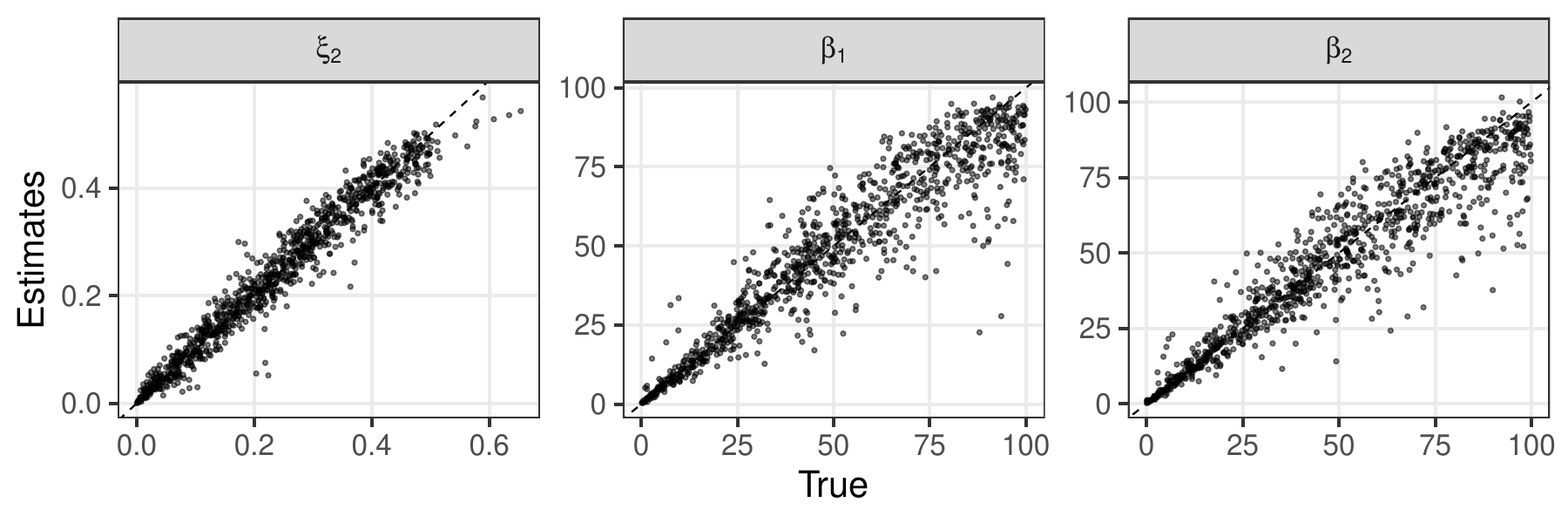}
        \caption{True vs estimated values for $\eta$, $\xi_1$, $\sigma_T$, $w$ (top), $\chi(q)$ with $q \in \{0.5, 0.8, 0.95, 0.99\}$ (middle), and $\xi_2, \beta_1$, $\beta_2$ (bottom) in the randomized parameter setting with $n = 1000$. Each point corresponds to one of the $K = 1000$ Monte Carlo replications.}
        \label{fig:simstudy_params_and_chi_combined}
    \end{figure}

\subsection{Simulation study: Penalised NBE (fixed-parameter setting)} \label{app:sim_pen_fixed}

In Figures~\ref{fig:pen_fixed_params_boxplots_full} and~\ref{fig:pen_fixed_params_chi_full} and Table~\ref{tab:ci_summary_pen}, we report additional results for the fixed-parameter simulation setting using the penalized NBE. The experimental design is identical to that described in Section~\ref{sec:simstudy}, with the only difference being the use of the penalized loss function.

\begin{figure}[htbp]
	\centering
	\includegraphics[width=\textwidth]{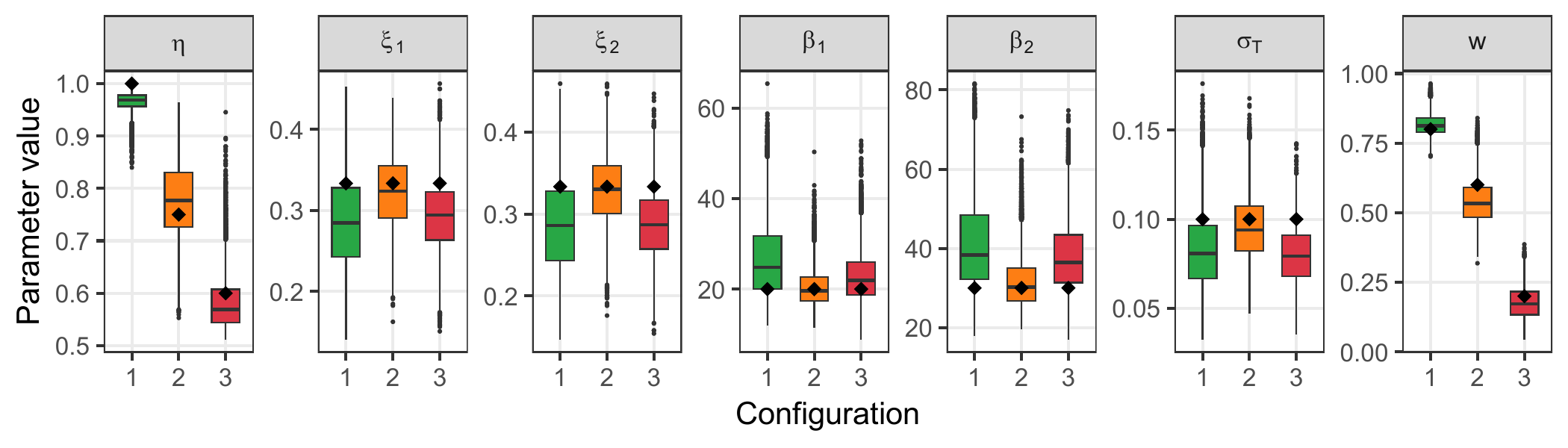}
    \caption{Boxplots of parameter estimates ($\eta, \xi_{1}, \xi_{2}, \beta_{1}, \beta_{2}, \sigma_{T}, w$) obtained with the penalised NBE under three dependence configurations: $\vthe^{(1)}$, $\vthe^{(2)}$, and $\vthe^{(3)}$ (see~\eqref{eq:true_theta_param}), based on $K = 1000$ samples of size $N = 1000$. Diamond markers indicate true values.}
	\label{fig:pen_fixed_params_boxplots_full}
\end{figure}

\begin{figure}[htbp]
 \centering 
 \includegraphics[width=0.98\textwidth]{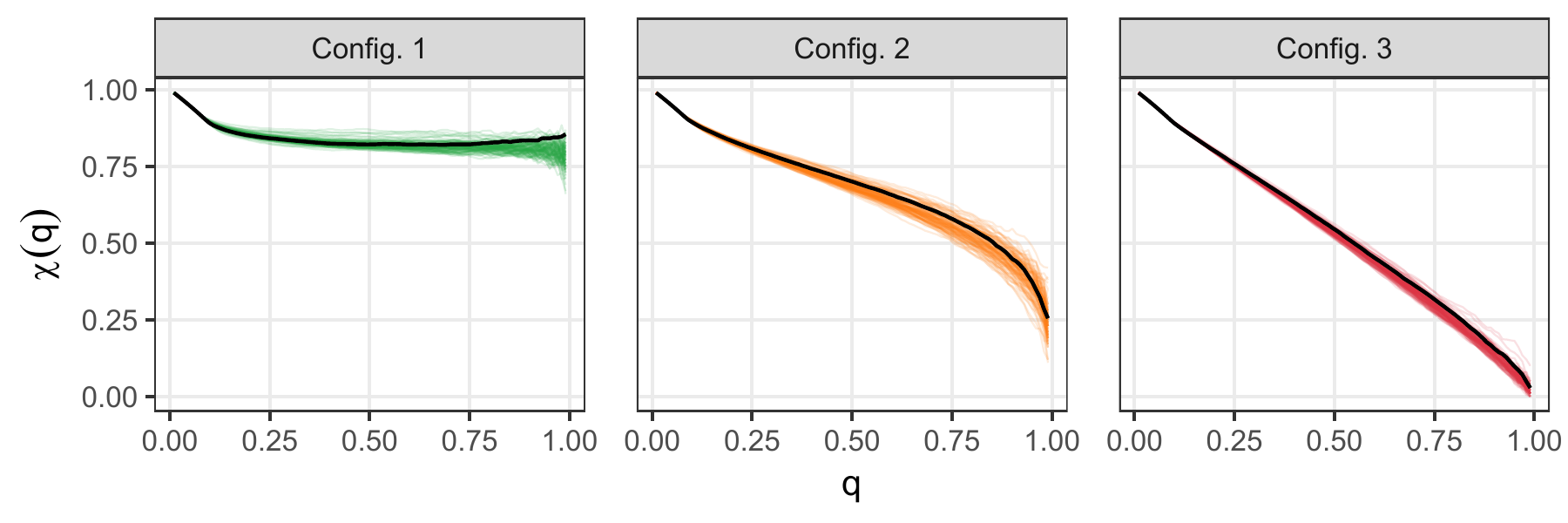}
 \caption{Estimated $\chi(q)$ curves obtained with the penalised NBE, where $\chi(q)$ is estimated using~\eqref{eq:chi_esti}. Solid lines correspond to reference curves empirically estimated from samples of size $10^5$ generated under the true parameters. Coloured lines show 100 estimated curves based on samples of size $10^4$ generated under fitted parameter values.}
 \label{fig:pen_fixed_params_chi_full}
\end{figure}

\begin{table}[htbp]
\centering
\begin{tabular}{lcccccccccc}
\toprule
 & \multicolumn{3}{c}{Config. 1} & \multicolumn{3}{c}{Config. 2} & \multicolumn{3}{c}{Config. 3} \\
\cmidrule(lr){2-4} \cmidrule(lr){5-7} \cmidrule(lr){8-10}
$\hat{\theta}_j$ & True & Cov. & Width & True & Cov. & Width & True & Cov. & Width \\
\midrule
$\hat{\eta}$        & 1.00 & 0.00 & 0.05 & 0.75 & 0.90 & 0.20 & 0.60 & 0.86 & 0.14 \\
$\hat{\xi}_{1}$     & 0.33 & 0.73 & 0.15 & 1.50 & 0.89 & 0.11 & 0.33 & 0.68 & 0.12 \\
$\hat{\xi}_{2}$     & 0.33 & 0.72 & 0.15 & 1.50 & 0.90 & 0.11 & 0.33 & 0.69 & 0.12 \\
$\hat{\beta}_{1}$   & 20.00 & 0.79 & 20.9 & 20.00 & 0.96 & 10.8 & 20.00 & 0.87 & 15.2 \\
$\hat{\beta}_{2}$   & 30.00 & 0.63 & 28.8 & 30.00 & 0.94 & 17.1 & 30.00 & 0.76 & 23.5 \\
$\hat{\sigma}_{T}$  & 0.10 & 0.78 & 0.06 & 0.10 & 0.89 & 0.04 & 0.10 & 0.59 & 0.04 \\
$\hat{w}$           & 0.80 & 0.86 & 0.10 & 0.60 & 0.78 & 0.21 & 0.20 & 0.92 & 0.15 \\
\midrule
$\chi(0.50)$ & 0.82 & 1.00 & 0.05 & 0.70 & 1.00 & 0.05 & 0.55 & 1.00 & 0.04 \\
$\chi(0.70)$ & 0.82 & 1.00 & 0.06 & 0.60 & 1.00 & 0.08 & 0.37 & 1.00 & 0.06 \\
$\chi(0.80)$ & 0.83 & 1.00 & 0.06 & 0.54 & 1.00 & 0.11 & 0.27 & 1.00 & 0.06 \\
$\chi(0.90)$ & 0.84 & 1.00 & 0.07 & 0.45 & 1.00 & 0.13 & 0.16 & 1.00 & 0.07 \\
$\chi(0.95)$ & 0.84 & 0.77 & 0.09 & 0.37 & 1.00 & 0.16 & 0.10 & 1.00 & 0.07 \\
$\chi(0.99)$ & 0.85 & 0.51 & 0.16 & 0.23 & 1.00 & 0.22 & 0.03 & 1.00 & 0.06 \\
\bottomrule
\end{tabular}
\caption{Empirical coverage probability (Cov.) and mean confidence interval width (Width) for model parameters and tail dependence curves $\chi(q)$ obtained with the penalised NBE under three fixed configurations defined in~\eqref{eq:true_theta_param} ($N=1000$, $K=1000$, $B=200$ bootstrap resamples per estimate), together with the true values.}
\label{tab:ci_summary_pen}
\end{table}

Overall, the penalised NBE introduces small biases, leading to a moderate reduction in coverage for several parameters, in particular $(\eta, \xi_1, \xi_2)$ and $\sigma_T$. This effect is more pronounced in the weakest dependence setting (Config.~3). 
In contrast, the estimation of $\chi(q)$ remains stable across most quantiles, with coverage close to the nominal level, except at the most extreme levels (e.g., $q=0.99$ in Config.~1). 
These results highlight a trade-off induced by penalisation: improved structural coherence of the estimates at the cost of a modest loss in estimation quality due to the introduction of small biases.

\section{Application to Belgian Weekly Rainfall: additional results} \label{app:be_rainfall_pen}

\subsection{Parameter stability and QQ-plots}

We present some supplementary figures for the Belgian rainfall application. 

Figure~\ref{fig:stations_squares_map} shows the spatial domain used in the analysis, corresponding to all ERA5 grid points in the rectangle $[0^\circ,7^\circ]$E $\times$ $[47^\circ,52^\circ]$N. 
The central reference grid point covering Brussels is marked, along with neighbouring points. 
Coordinates correspond to the centres of the ERA5 grid cells.

\begin{figure}[H]
	\centering
	\includegraphics[width=0.6\textwidth]{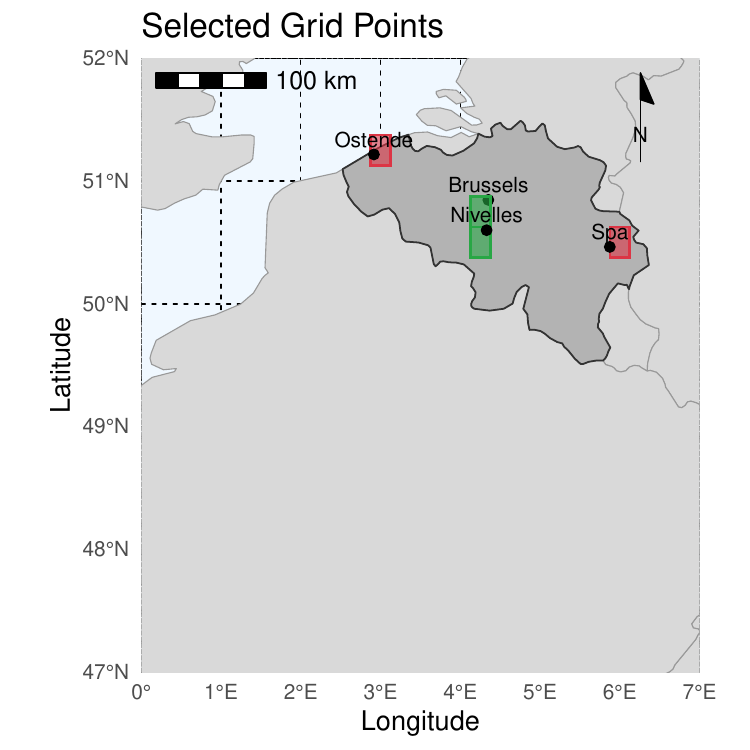}
	\caption{Reference grid point (black point, Brussels) and surrounding ERA5 grid points within the analysis domain $[0^\circ,7^\circ]$E $\times$ $[47^\circ,52^\circ]$N. Coordinates indicate grid cell centres.}
	\label{fig:stations_squares_map}
\end{figure}

Next, Figure~\ref{fig:pen_params_vs_threshold_Full_faceted} reports parameter estimates as functions of the marginal threshold, for all pairs involving the Brussels reference grid point. 
These plots provide a visual check of threshold stability for both marginal and dependence parameters.

\begin{figure}[H]
	\centering
	\includegraphics[width=0.95\textwidth]{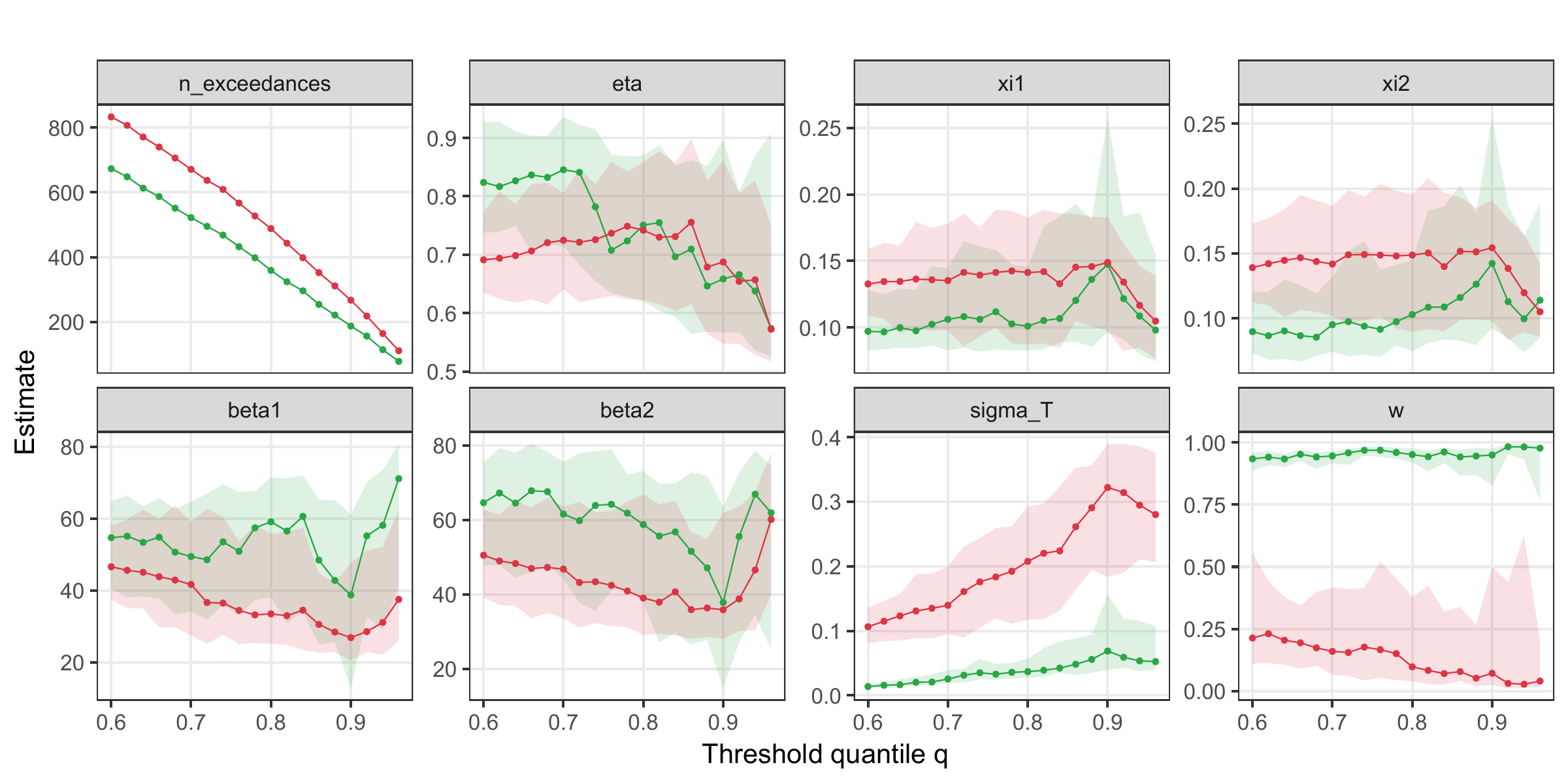}
    \caption{Estimated parameter values of the sBGP model (Definition~\ref{def:biv_model}) as a function of the threshold quantile $q$ for Brussels–Nivelles and Ostend–Spa. Green corresponds to the Brussels–Nivelles pair and red to the Ostend–Spa pair. Shaded areas denote the $95\%$ bootstrap confidence intervals.}
	\label{fig:pen_params_vs_threshold_Full_faceted}
\end{figure}

Finally, Figure~\ref{fig:spatial_pattern_param_pen2} complements Figure~\ref{fig:spatial_pattern_param} in the main paper, illustrating several additional parameters as spatial fields.

\begin{figure}[H]
	\centering
		\includegraphics[width=0.24\textwidth, trim=50 0 20 0, clip]{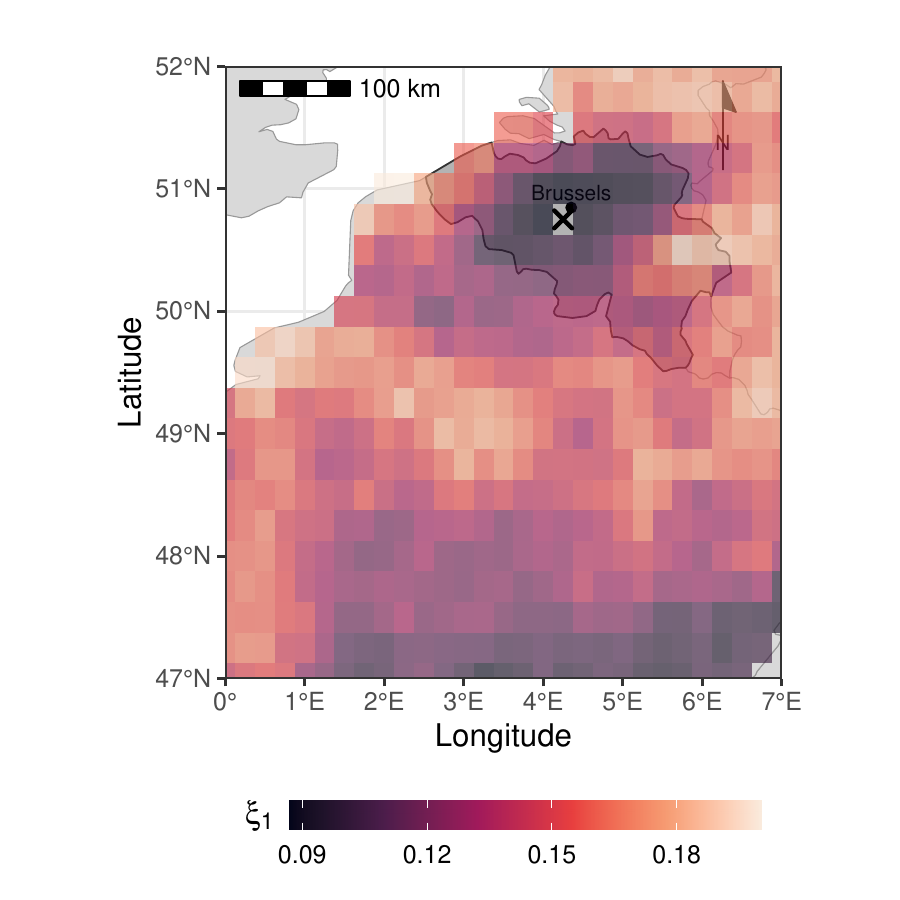}
		\includegraphics[width=0.24\textwidth, trim=50 0 20 0, clip]{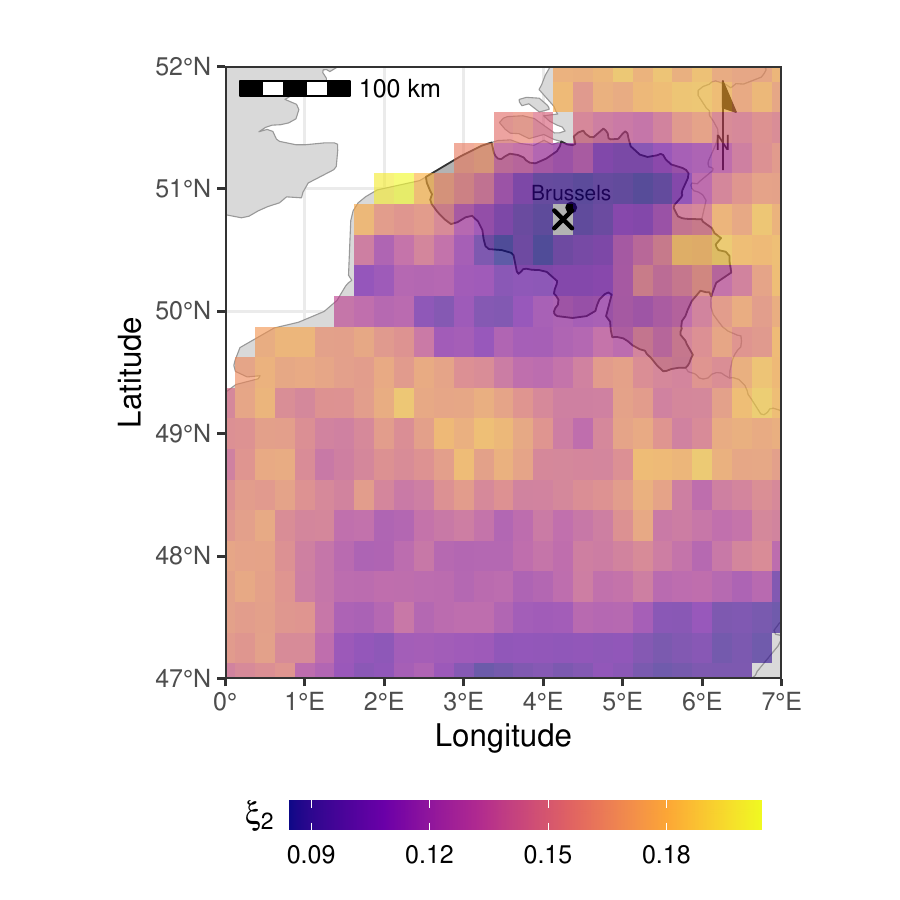}
		\includegraphics[width=0.24\textwidth, trim=50 0 20 0, clip]{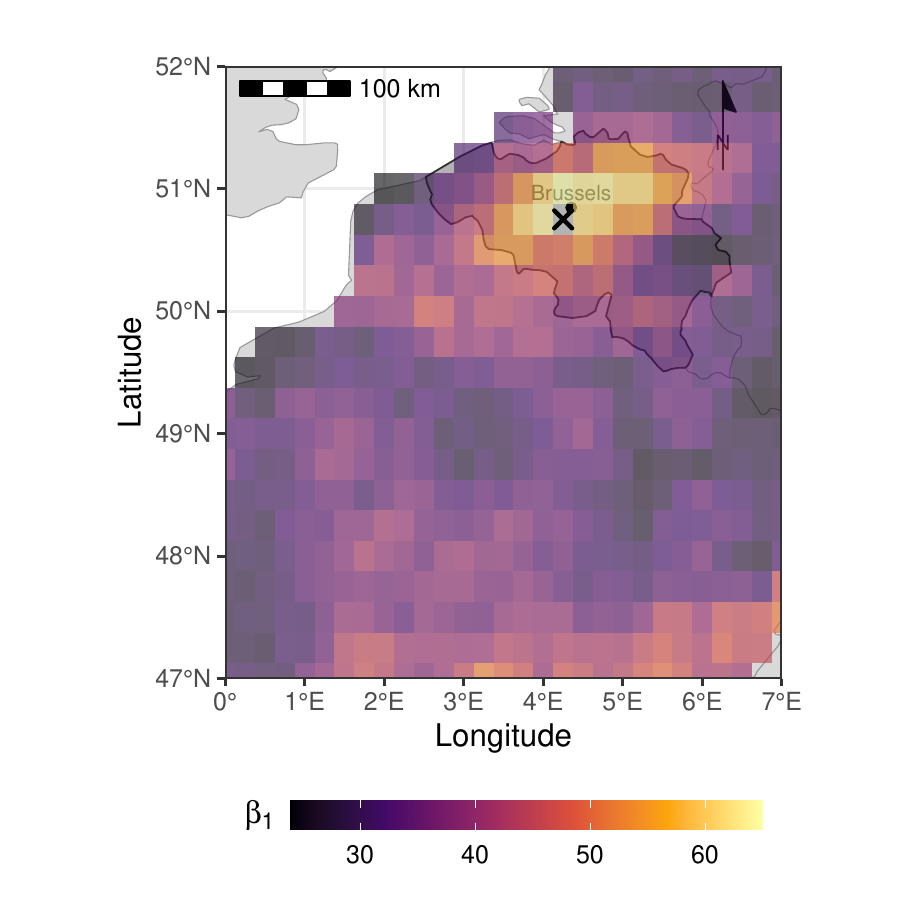}
		\includegraphics[width=0.24\textwidth, trim=50 0 20 0, clip]{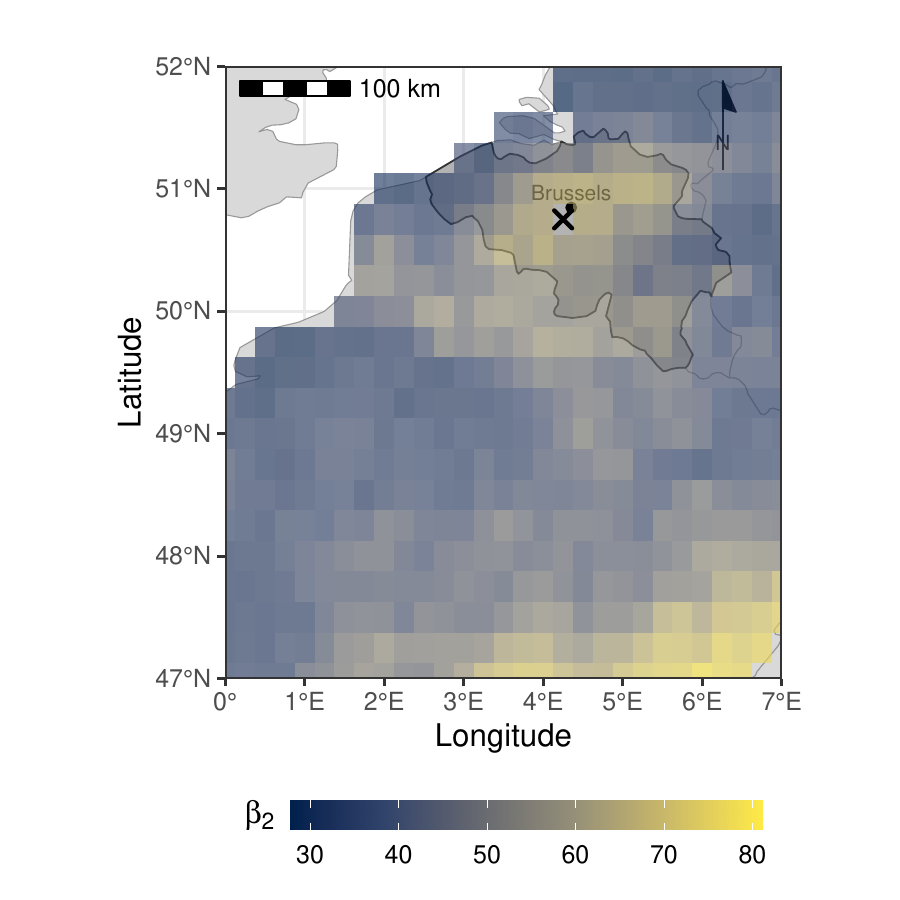}
	\caption{Spatial fields of estimated parameters of the sBGP model (Definition~\ref{def:biv_model}) for bivariate precipitation extremes between Brussels (reference location, black cross) and each grid point.}
	\label{fig:spatial_pattern_param_pen2}
\end{figure}

\subsection{Non-Penalised NBE} \label{app:be_rainfall_non_pen}
Figure~\ref{fig:spatial_pattern_param_non_pen} shows the spatial fields of estimated parameters for the Belgian rainfall application obtained using the NBE trained without penalisation. These results are provided for comparison with the penalised version discussed in the main text, and illustrate the impact of the loss specification on both parameter estimates and goodness-of-fit diagnostics.


\begin{figure}[H]
	\centering
	\begin{minipage}{0.32\textwidth}
		\centering
		\includegraphics[width=\linewidth]{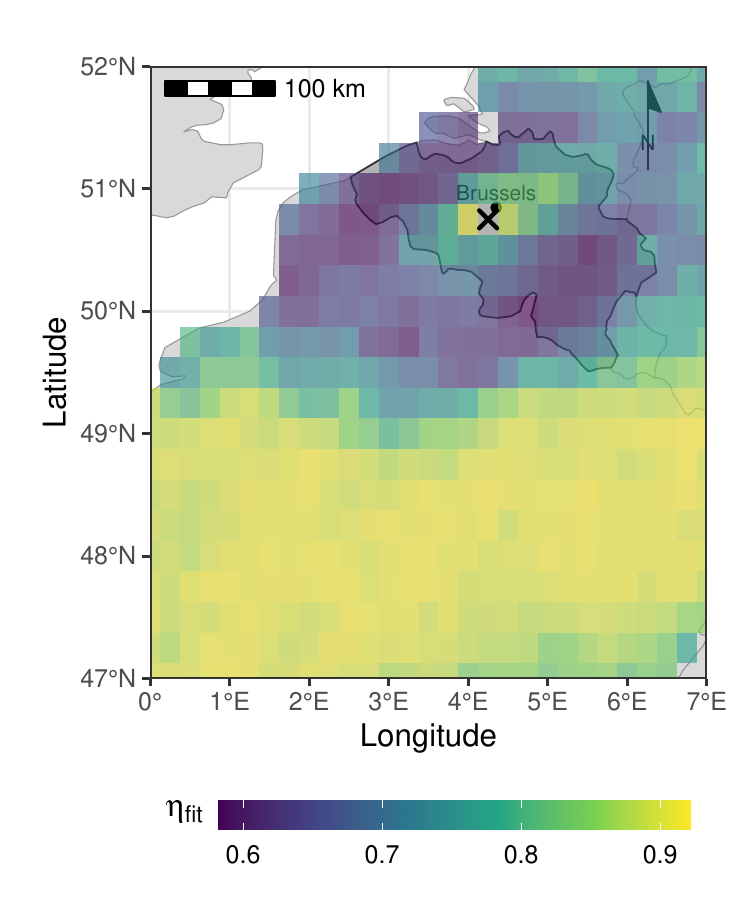}
	\end{minipage}\hfill
	\begin{minipage}{0.32\textwidth}
		\centering
		\includegraphics[width=\linewidth]{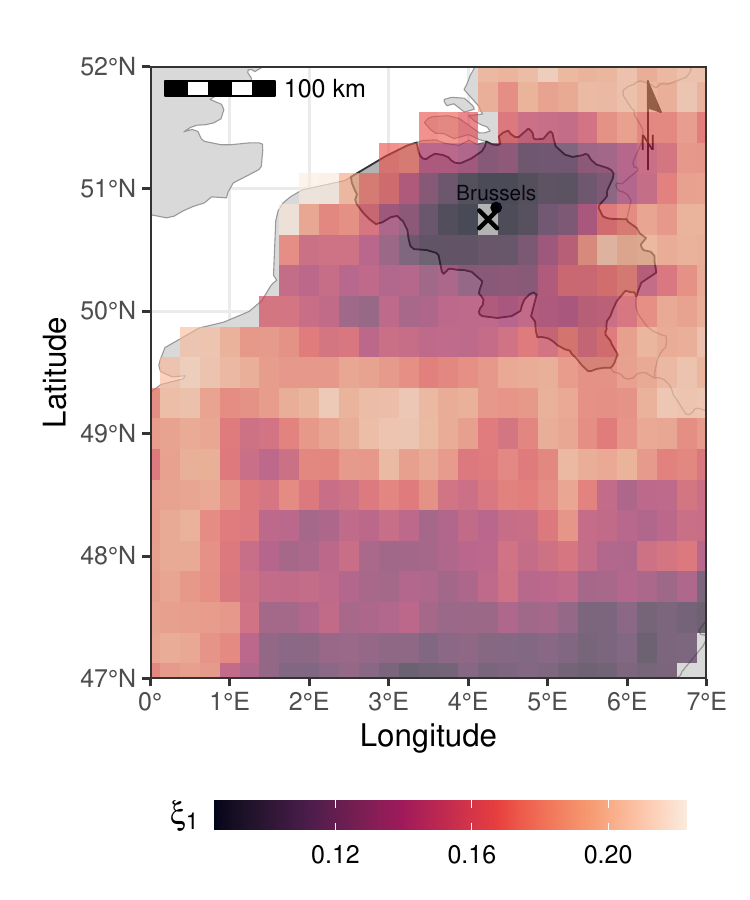}
	\end{minipage}\hfill
	\begin{minipage}{0.32\textwidth}
		\centering
		\includegraphics[width=\linewidth]{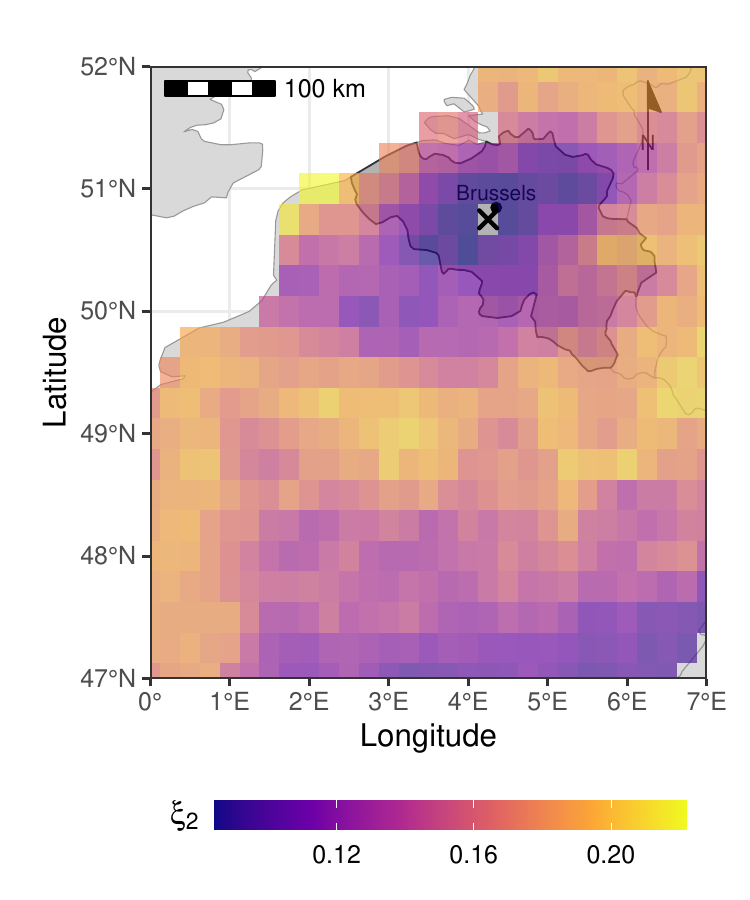}
	\end{minipage}
	
	\vspace{0.3cm}
	
	\begin{minipage}{0.32\textwidth}
		\centering
		\includegraphics[width=\linewidth]{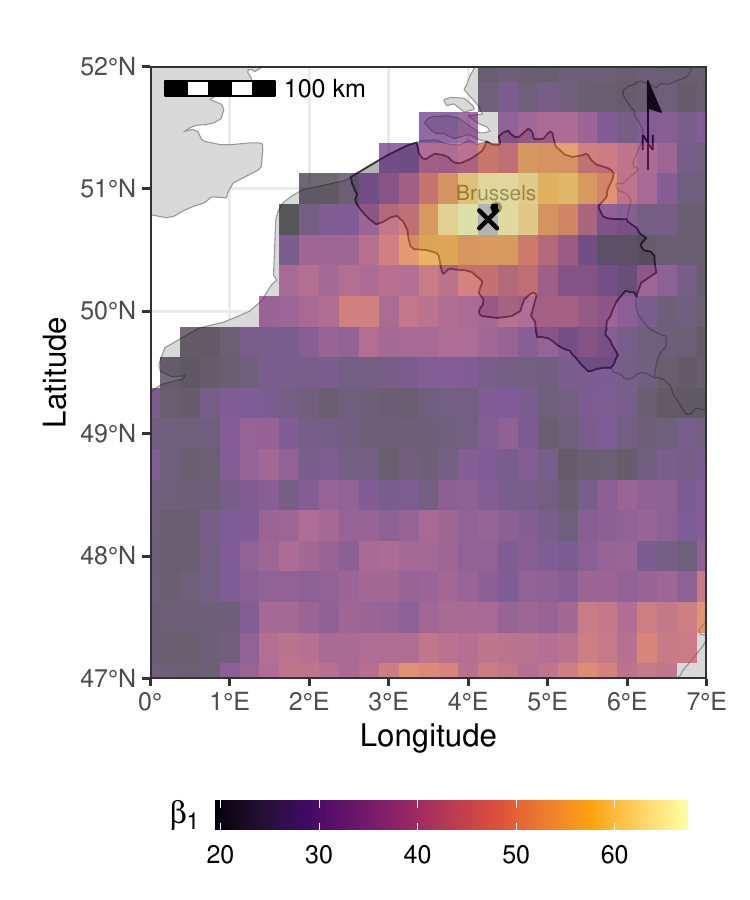}
	\end{minipage}\hfill
	\begin{minipage}{0.32\textwidth}
		\centering
		\includegraphics[width=\linewidth]{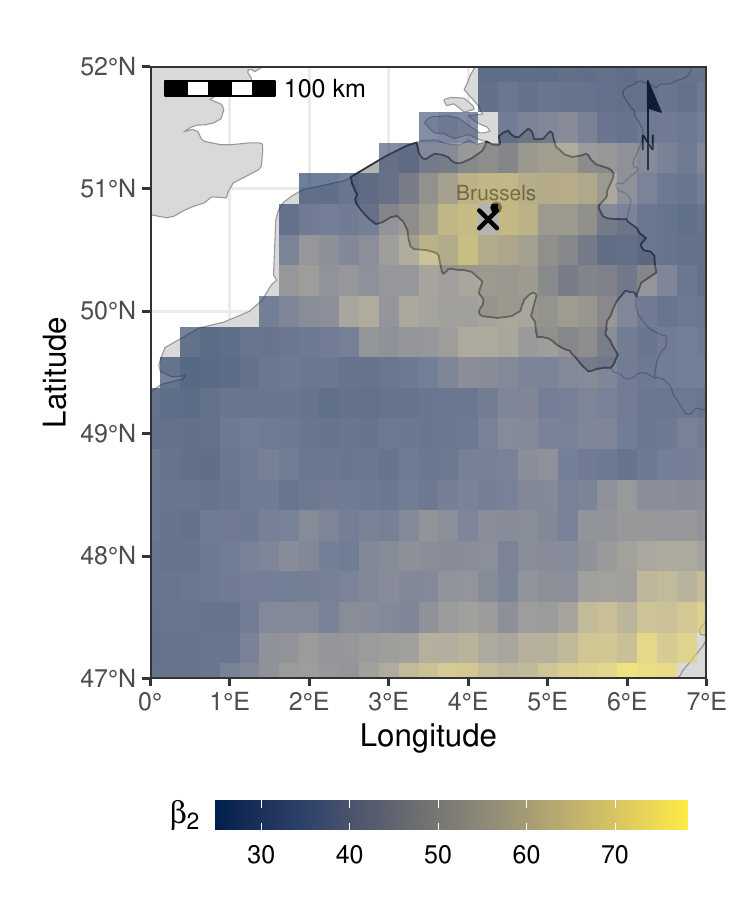}
	\end{minipage}\hfill
	\begin{minipage}{0.32\textwidth}
		\centering
		\includegraphics[width=\linewidth]{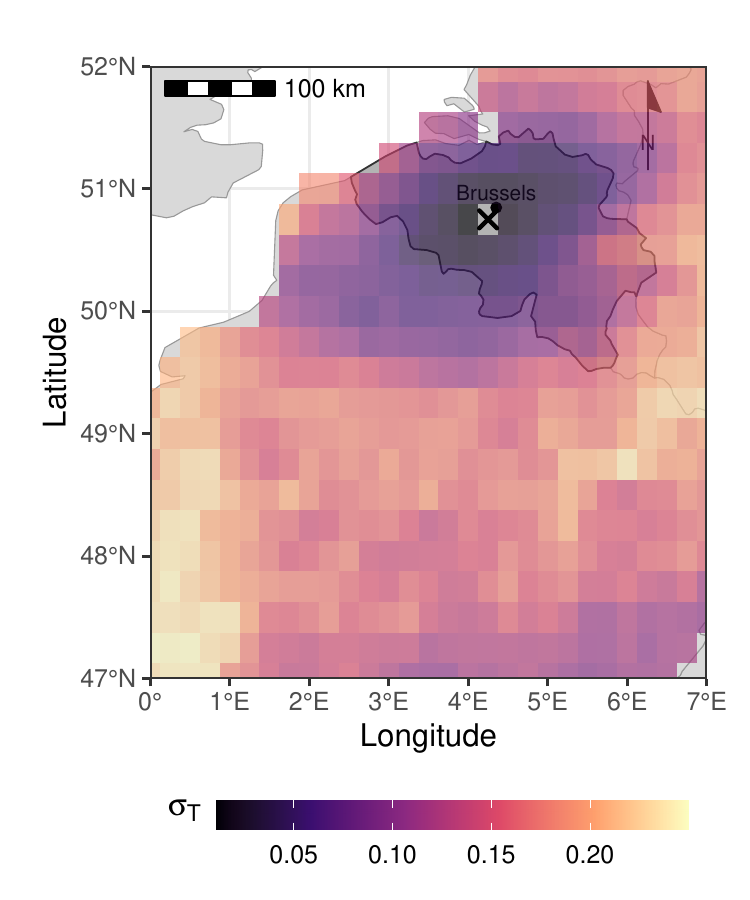}
	\end{minipage}
	
	\vspace{0.3cm}
	
	\begin{minipage}{0.32\textwidth}
		\centering
		\includegraphics[width=\linewidth]{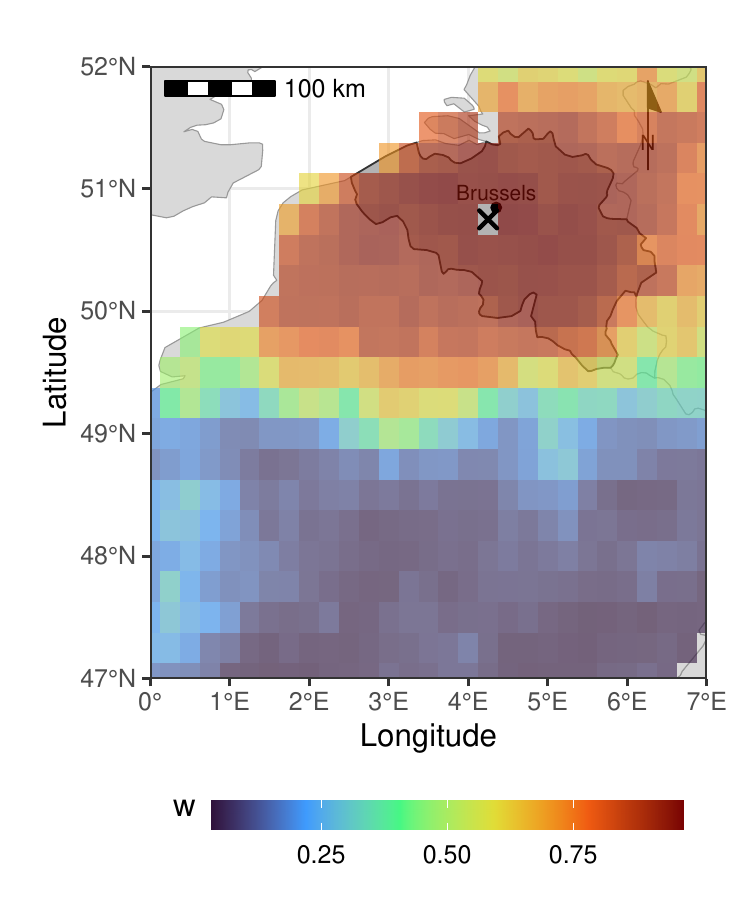}
	\end{minipage}
	
	\caption{Spatial fields of estimated dependence parameters of the sBGP model (Definition~\ref{def:biv_model}) for bivariate precipitation extremes between Brussels (reference location, black cross) and each grid point, obtained using the non-penalised NBE.}
	\label{fig:spatial_pattern_param_non_pen}
\end{figure}

\begin{table}[htbp]
\centering
\begin{tabular}{lcc}
\toprule
 & Brussels--Nivelles & Ostend--Spa \\
\midrule
$\eta$      & 0.77 (0.67, 0.85) & 0.92 (0.84, 0.94) \\
$\xi_1$     & 0.09 (0.06, 0.14) & 0.08 (0.07, 0.11) \\
$\xi_2$     & 0.08 (0.06, 0.13) & 0.08 (0.07, 0.11) \\
$\beta_1$   & 47.11 (28.19, 65.55) & 66.45 (49.17, 74.91) \\
$\beta_2$   & 60.46 (35.26, 74.40) & 71.79 (52.62, 76.53) \\
$\sigma_T$  & 0.02 (0.02, 0.04) & 0.07 (0.06, 0.10) \\
$w$         & 0.94 (0.89, 0.97) & 0.10 (0.04, 0.29) \\
\bottomrule
\end{tabular}
\caption{Parameter estimates (with 95\% confidence intervals) for the sBGP model (Definition~\ref{def:biv_model}) fitted to the two representative pairs at the 70th-percentile threshold using the non-penalised NBE.}
\label{tab:rep_pairs}
\end{table}

\begin{figure}[htbp]
  \centering
  \includegraphics[width=0.95\textwidth]{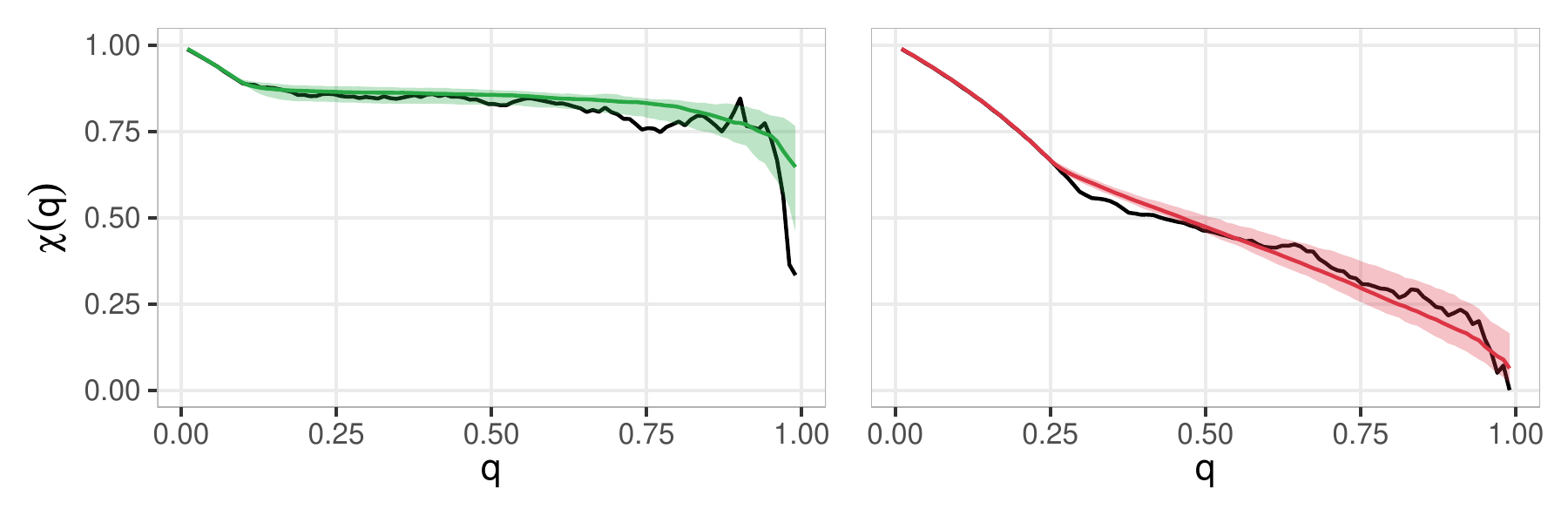}
  \caption{Estimated $\chi_{\tilde{\mathcal{Z}}_{0.7}}(q)$ curves for Brussels–Nivelles and Ostend–Spa using the non-penalised NBE, where $\chi(q)$ is estimated using~\eqref{eq:chi_esti}. Quantiles $q$ are defined relative to the exceedance subset $\tilde{\mathcal{Z}}_{0.7}$ (i.e., within the exceedance region, $u_j = F_j^{-1}(0.7)$). The black line shows the empirical $\chi_{\tilde{\mathcal{Z}}_{0.7}}(q)$ from observed exceedances; the green and red lines are the model estimates; the shaded area denotes the 95\% bootstrap confidence region.}
  \label{fig:p_chi}
\end{figure}

\begin{figure}[htbp]
  \centering
  \includegraphics[width=\textwidth]{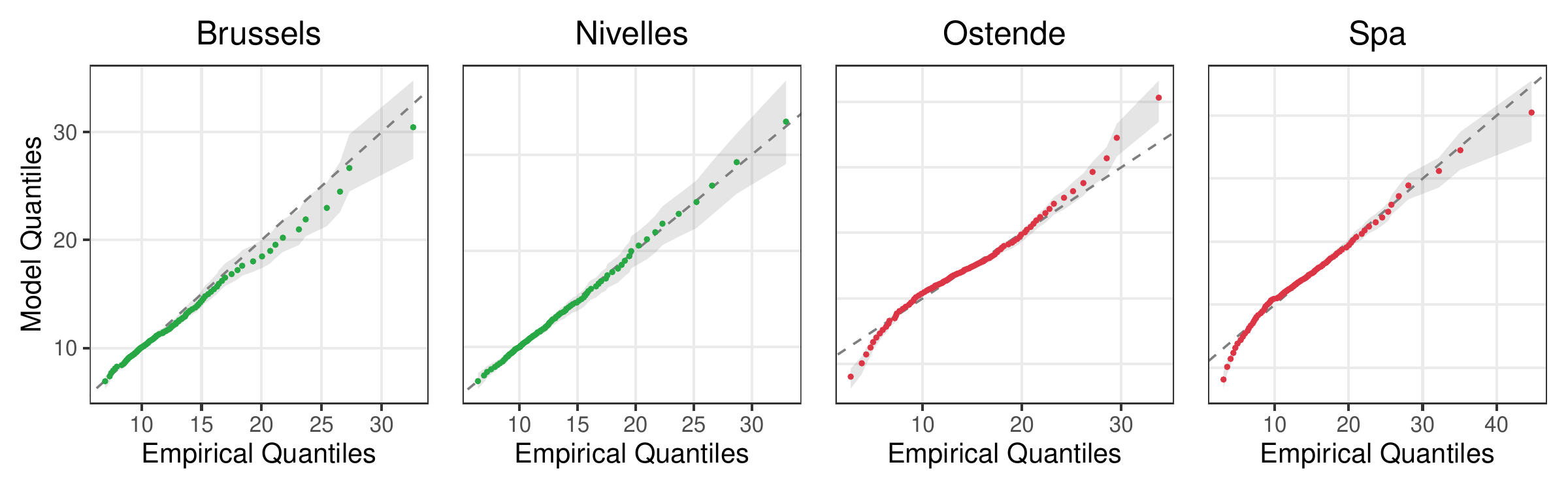}
  \caption{QQ-plots comparing empirical marginal exceedances to fitted marginal distributions for Brussels–Nivelles and Ostend–Spa under the sBGP model (Definition~\ref{def:biv_model}) using the non-penalised NBE. Each panel corresponds to one margin of the respective pair. All plots are based on exceedances above the marginal threshold ($u_j = F_j^{-1}(0.7)$), with dashed vertical lines marking the threshold values. Theoretical quantiles are shown on the original data scale, obtained by adding the threshold to the modeled exceedances.}
  \label{fig:qq_marginals_main}
\end{figure}

\renewcommand\refname{REFERENCES} 
\bibliographystyle{chicago} 
\bibliography{libAsympIndGPD}

\end{document}